%% file: main.tex
\documentclass[a4paper,11pt]{article}
\pdfoutput=1

\usepackage[utf8]{inputenc}
\usepackage[british]{babel}
\usepackage[mono=false]{libertine}
\usepackage[T1]{fontenc}
\usepackage{amsthm}
\usepackage[top=2.5cm, bottom=2.5cm, left=2cm, right=2cm]{geometry}
\usepackage[font=small]{caption}
\usepackage{enumitem}
\setlist[1]{itemsep=-5pt}
\usepackage{titlesec}
\usepackage[usenames, dvipsnames]{color}
\makeatletter

\newcommand{\setappendix}{Appendix~\thesection:~~}
\newcommand{\setsection}{\thesection~~}
\titleformat{\section}{\bfseries\LARGE}{%
	\ifnum\pdfstrcmp{\@currenvir}{appendices}=0
	\setappendix
	\else
	\setsection
\fi}{0em}{}
\makeatother
\usepackage[titletoc]{appendix}
\usepackage[pdftex]{graphicx}
\graphicspath{./plots}
\usepackage{array}
\usepackage{url}
\usepackage{mathtools}
\usepackage{amssymb,amsfonts,amsmath}
\usepackage{dsfont}
\usepackage{bm}
\usepackage{footnote}

\usepackage[utf8]{inputenc} 
\usepackage[T1]{fontenc}    
\usepackage{booktabs}       
\usepackage{nicefrac}       
\usepackage{microtype}      
\usepackage{tikz}
\usepackage{subcaption}
\usepackage{algorithmic}
\usepackage{algorithm}
\usepackage{wrapfig}
\usepackage{xr}
\usepackage{xr-hyper}
\usepackage{braket}
\usepackage{cases}

\newcommand{\Tr}{\text{Tr}}

\usepackage[pdfpagemode=UseNone,bookmarksopen=false,colorlinks=true,urlcolor=blue,citecolor=magenta,citebordercolor=blue,linkcolor=blue]{hyperref}

\input{./settings_custom.tex}
\usepackage{setspace}

\usepackage{stmaryrd}
\SetSymbolFont{stmry}{bold}{U}{stmry}{m}{n}

\begin{document}

\setcounter{tocdepth}{3}
\setcounter{secnumdepth}{3}

\title{High-temperature Expansions and Message Passing Algorithms}
\date{}
\author{Antoine Maillard$^{\star,\otimes}$, Laura Foini$^\dagger$, Alejandro Lage Castellanos$^\diamond$, \\
Florent Krzakala$^\star$, Marc M\'ezard$^\star$, Lenka Zdeborov\'a$^\dagger$}
\maketitle
{\let\thefootnote\relax\footnote{
\!\!\!\!\!\!\!\!\!\!
$\star$ Laboratoire de Physique de l'ENS, PSL University, 
CNRS, Sorbonne Universit\'es, Paris, France.\\
$\dagger$ Institut de Physique Th\'eorique, CNRS, CEA, Universit\'e Paris-Saclay, Saclay, France.\\
$\diamond$ University of Havana - Departamento de Física Te\'orica, Havana, Cuba.\\
$\otimes$ To whom correspondence shall be sent: \href{mailto:antoine.maillard@ens.fr}{antoine.maillard@ens.fr}
}}
\setcounter{footnote}{0}

\begin{abstract}	
        Improved mean-field technics are a central theme of
        statistical physics methods applied to inference and learning.
        We revisit here some of these methods using high-temperature
        expansions for disordered systems  initiated by Plefka,
        Georges and Yedidia.  We derive the Gibbs free entropy and the
        subsequent self-consistent equations 
        for a generic class of statistical models with correlated matrices and
        show in particular that many classical approximation schemes, such as
        adaptive TAP, Expectation-Consistency, or the approximations behind
        the Vector Approximate Message Passing algorithm all rely on the same
        assumptions, that are also at the heart of high-temperature expansions. 	
        We focus on the case of rotationally invariant random
        coupling matrices in the `high-dimensional' limit in which
        the number of samples and the dimension are both large, but
        with a fixed ratio. This encapsulates many widely studied
        models, such as Restricted Boltzmann Machines or Generalized
        Linear Models with correlated data matrices. In this general
        setting,  we show that all the approximation schemes
        described before are equivalent, and we conjecture that they
        are exact in the thermodynamic limit in the replica symmetric phases.
        We achieve this conclusion by resummation of the infinite
        perturbation series, which generalises a seminal result of
        Parisi and Potters. A rigorous derivation of this conjecture
        is an interesting mathematical challenge.
        On the way to these conclusions, we uncover several diagrammatical results
        in connection with free probability and random matrix theory,
        that are interesting independently of the rest of our work.
\end{abstract}

\begin{spacing}{1.09}
\tableofcontents
\end{spacing}
\renewcommand{\labelitemi}{$\bullet$}

\newpage

\input{introduction.tex}

\input{spherical_model.tex}
\input{stat_models.tex}
\input{fixed_point.tex}
\input{diagrammatics.tex}

\section*{Acknowledgments}

The authors would like to thank Yoshiyuki Kabashima, Marylou Gabri\'e, Bertrand Eynard and Jorge Kurchan for many insightful discussions.  
This work is supported by ``Investissements d'Avenir" LabEx PALM
(ANR-10-LABX-0039-PALM) (EquiDystant project, L. Foini), as well as by the
French Agence Nationale de la Recherche under grant
ANR-17-CE23-0023-01 PAIL, the European Union’s Horizon 2020 Research and Innovation Program 714608-SMiLe, and the ERC 307087 SPARCS.
Additional funding is acknowledged by AM from `Chaire de recherche sur les mod\`eles et sciences des donn\'ees', Fondation CFM pour la Recherche-ENS.

\bibliographystyle{alpha}
\bibliography{refs}

\newpage
\appendix
\input{appendix_operator_U.tex}
\input{appendix_order4_spherical.tex}
\input{appendix.tex}

\end{document}

%% file: settings_custom.tex
\usepackage{amsthm}
\usepackage{cleveref}

\DeclareUnicodeCharacter{00A0}{~}

\def \({\left(}
\def \){\right)}
\def \[{\left[}
\def \]{\right]}

\newcommand{\bb}{{\textbf {b}}}
\newcommand{\bmm}{{\textbf {m}}}
\newcommand{\br}{{\textbf {r}}}

\newcommand{\bomega}{{\boldsymbol{\omega}}}

\newcommand{\bF}{{\textbf {F}}}
\newcommand{\bY}{{\textbf {Y}}}

\newcommand{\bh}{{\textbf {h}}}

\newcommand{\bv}{{\textbf {v}}}

\newcommand{\bX}{{\textbf {X}}}
\newcommand{\bx}{{\textbf {x}}}

\newcommand{\blambda}{{\boldsymbol{\lambda}}}

\newcommand{\by}{{\textbf {y}}}
\newcommand{\bz}{{\textbf {z}}}

\newcommand{\bc}{{\textbf {c}}}

\newcommand{\be}{\begin{equation}}
\newcommand{\ee}{\end{equation}}

\newcommand\smallO{
  \mathchoice
    {{\scriptstyle\mathcal{O}}}
    {{\scriptstyle\mathcal{O}}}
    {{\scriptscriptstyle\mathcal{O}}}
    {\scalebox{.7}{$\scriptscriptstyle\mathcal{O}$}}
  }
\newcommand{\bea}{\begin{align}}
\newcommand{\eea}{\end{align}}
\newcommand{\norm}[1]{\left\lVert#1\right\rVert} 

\newtheorem{model}{\textbf{Model}}
  \newenvironment{custommodel}[1]
  {\renewcommand\themodel{#1}\model}
  {\endmodel}

\newtheorem{theorem}{Theorem}
\newtheorem{conjecture}{Conjecture}

\newtheorem{property}{\textbf{Property}}

\theoremstyle{definition}

\DeclareMathAlphabet{\varmathbb}{U}{bbold}{m}{n}

\newcommand{\EE}{\mathbb{E}}
\newcommand{\bbR}{\mathbb{R}}

\newcommand{\bbN}{\mathbb{N}}

\newcommand{\bbC}{\mathbb{C}}

\renewcommand{\Tr}{{\rm Tr}}

%% file: introduction.tex
\section{Introduction}\label{sec:introduction}

   \subsection{Background and overview of related works}\label{subsec:background}

   Many inference and learning tasks can be formulated as a statistical physics problem, where one needs to compute or approximate the marginal distributions of single variables in an interacting model. This is, for instance, the basis behind the popular variational mean-field approach \cite{Jordan}. Going beyond the naive mean-field theory has been a constant goal in both physics and machine learning.  One approach, for instance, has been very effective on tree-like structures: the Bethe approximation, or Belief-Propagation. Its development in the statistical physics of disordered systems can be traced back to Thouless-Anderson-Palmer (TAP) \cite{thouless1977solution} and has seen many developments since then \cite{mezard1987spin,yedidia2003understanding,mezard2009information,zdeborova2016statistical}. Over the last decades, in particular, there has been many works on densely connected models, leading to a myriad of different approximation schemes. In many disordered problems with i.i.d.\ couplings, a classical approach has been to write the TAP equations as an iterative scheme. Iterative algorithms based on this scheme are often called \emph{Approximate Message Passing} (AMP) \cite{donoho2009message,krzakala2012probabilistic} in this context.

   AMP, or TAP, is an especially powerful approach when the coupling constants in the underlying statistical model
   are distributed as i.i.d.\ variables. 
   This is, of course, a strong limitation and many inference schemes have been  designed to improve on it: 
   the adaptive TAP (adaTAP) method \cite{opper2001adaptive,opper2001tractable}, approximation schemes such as
   \emph{Expectation-Consistency} (EC) \cite{minka2001expectation,opper2005expectation} and the recent improvements of AMP such as \emph{Vector approximate Message Passing} (VAMP) and its variants \cite{ma2017orthogonal,rangan2017vector,schniter2016vector,opper2016theory,ccakmak2016self}. 
   Given all these approaches, one may wonder how different they are, and when they actually lead to asymptotically exact inference. 
   In this paper, we wish to address this question using two main tools: high-temperature expansions and random matrix theory.

   High-temperature expansions at fixed order parameters (denoted in this paper as ``Plefka expansions'') are an important tool of the study of disordered systems. 
   In the context of spin glass models, they have been introduced by Plefka \cite{plefka1982convergence} for the Sherrington-Kirkpatrick (SK) model, and have been subsequently generalized, in particular by Georges-Yedidia \cite{georges1991expand}. 
   This latter paper provides a systematic way to compute high-temperature (or high-dimension) expansions of the Gibbs free entropy \emph{for a fixed value of the order parameters} (that is Plefka expansions).

   One aim of the present paper is to apply this method to a general class of inference problems with 
   pairwise interactions, in which the coupling constants are not i.i.d., but they can have strong correlations, while keeping a rotational invariance that will be made explicit below. 
   In particular, we generalize earlier and inspirational work by Parisi and Potters \cite{parisi1995mean}, who computed the self-consistent equations for the marginals in 
   Ising models with orthogonal couplings via a resummation of the infinite series given by the high-temperature expansion. We shall show that a similar resummation yields the EC, adaTAP and VAMP formalisms.

   \subsection{Structure of the paper, and summary of our contributions}\label{subsec:main_results}

   In this paper, we perform Plefka expansions for a generic class of models of pairwise interactions with correlated matrices.
   We provide a detailed derivation of the method, inspired by the work of Georges-Yedidia \cite{georges1991expand} for Ising models,
   and we include new results on the diagrammatics of the expansions, leveraging  rigorous results of random matrix theory. This yields a general framework that encapsulates many known properties of systems sharing this pairwise structure.
   The main message of this work is that the three successful approximation schemes that have been developed in the last two decades, Expectation-Consistency, adaTAP or Vector Approximate Message Passing, are equivalent and rely on the same hidden hypothesis. A careful analysis of the Plefka expansion reveals this hypothesis, as it identifies the class of  high-temperature expansion diagrams that are effectively kept in these three schemes. A diagrammatic analysis leads us to conjecture that all these methods are asymptotically exact 
   for rotationally-invariant models, in the high-temperature phase. It is also worth noting that although all four methods (Expectation-Consistency, adaTAP, Vector Approximate Message Passing, Plefka expansion) lead to the same mean-field equations, the (most recent) VAMP approach presents the advantage of generating a ``natural'' way to iterate these equations, which turns them into efficient algorithms.
   We now turn to a more precise description of the content of the paper.
   Throughout the paper, we will use two random matrix ensembles that we will both refer to as being \emph{rotationally invariant}.
   The first one is defined as a measure over the set $\mathcal{S}_N$ of symmetric matrices:

   \begin{custommodel}{S}[Symmetric rotationally invariant matrix]\label{model:sym_rot_inv}
      Let $N \geq 1$. $J \in {\cal S}_N$ is generated as $J = O D O^\intercal$,
      in which $O \in \mathcal{O}(N)$ is drawn uniformly from the (compact) orthogonal group $\mathcal{O}(N)$,
      and $D = \mathrm{Diag}(\{d_i\}_{i=1}^N)$ is a random diagonal matrix, such that its empirical spectral distribution 
      $\rho^{(N)}_D \equiv \frac{1}{N} \sum_{i=1}^N \delta_{d_i}$ converges (almost surely) as $N \to \infty$ to a probability distribution $\rho_D$ with compact support.
      The smallest and largest eigenvalue of $D$ are assumed to converge almost surely to the infimum and supremum of the support of $\rho_D$.
   \end{custommodel}
   In a similar way, we define an ensemble of rectangular rotationally invariant matrices:
   \begin{custommodel}{R}[Rectangular rotationally invariant matrix]\label{model:nsym_rot_inv}
      Let $N \geq 1$, and $M = M(N) \geq 1$ such that $M/N \to \alpha >0$ as $N \to \infty$. $L \in \bbR^{M \times N}$ is generated via its SVD decomposition $L = U \Sigma V^\intercal$, in which $U \in \mathcal{O}(M)$ and $V \in \mathcal{O}(N)$ are
      drawn uniformly from their respective orthogonal group. $D \equiv \Sigma^\intercal \Sigma = \mathrm{Diag}(\{d_i\}_{i=1}^{N})$ is a diagonal matrix, such that its empirical spectral distribution 
      $\rho^{(N)}_D \equiv \frac{1}{N} \sum_{i=1}^N \delta_{d_i}$ converges (almost surely) as $N \to \infty$ to a probability distribution $\rho_D$, which has compact support.
      The smallest and largest eigenvalue of $D$ are assumed to converge almost surely to the infimum and supremum of the support of $\rho_D$.
   \end{custommodel}
   \paragraph{Examples} Examples of such random matrix ensembles include matrices generated via a \emph{potential} $V(x)$: one can generate 
   $J \in \mathcal{S}_N$ with a probability density proportional to $e^{-\frac{N}{2} \, \mathrm{Tr}  \, V(J)}$, and this kind of matrix satisfies the 
   hypotheses of Model~\ref{model:sym_rot_inv}. These ensembles also include the following well-known examples:
   \begin{itemize} 
      \item The Gaussian Orthogonal Ensemble (GOE), in the case of Model~\ref{model:sym_rot_inv} with a potential $V(x) = x^2/2$.
      \item The Wishart ensemble with a ratio $\psi \geq 1$. This corresponds to a random matrix $W = X X^\intercal / m$, with $X \in \bbR^{n \times m}$ an i.i.d.\ standard Gaussian
      matrix, and $n,m \to \infty$ with $m / n \to \psi$.
      This ensemble satisfies Model~\ref{model:sym_rot_inv}, with a potential $V(x) = x - (\psi-1) \log x$.
      \item Standard Gaussian i.i.d.\ rectangular matrices, for Model~\ref{model:nsym_rot_inv}. One can also think of them as generated via a potential, 
      as the probability density of such a matrix is $\mathbb{P}(L) \propto e^{-\frac{1}{2}\mathrm{Tr}\, L^\intercal L}$.
      \item Generically, consider a random matrix $L$ from Model~\ref{model:nsym_rot_inv}. Then, both $J_1 \equiv L^\intercal L$ and $J_2 \equiv L L^\intercal$ 
      satisfy the hypotheses of Model~\ref{model:sym_rot_inv}.
   \end{itemize}
    The structure of our work is as follows:
\begin{itemize}[itemsep=-2pt]
   \item \textbf{Spherical models with rotationally invariant couplings} 
   In Sec.~\ref{sec:spherical_bipartite}, we focus on spherical models and we generalize the seminal works of \cite{marinari1994replica,marinari1994replica2, parisi1995mean}.
   While they studied Ising models with orthogonal couplings,
   we consider spherical models, just assuming the coupling matrix to be rotationally invariant. We consider two types of models: ``symmetric'' models with an interaction 
   of the type $\bx^\intercal J \bx$ , in which $J$ follows Model~\ref{model:sym_rot_inv}, and ``bipartite'' models with interactions of the type $\bh^\intercal F \bx$, 
   in which $F$ follows Model~\ref{model:nsym_rot_inv}.
    This encapsulates orthogonal couplings, but can also be applied to 
   other random matrix ensembles such as the Gaussian Orthogonal Ensemble (GOE), the Wishart ensemble, and many others.
    Using diagrammatic results that we derive with
   random matrix theory, we  conjecture a resummation of the Plefka expansion giving the Gibbs free entropy in these models. 
   Our results are in particular consistent with the findings of classical works for Gaussian couplings \cite{plefka1982convergence} and orthogonal couplings \cite{parisi1995mean}.
   \item \textbf{Plefka expansion for statistical models with correlated couplings} 
   Sec.~\ref{sec:stat_models} is devoted to 
   the description of the Plefka expansion for different statistical models and inference problems which possess a coupling or data matrix that has rotation invariance properties. 
   We consider models similar to the spherical models of Sec.~\ref{sec:spherical_bipartite}, but with generic prior distributions
   on the underlying variables. 
   In Sec.~\ref{subsec:ep_adatap}, we recall the Expectation-Consistency (EC), adaTAP and VAMP approximations
   and comment briefly on their respective history, before showing that they are equivalent. 
   As a consequence, we will generically refer to these approximations as the \emph{Expectation-Consistency approximations} (EC).
   We hope that our paper will help providing a unifying presentation of these works, generalizing them by leveraging random matrix theory.
   Our main conjecture for this part can be stated as the following:
   \begin{conjecture}\label{conj:main_conj_approximation}[Informal]
    For statistical models of symmetric or bipartite interactions with coupling matrices that satisfy respectively Model~\ref{model:sym_rot_inv} or
      Model~\ref{model:nsym_rot_inv}, the three equivalent approximations, Expectation-Consistency, adaTAP and VAMP (generically denoted EC approximations), are exact in the large size limit in the high temperature phase. 
   \end{conjecture}
We believe that the validity of the above conjecture extends beyond the high temperature phase. In particular that it is correct for inference problems in the Bayes-optimal setting, and more generally anytime the system is in a replica symmetric phase as defined in \cite{mezard1987spin}.

The approximation behind EC approximations can be checked order by order using our high-temperature Plefka expansions technique
   and its resummation.
   We then derive Plefka expansions for these generic models, and we apply it to different 
   situations, namely:
   \begin{itemize}  
    \item In Sec.~\ref{subsubsec:sym_generic_prior} we perform a Plefka expansion for a generic symmetric rotationally invariant model with pairwise interactions.
      Using this method and our diagrammatic results, we show then in Sec.~\ref{subsubsec:connection_plefka_ep_adatap} that the EC approximations
      are exact for these models in the large size limit. 
    \item In Sec.~\ref{subsubsec:plefka_hopfield} we apply our general result to the TAP free energy of the Hopfield model \cite{hopfield1982neural}, an Ising spin model with a correlated matrix of the Wishart ensemble, used as a basic model of neural network.
    In particular, we find back straightforwardly the results of \cite{nakanishi1997mean} and \cite{mezard2017mean}.
    \item In Sec.~\ref{subsec:plefka_replicas} we extend our Plefka expansion and the corresponding diagrammatic techniques to the study of a
    \emph{replicated} system, in which we constraint the overlap between different replicas. The interest for such systems comes as a consequence of the celebrated replica method of theoretical physics \cite{mezard1987spin}.
    \item Finally, we show in Sec.~\ref{subsec:plefka_bipartite_models} how we can use these results to derive the Plefka-expanded free entropy 
    for a very broad class of bipartite models, which includes the Generalized Linear Models (GLMs) with correlated data matrices, and the Compressed Sensing problem.
   \end{itemize}
   We emphasize that we were able to derive the free entropy of all these models using very generic arguments relying only on the rotational invariance of the problem.
   \item \textbf{The TAP equations and message passing algorithms}
    Finally, we show in Sec.~\ref{sec:applications_algorithms} that the TAP (or EC) equations that we derived by maximizing the Gibbs free entropy of rotationally invariant models
   can strikingly be understood as the fixed point equations of message passing algorithms.
   In the converse way, many message-passing algorithms can be seen as an iteration scheme of the TAP equations. 
   This was known in many models
   in which the underlying data matrix was assumed to be i.i.d.\ For instance, the Generalized Approximate Message Passing (GAMP) algorithm \cite{rangan2011generalized} was shown in \cite{krzakala2012probabilistic}
   to be equivalent to the TAP equations, a result that we find back in Sec.~\ref{subsec:gamp}, while TAP equations were already iterated for Restricted Boltzmann Machines, see \cite{tramel2018deterministic}.
   In the Plefka expansion language, these results relied on the early stopping of the expansion at order $2$ (in powers of the couplings) as a consequence of the i.i.d.\ hypothesis. 
   Using our resummation results to deal with the series at infinite orders, we were able to generalize these correspondences to correlated models.
   We argue that the stationary limit of the Vector Approximate Message Passing (VAMP)
   algorithm \cite{rangan2017vector} (that is its fixed point equations) for compressed sensing with correlated matrices gives back our TAP equations derived via Plefka expansion, see Sec.~\ref{subsec:vamp}.
   Even more generally, the Generalized Vector Approximate Passing (G-VAMP) algorithm \cite{schniter2016vector}, defined for the very broad class of Generalized Linear Models with correlated matrices,
   yields fixed point equations that are equivalent to our Plefka-expanded TAP equations, see Sec.~\ref{subsec:gvamp}.
   Combined with the results of Sec.~\ref{sec:stat_models}, this indicates that the VAMP algorithm is an example of an approximation 
   scheme that follows conjecture~\ref{conj:main_conj_approximation}.
    \item \textbf{Diagrammatics of the expansion and random matrix theory} Our results are largely based on a better control on the diagrammatics of the Plefka expansions for rotationally invariant
    random matrices, which are presented in Sec.~\ref{sec:diagrammatics}. We leverage mathematically rigorous results on Harish-Chandra-Itzykson-Zuber (HCIZ) integrals \cite{harish1957differential,itzykson1980planar,guionnet2005fourier, collins2007new}, involving transforms of the 
    asymptotic spectrum of the coupling matrix, to argue that only a very specific class of diagrams contributes to the high-temperature expansion of a system with rotationally invariant couplings.
    These results are used throughout our study, and are detailed in Sec.~\ref{sec:diagrammatics}.
    Some generalizations are postponed to Appendix~\ref{sec:generalizations_expansions}.
\end{itemize}

%% file: spherical_model.tex
\section{Symmetric and bipartite spherical models with rotationally-invariant couplings}\label{sec:spherical_bipartite}

In this section we consider two spherical models that will serve both as
guidelines and building blocks for our subsequent analysis. We show
in details how to perform the Plefka-Georges-Yedidia high-temperature expansion in this context, and the precise diagrammatic results
that allow us to resum the Plefka series for rotationally invariant couplings.
These results will be useful to clarify our subsequent derivation of the TAP equations in more involved models, 
and are also interesting by themselves from a random matrix theory point of view.

\subsection{Symmetric spherical model}\label{subsec:sym_spherical_model}

In this section $N \geq 1$, $\sigma > 0$, and we define the following pairwise interaction Hamiltonian 
on $\mathbb{S}^{N-1}(\sigma \sqrt{N})$, the $N$-th dimensional sphere of radius $\sigma \sqrt{N}$:
   \begin{align}\label{eq:hamiltonian_spherical}
      H_{J}(\bx) &= -\frac{1}{2} \bx^\intercal J \bx = -\frac{1}{2}\sum_{1 \leq i,j \leq N} J_{ij} x_i x_j, \qquad \bx \in \mathbb{S}^{N-1}(\sigma \sqrt{N}).
   \end{align}
   The coupling matrix $J$ is a $N \times N$ symmetric random matrix
 drawn from Model~\ref{model:sym_rot_inv}.

   \subsubsection{Direct free entropy computation}\label{subsubsec:direct_symmetric}
The Gibbs measure for our model at inverse temperature $\beta$ is defined as:
\begin{align}\label{eq:gibbs_measure_original}
   P_{\beta,J}(\mathrm{d}\bx) &\equiv \frac{1}{Z_{\beta,J}} e^{\frac{\beta}{2} \sum_{i,j} J_{ij} x_i x_j}\mathrm{d}\bx,
\end{align}
in which ${\rm d}\bx$ is the usual surface measure on the sphere $\mathbb{S}^{N-1}(\sigma\sqrt{N})$.
 We write the partition function of the model introducing a Lagrange multiplier $\gamma$ to enforce the condition $\norm{\bx}^2 = N \sigma^2$.
 We will write $A_N \simeq B_N$ to denote that $\frac{1}{N} \log A_N = \frac{1}{N} \log B_N + \smallO_N(1)$.
 At leading exponential order, one has:
\begin{align}
   Z_{\beta,J} &\equiv \int_{\mathbb{S}^{N-1}(\sigma \sqrt{N})} \, \mathrm{d}\bx\, e^{\frac{\beta}{2} \sum_{i,j} J_{ij} x_i x_j}, \\
   &\simeq \int {\rm d} \gamma \ \prod_{i=1}^N \int_{\bbR} {\rm d} x_i \ e^{\frac{\beta}{2} \sum_{i,j} J_{ij} x_i x_j + \frac{\gamma}{2} (N \sigma^2 - \sum_{i} x_i^2)},  \nonumber \\
   &\simeq \exp \left[\inf_\gamma \left\{ \log \left[ \int \prod_{i=1}^N {\rm d} x_i \ e^{\frac{\beta}{2} \sum_{i,j} J_{ij} x_i x_j + \frac{\gamma}{2} (N \sigma^2 - \sum_{i} x_i^2)} \right] \right\}\right].
   \label{eq:Z_spherical_lagrange}
\end{align} 
Denoting $\gamma(\beta)$ the solution to the saddle-point equation in eq.~(\ref{eq:Z_spherical_lagrange}), we have effectively defined a new Gibbs measure:
\begin{align}\label{eq:gibbs_measure}
   P_{\beta,J}(\mathrm{d}\bx) &\equiv \frac{1}{Z_{\beta,J}(\gamma)}e^{\frac{\beta}{2} \sum_{i,j} J_{ij} x_i x_j} e^{-\frac{\gamma(\beta)}{2} ||\bx||_2^2} \,\mathrm{d}\bx,
\end{align}
where now $\mathrm{d}\bx$ is the usual Euclidian measure on $\bbR^N$.
Following~\cite{KTJ} we diagonalize the Hamiltonian and we integrate over the spins in this new basis, which yields:
\begin{equation}\label{eq:introduction_gamma}
\displaystyle
Z_{\beta,J} \simeq \exp\left[\inf_\gamma \left\{\frac{N}{2} \left(\log 2 \pi +  \gamma \sigma^2  - \frac{1}{N} \sum_{\lambda} \log (\gamma - \beta \lambda)\right) \right\}\right],
\end{equation} 
in which the sum over $\lambda$ runs over the set of eigenvalues of $J$. Taking the $N \to \infty$ limit, the saddle point equation reads:
\begin{equation}
\lim_{N \to \infty} \frac{1}{N} \sum_{\lambda} \frac{1}{\gamma - \beta \lambda} = \sigma^2,
\end{equation}
which we can write as a function of the limiting spectral law $\rho_D$
of the matrix $J$ (defined in Model~\ref{model:sym_rot_inv}):
\begin{equation}\label{SP_toy}
\int \,  \frac{\rho_D(\mathrm{d}\lambda)}{\gamma - \beta \lambda} = \sigma^2.
\end{equation}
We assumed (see Model~\ref{model:sym_rot_inv}) that the support of $ \rho_D$ is compact so that we can define its maximum $\lambda_{\rm max} \in \bbR$.
Under these assumptions, eq.~(\ref{SP_toy}) has the solution:
\begin{equation}
\gamma = \beta {\cal R}_{\rho_D}(\beta \sigma^2) + \frac{1}{\sigma^2} = \beta {\cal S}_{\rho_D}^{-1}(-\beta \sigma^2),
\end{equation}
as long as $-{\cal S}_{\rho_D}(\lambda_{\rm max}) \geq \beta \sigma^2$, where ${\cal R}_{\rho_D}$ is the ${\cal R}$-transform of $\rho_D$ and ${\cal S}_{\rho_D}$ its Stieltjes transform (see Appendix~\ref{sec:appendix_rmt} for their definitions).
In the opposite case (if $\beta \sigma^2 > -{\cal S}_{\rho_D}(\lambda_{\rm max})$), $\gamma$ `sticks' to the solution $\gamma=\lambda_{\rm max}\beta$.
The intensive free entropy $\Phi_{J}(\beta)$ is defined as:
\begin{equation}
\Phi_{J}(\beta) \equiv \lim_{N \to \infty} \frac{1}{N} \log Z_{\beta,J}.
\end{equation}
In the end, we can compute the free entropy in the \emph{high-temperature} phase $\beta \leq \beta_c \equiv - \sigma^{-2}{\cal S}_{\rho_D}(\lambda_{\rm max})$:
\begin{align}
   \Phi_{J}(\beta) &=  \frac{1}{2} \left( 1 + \log 2 \pi \sigma^2\right) + \frac{\beta \sigma^2}{2} {\cal R}_{\rho_D}(\beta \sigma^2) - \frac{1}{2} \int \, \rho_D(\mathrm{d}\lambda) \log\left[\beta \sigma^2 {\cal R}_{\rho_D}(\beta \sigma^2) - \beta \sigma^2 \lambda + 1\right].
\end{align}
By taking the derivative of this expression with respect to $\beta$ it is easy to show that this simplifies to:
\begin{equation}\label{eq:spherical_hight}
\Phi_{J}(\beta) =  \frac{1}{2} \left( 1 + \log 2 \pi \sigma^2 \right)+ \frac12 \int_0^{\beta \sigma^2} {\cal R}_{\rho_D}(x) {\rm d} x.
\end{equation}
In the low temperature phase (for $\beta \geq \beta_c = -\sigma^{-2} {\cal S}_{\rho_D}(\lambda_{\rm max})$) one has
\begin{equation}\label{eq:spherical_lowt}
   \Phi_{J}(\beta) = \frac12  \left( \log 2 \pi +  \lambda_{\rm max}\beta \sigma^2 -  \log \beta  - \int \,  \rho_D(\mathrm{d}\lambda) \ \log (\lambda_{\rm max} - \lambda) \right).  
\end{equation}
Note that both in the high and low temperature phases the free entropy can formally be expressed as:
\begin{align}
   \Phi_{J}(\beta) = \frac{1}{2} \log 2 \pi + \frac{1}{2}\inf_{\gamma}\left[\gamma \sigma^2 - \int \, \rho_D(\mathrm{d}\lambda) \log(\gamma - \beta \lambda)\right],
\end{align}
a formulation which is both more compact and easier to implement algorithmically for generic matrices $J$.

\paragraph{Remark} The free entropy is usually defined as an average over the quenched disorder $J$, but here it is clear that
the free entropy is self-averaging as a function of $J$, so that taking this average is trivial. Moreover, $\Phi_J(\beta)$ only depends on $J$ via $\rho_D$, its
asymptotic eigenvalue distribution.
\paragraph{Remark} The derivation of the free entropy both in the high and low temperature phase has been made rigorous in \cite{guionnet2005fourier},
and the method of proof also essentially consists in fixing a Lagrange multiplier to enforce the condition $\sum_{i} s_i^2 =  \sigma^2 N$.

\subsubsection{Plefka expansion and the Georges-Yedidia formalism}\label{subsubsec:plefka_sym_spherical}
   
A more generic way to compute the free entropy is to follow the formalism of \cite{georges1991expand} to perform a high-temperature Plefka expansion \cite{plefka1982convergence}. The goal is to expand the free entropy at low $\beta$, in the high-temperature phase. 
In order to do so, we introduce the very useful $U$ operator defined in Appendix~A of \cite{georges1991expand}.
We will compute the free entropy given the constraints on the means $\braket{x_i}_\beta = m_i$ and on the variances $\braket{x_i^2}_\beta = v_i + m_i^2$.
The notation $\braket{\cdot}_\beta$ indicates an average over the Gibbs measure of our system at inverse temperature $\beta$, see eq.~(\ref{eq:gibbs_measure}).
A set of parameters $\{m_i,v_i\}$ will thus determine a free entropy value, and the comparison with the direct calculation of Sec.~\ref{subsubsec:direct_symmetric} will be made by 
maximizing the free entropy with respect to $\{m_i,v_i\}$. 
We can enforce the spherical constraint $\norm{\bx}_2^2 = \sigma^2 N$ by constraining our choice of parameters $\{m_i,v_i\}$ to satisfy the identity:
\begin{align}\label{eq:spherical_constraint_mv}
   \sigma^2 = \frac{1}{N} \sum_{i=1}^N \left[v_i + m_i^2\right].
\end{align}
The Lagrange parameters introduced to fix the magnetizations are denoted $\{\lambda_i\}$, and the ones used to fix the variances are denoted $\{\gamma_i\}$.
For clarity we will keep their dependency on $\beta$ explicit only when needed. 
For a given $\beta$ and a given $J$ one defines the operator $U$ of Georges-Yedidia:
\begin{align}\label{eq:def_U}
U(\beta,J) \equiv H_{J} - \braket{H_{J}}_\beta + \sum_{i=1}^N \partial_\beta\lambda_i(\beta) (x_i - m_i) + \frac{1}{2} \sum_{i=1}^N\partial_\beta \gamma_i(\beta) \left[x_i^2 - v_i - m_i^2 \right],
\end{align}
The derivation of $U$ as well as its (many) useful properties are briefly recalled in Appendix~\ref{sec:appendix_operator_U}.
We are now ready to compute the first orders of the expansion of the free entropy $\Phi_{J}(\beta)$ in terms of $\beta$. 
In this expansion the Lagrange parameters $\{\lambda_i(\beta), \gamma_i(\beta)\}$ are always considered at $\beta = 0$, so we drop their $\beta$-dependency.
We detail the first orders of the expansion, following Appendix~\ref{sec:appendix_operator_U} (cft. Appendix~A of \cite{georges1991expand}).
\paragraph{Order 0}
First of all, taking $\beta = 0$ one has easily: 
\begin{align*}
\Phi_{J}(\beta = 0) &= \frac{1}{2 N} \sum_{i=1}^N \gamma_i (v_i + m_i^2) + \frac{1}{N}\sum_{i=1}^N \lambda_i m_i + \frac{1}{N}\log \int_{\mathbb{R}^N} e^{ - \frac{1}{2} \sum_{i} \gamma_i x_i^2 - \sum_{i} \lambda_i x_i}  \mathrm{d} \bx, \\
&= \frac{1}{2} \log 2 \pi + \frac{1}{N} \sum_{i=1}^N \left[\frac{\gamma_i}{2}(v_i + m_i^2) - \frac{1}{2} \log \gamma_i + \lambda_i m_i + \frac{\lambda_i^2}{2 \gamma_i}\right] .
\end{align*}
This yields after extremization over $\{\lambda_i,\gamma_i\}$:  
\begin{align}\label{eq:spherical_order0}
\Phi_{J}(\beta = 0) = \frac{1}{2} \left[1 + \log 2 \pi \right] + \frac{1}{2N} \sum_{i=1}^N \log v_i.
\end{align}

\paragraph{Order 1}

At order $1$, one easily derives: 
\begin{align}\label{eq:spherical_order1}
\left(\frac{\partial \Phi_J}{\partial \beta}\right)_{\beta=0} =- \frac{1}{N} \braket{H_{J}}_{\beta = 0} = \frac{1}{2N} \sum_{i,j} J_{ij} m_i m_j + \frac{1}{2 N}\sum_{i=1}^N J_{ii} v_i.
\end{align}
We can now make use of the Maxwell-type relations which are valid at any $\beta$:  
\begin{align}
   \label{eq:maxwell_spherical_1}
   \gamma_i(\beta) &= 2 N\frac{\partial \Phi_{J}(\beta)}{\partial v_i}, \\
   \label{eq:maxwell_spherical_2}
   m_i \gamma_i(\beta) + \lambda_i(\beta) &= N\frac{\partial \Phi_{J}(\beta)}{\partial m_i}.
\end{align}
These relations plugged in eq.~(\ref{eq:spherical_order1}) lead to $\partial_\beta \gamma_i (\beta = 0) = J_{ii}$ and $\partial_\beta \lambda_i(\beta=0) = \sum_{j(\neq i)} J_{ij} m_j$. 
We then obtain the $U$ operator at $\beta = 0$ from eq.~(\ref{eq:def_U}): 
\begin{align}\label{eq:U_spherical}
U(\beta = 0,J) &= -\frac{1}{2}\sum_{i \neq j}  J_{ij} (x_i - m_i) (x_j - m_j).
\end{align}

\paragraph{Order 2}

Following eq.~(\ref{Eq_U2}) in Appendix~\ref{sec:appendix_operator_U}, we have the relation: 
\begin{align}\label{eq:spherical_order2}
\frac{1}{2} \left(\frac{\partial^2 \Phi_J}{\partial \beta^2}\right)_{\beta=0} &= \frac{1}{ 2N}\braket{U^2}_{\beta = 0} = \frac{1}{4 N} \sum_{i \neq j} J_{i j}^2 v_i v_j.
\end{align}

\paragraph{Order 3 and 4}

For the order $3$, we obtain:
\begin{align}\label{eq:spherical_order3}
\frac{1}{3!}\left(\frac{\partial^3 \Phi_J}{\partial \beta^3}\right)_{\beta=0} &= -\frac{1}{6 N}\braket{U^3}_{\beta = 0} = \frac{1}{6N} \sum_{i,j,k } J_{ij} J_{jk} J_{ki} v_i v_j v_k + \smallO_N(1),
\end{align}
in which the sum is made over pairwise distinct $i,j,k$ indices. 
Applying eq.~(\ref{Eq_U4}) we reach: 
\begin{align}\label{eq:spherical_order4}
\frac{1}{4!}\left(\frac{\partial^4 \Phi_J}{\partial \beta^4}\right)_{\beta=0} &= \frac{1}{8 N} \sum_{i,j,k,l} J_{ij} J_{jk} J_{kl} J_{li} v_i v_j v_k v_l + \smallO_N(1),
\end{align}
where again, $i,j,k,l$ are pairwise distinct indices. For pedagogical purposes (and since it will be useful for the following sections), we detail this calculation in Appendix~\ref{sec:appendix_order4_spherical}.
\paragraph{Larger orders}

By its very nature, the perturbative expansion of Georges-Yedidia \cite{georges1991expand} can not (somehow disappointingly) give an analytic result for an arbitrary perturbation order $n$. 
However, the results up to order $4$ of eqs.~(\ref{eq:spherical_order0}), (\ref{eq:spherical_order1}), (\ref{eq:spherical_order2}), (\ref{eq:spherical_order3}), (\ref{eq:spherical_order4}) lead to the following natural conjecture for the free entropy \emph{at a given realization of the disorder}: 
\begin{align}\label{eq:spherical_full}
\Phi_{J}(\beta) &= \frac{1}{2} \left[1 + \log 2 \pi \right] + \frac{1}{2N} \sum_{i=1}^N \log v_i  + \frac{\beta}{2N} \sum_{i \neq j} J_{ij} m_i m_j  \nonumber \\
&+ \frac{1}{N} \sum_{p=1}^\infty \frac{\beta^p}{2p} \sum_{\substack{i_1,\cdots,i_p \\ \text{pairwise distincts }}} J_{i_1 i_2} J_{i_2 i_3} \cdots J_{i_{p-1} i_p} J_{i_p i_{1}} \prod_{\alpha=1}^p v_{i_\alpha}     + \smallO_N(1) .  
\end{align}
Note that in order to obtain this formula, we took the $N \to \infty$ limit at every perturbation 
order in $\beta$, which is part of the implicit assumptions of the Plefka expansion.
The terms of this perturbative expansion can be represented diagrammatically as \emph{simple cycles} of order $p$, see Fig.~\ref{fig:simple_cycle}.
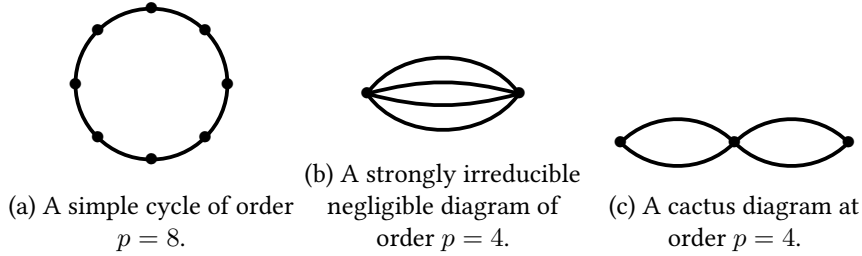
\begin{figure}[t]
 \centering
\captionsetup{justification=centering}
\begin{subfigure}[b]{0.22\textwidth}
   \centering
\begin{tikzpicture}[scale=1.]
\node (i0) at (-1,0) {$\bullet$};
\node (i1) at (1,0) {$\bullet$};
\node (i2) at (0,1) {$\bullet$};
\node (i3) at (0,-1) {$\bullet$};
\node (i4) at (0.707,0.707) {$\bullet$};
\node (i5) at (-0.707,0.707) {$\bullet$};
\node (i6) at (0.707,-0.707) {$\bullet$};
\node (i7) at (-0.707,-0.707) {$\bullet$};
\draw [line width = 0.5mm](0,0) circle (1cm);
\end{tikzpicture}
\caption{A simple cycle of order $p=8$.}\label{fig:simple_cycle}
\end{subfigure}
\begin{subfigure}[b]{0.22\textwidth}
   \centering
\begin{tikzpicture}[scale=1.]
\node (i0) at (0,2) {$\bullet$};
\node (i1) at (2,2) {$\bullet$};
\draw [line width = 0.5mm] (i0.center) to [out=55,in=125] (i1.center);
\draw [line width = 0.5mm] (i0.center) to [out=-55,in=-125] (i1.center);
\draw [line width = 0.5mm] (i0.center) to [out=15,in=165] (i1.center);
\draw [line width = 0.5mm] (i0.center) to [out=-15,in=-165] (i1.center);
\end{tikzpicture}
\caption{A strongly irreducible negligible diagram of order $p = 4$.}\label{fig:negligible_diagram}
\end{subfigure}
\begin{subfigure}[b]{0.22\textwidth}
   \centering
\begin{tikzpicture}[scale=1.]
\node (i0) at (0,2) {$\bullet$};
\node (i1) at (1.5,2) {$\bullet$};
\node (i2) at (3,2) {$\bullet$};
\draw [line width = 0.5mm] (i0.center) to [out=45,in=135] (i1.center);
\draw [line width = 0.5mm] (i0.center) to [out=-45,in=-135] (i1.center);
\draw [line width = 0.5mm] (i1.center) to [out=45,in=135] (i2.center);
\draw [line width = 0.5mm] (i1.center) to [out=-45,in=-135] (i2.center);
\end{tikzpicture}
\caption{A cactus diagram at order $p=4$.} \label{fig:nonnegligible_diagram}
\end{subfigure}
 \caption{Representation of the expansion eq.~(\ref{eq:spherical_full})
  with diagrams. Each vertex represents an index $i_\alpha$ and carries a factor $v_{i_\alpha}$, while each edge is a factor $J_{ij}$.}\label{fig:diagrams}
\end{figure}

In general, at any order in the expansion one can construct a
diagrammatic representation of the contributing terms, and one
expects that only strongly
irreducible diagrams contribute to the free entropy. Strongly
irreducible diagrams are those that cannot be split into two pieces by removing a vertex~\cite{georges1991expand}
(examples are given in Fig.~\ref{fig:simple_cycle} and \ref{fig:negligible_diagram}).
However we retain only simple cycles as the one depicted in Fig.~\ref{fig:simple_cycle} because
other diagrams as in Fig.~\ref{fig:negligible_diagram} are negligible when $N \to \infty$ for rotationally invariant models, 
as we argue in Section~\ref{subsec:generic_diagrams_expectation}. For
the case of orthogonal couplings, this dominance of simple cycles was already noted in \cite{parisi1995mean}. 
On the other hand, generic cactus diagrams like the one pictured in Fig.~\ref{fig:nonnegligible_diagram}
are not negligible, but they cancel out and do not appear in the final
form of the expansion 
(at order $4$, this is shown in Appendix ~\ref{sec:appendix_order4_spherical}).

We shall now prove the dominance of simple cycles, and the correctness
of eq.~(\ref{eq:spherical_full}), in the high-temperature phase. In this
phase, the solution to the maximization of
eq.~(\ref{eq:spherical_full}) under $\{m_i\}$ is the paramagnetic
solution
 $m_i= 0$. Furthermore, we expect that the $\{v_i\}$ that maximize the 
free entropy of eq.~(\ref{eq:spherical_full}) are \emph{homogeneous}, that is $\forall i, \, v_i = v$. 
The constraint of eq.~(\ref{eq:spherical_constraint_mv}) thus gives $v =
\sigma^2$.

We can compare the result of the resummation of simple cycles,
eq.~(\ref{eq:spherical_full})  with the exact results of
eq.~(\ref{eq:spherical_hight}) in the paramagnetic phase. For these two
results to agree, we need the generating function for simple
cycles to be related to the ${\cal R}$-transform of $\rho_D$ by:

\begin{align}\label{eq:correspondance_spherical_expectation}
  \EE \left[ \frac{1}{N} \sum_{p=1}^\infty \frac{\beta^p \sigma^{2p}}{2p} \sum_{\substack{i_1,\cdots,i_p \\ \text{pairwise distincts }}} J_{i_1 i_2} J_{i_2 i_3} \cdots J_{i_{p-1} i_p} J_{i_p i_{1}} \right] &= \frac12 \int_0^{\beta \sigma^2} {\cal R}_{\rho_D}(x) {\rm d} x,
\end{align}
in which the outer expectation is with respect to the distribution of $J$. 
In particular, an order-by-order comparison yields that the \emph{free cumulants} $\{c_p(\rho_D)\}_{p\in \bbN^\star}$ (see Appendix~\ref{sec:appendix_rmt} for their definition) must satisfy:
  
\begin{align}\label{eq:conjecture_freecumulants}
  \forall p \in \bbN^\star, \quad c_p(\rho_D) &= \lim_{N \to \infty} \EE \left[\frac{1}{N}\sum_{\substack{i_1,\cdots,i_p \\ \text{pairwise distincts }}} J_{i_1 i_2} J_{i_2 i_3} \cdots J_{i_{p-1} i_p} J_{i_p i_{1}} \right].
\end{align}
Using rigorous results of \cite{guionnet2005fourier}, we were able to prove a stronger version of eq.~(\ref{eq:conjecture_freecumulants}), namely 
convergence in $L^2$ norm, so we state it as a theorem:
\begin{theorem}\label{thm:free_cum}
   For a matrix $J \in {\cal S}_N$ generated by Model~\ref{model:sym_rot_inv}, one has for every $p \in \bbN^\star$:
   \begin{align*}
    \lim_{N \to \infty} \EE \left|\frac{1}{N} \sum_{\substack{i_1,\cdots,i_p \\ \text{pairwise distincts }}} J_{i_1 i_2} J_{i_2 i_3} \cdots J_{i_{p-1} i_p} J_{i_p i_{1}} -  c_p(\rho_D) \right|^2 & = 0.
   \end{align*}
\end{theorem}
We postpone the proof to Sec.~\ref{sec:diagrammatics}. We assume that we can invert the summation over $p$ and the $N \to \infty$ limit in eq.~(\ref{eq:spherical_full}),
so Theorem~\ref{thm:free_cum} implies that eq.~(\ref{eq:correspondance_spherical_expectation}) is true not only in expectation but that we can write:
\begin{align}\label{eq:correspondance_spherical}
  \lim_{N \to \infty} \frac{1}{N} \sum_{p=1}^\infty \frac{\beta^p \sigma^{2p}}{2p} \sum_{\substack{i_1,\cdots,i_p \\ \text{pairwise distincts }}} J_{i_1 i_2} J_{i_2 i_3} \cdots J_{i_{p-1} i_p} J_{i_p i_{1}} &= \frac12 \int_0^{\beta \sigma^2} {\cal R}_{\rho_D}(x) {\rm d} x,
\end{align}
in which the limit here means convergence in $L^2$ norm as $N \to \infty$. 
This is important, as it allows to ``resum'' the free entropy of eq.~(\ref{eq:spherical_full}), which is valid for a given instance of $J$.
As a final note, we can use the results of Sec.~\ref{subsubsec:direct_symmetric} to write our result in an alternative form (dropping $\smallO_N(1)$ terms):
\begin{align}\label{eq:correspondance_spherical_alternative}
  \frac{1}{N} \sum_{p=1}^\infty \frac{\beta^p \sigma^{2p}}{2p} \hspace{-0.5cm}\sum_{\substack{i_1,\cdots,i_p \\ \text{pairwise distincts }}}  \hspace{-0.5cm}J_{i_1 i_2}  \cdots J_{i_p i_{1}} &= \frac{1}{2} \inf_\gamma\left[\gamma \sigma^2 - \int \rho_D(\mathrm{d}\lambda) \log(\gamma-\beta \lambda)\right] - \frac{1+\log \sigma^2}{2}.
\end{align}

   \subsubsection{Stability of the paramagnetic phase}\label{subsubsec:stability_sym_spherical}

 We can check whether the paramagnetic solution is stable exactly up to the temperature $\beta =\beta_c$.
 Recall that in this model we do not optimize the free entropy simultaneously over $v$ and the $\{m_i \}$,
 because the norm $||\bx||_2^2 = \sigma^2 N$ is fixed, yielding the constraint $v = \sigma^2 - \frac{1}{N}\sum_{i} m_i^2$. 
 Solely as a function of the $\{m_i\}$, the free entropy therefore reads, up to $\smallO_N(1)$ terms:
\begin{align}\label{eq:spherical_para}
\Phi_{J}(\beta) &= \frac{1 + \log 2 \pi}{2} + \frac{1}{2} \log\left[\sigma^2 - \frac{1}{N} \sum_{i=1}^N m_i^2\right]  + \frac{\beta}{2N} \sum_{i \neq j} J_{ij} m_i m_j  + G_{\rho_D}\left(\beta \left[\sigma^2 - \frac{1}{N} \sum_{i=1}^N m_i^2\right]\right),
\end{align}
in which $G_{\rho_D}$ is the integrated $\mathcal{R}$-transform of $\rho_D$, see Appendix~\ref{sec:appendix_rmt} for its definition.
The Hessian of the \emph{extensive} free entropy $N \Phi_J$ at the paramagnetic solution $\bf{m} = 0$ is:
\begin{align}\label{eq:hessian_fentropy_spherical}
  N \left(\frac{\partial^2 \Phi_J}{\partial m_i \partial m_j}\right)_{m = 0} &= - \frac{\delta_{ij}}{\sigma^2}  \left[1 + \beta \sigma^2 \mathcal{R}_{\rho_D}(\beta \sigma^2)\right]+ \beta J_{ij} + \smallO_N\left(1\right).
\end{align}
The paramagnetic solution is stable as long as the Hessian $N \partial^2_m \Phi_{J}(\beta,m=0)$ is a negative matrix. 
This is true as long as $\beta < \beta_c = - \sigma^{-2} \mathcal{S}_{\rho_D}(\lambda_{\rm max})$, because at $\beta_c$ the spectrum of $N \left(\frac{\partial^2 \Phi_J}{\partial m_i \partial m_j}\right)_{m = 0} $ touches zero. For $\beta > \beta_c$ the Hessian is again negative, giving the impression of stability of the paramagnetic phase, however $\mathcal{S}_{\rho_D}^{-1}(-\beta\sigma^2)$
is evaluated in the non physical solution, so the solution has to be discarded.
Our Plefka expansion allows thus to compute the free entropy in the whole paramagnetic phase, coherently with the results of Sec.~\ref{subsubsec:direct_symmetric}.
More generically, as shown by Plefka \cite{plefka1982convergence} in the closely-related SK model, the Hessian of the free entropy with respect to $\{m_i\}$ is related to the inverse
susceptibility matrix of the system, and thus the non-inversibility of the Hessian implies a non-analyticity point of the free entropy.

\paragraph{Validity of the Plefka expansion and stability of the replica symmetric solution} 
We believe it is an open question to relate the range of validity $\beta_c$ of the Plefka expansion and the de Almeida-Thouless condition 
that characterizes the local stability of the replica symmetric solution (see \cite{de1978stability} for its original derivation, and \cite{kabashima2008inference,shinzato2008perceptron} for examples of its applications in inference problems).
The equivalence of these two conditions was shown in the seminal paper of Plefka \cite{plefka1982convergence} in the Sherrington-Kirpatrick model.
It is tedious but straightforward to generalize this conclusion to a model with Ising spins $x_i = \pm 1$ and a Hamiltonian given by eq.~\eqref{eq:hamiltonian_spherical} with a 
rotationally-invariant coupling matrix $J$ drawn from Model~\ref{model:sym_rot_inv}. 
However, investigating the relation between these two conditions in a general model appears to be an open problem,
and is beyond the scope of our paper.

\paragraph{The region of validity of the expansion and the free cumulant series} 
In this remark, we clarify some possible confusion about the radius of convergence of the free cumulant series and the domain of validity of the Plefka expansion.
We performed an expansion of $\Phi_J(\beta)$ close to $\beta = 0$, which implies that this expansion is thus valid in the region $(0,\beta_c)$, in which $\beta_c = - \sigma^{-2} {\cal S}(\lambda_{\rm max})$ is the first non-analyticity of $\Phi_J(\beta)$, 
see eqs.~\eqref{eq:spherical_hight}, \eqref{eq:spherical_lowt} and the discussion above. Note that there exists spectrums for which the function ${\cal R}_{\rho_D}(x)$ can 
be analytically extended beyond $x_c \equiv -{\cal S}_{\rho_D}(\lambda_{\rm max})$, as for instance Wigner's semi-circle law, for which ${\cal R}_{\rm s.c.}(x) = x$. Yet, one has to be careful that this does not imply 
that the free entropy $\Phi_J(\beta)$ is analytic beyond $\beta_c$, and thus even in this case our Plefka expansion is \emph{a priori} only valid up to $\beta = \beta_c$.

\subsection{Bipartite spherical model}\label{subsec:bipartite_spherical_model}

   In this section we consider $N,M \geq 1$. We let $\alpha > 0$, and we will take the limit (sometimes referred to as the \emph{thermodynamic limit}) in which $N,M \to \infty$ with a fixed ratio $M/N \to \alpha$.
   We let $\sigma_x,\sigma_h > 0$. Let us consider the following Hamiltonian, which is a function of two fields $\bh \in \bbR^M$ and $\bx \in \bbR^N$:
   \begin{align}\label{eq:hamiltonian_bipartite}
      H_{L}(\bh,\bx) &= - \bh^\intercal L \bx = -\sum_{\mu=1}^M \sum_{i=1}^N L_{\mu i} h_\mu x_i, \qquad \bh \in \mathbb{S}^{M-1}(\sigma_h \sqrt{M}), \quad \bx \in \mathbb{S}^{N-1}(\sigma_x \sqrt{N}).
   \end{align}
   The coupling matrix $L \in \bbR^{M \times N}$ is assumed to be
   drawn from Model~\ref{model:nsym_rot_inv}.

      \subsubsection{Direct free entropy computation}\label{subsubsec:direct_bipartite}

      The calculation for this bipartite case is very similar to the
      calculation performed in Sec.~\ref{subsubsec:direct_symmetric}, 
      although one can not always express the result
      as a well-known transform of the measure $\rho_D$ of
      Model~\ref{model:sym_rot_inv}. For all values of $\beta$, the result can be expressed as:
      \begin{align}
         \Phi_{L}(\beta) &\equiv \lim_{N \to \infty} \frac{1}{N} \log \int {\rm d} \bh \, \int {\rm d} \bx\, e^{\beta \bh^\intercal L \bx},\nonumber \\
           &= \frac{1+\alpha}{2} \log 2 \pi  +
           \frac{1}{2}\inf_{\gamma_h,\gamma_x}\left[\alpha \gamma_h
           \sigma_h^2 + \gamma_x \sigma_x^2 - (\alpha-1) \log \gamma_h
           - \int \rho_D(\mathrm{d}\lambda) \log(\gamma_x \gamma_h -
           \beta^2 \lambda)\right] \ , \label{eq:bipartite_all_temp}
      \end{align} 
      where $\rho_D$ is the asymptotic eigenvalue distribution of $L^TL$ (see the definition of Model~\ref{model:nsym_rot_inv}).

      \subsubsection{Plefka expansion}\label{subsubsec:plefka_bipartite_spherical}

      The Plefka expansion for this model is a straightforward generalization of Sec.~\ref{subsubsec:plefka_sym_spherical}. 
      We will fix the averages to be $\braket{h_\mu} = m^h_\mu$ and $\braket{x_i} = m^x_i$, and the second moments $\braket{h_\mu^2} = v^h_\mu + (m^h_\mu)^2$ and $\braket{x_i^2} = v^x_i + (m^x_i)^2$,
      again with the constraints $\sigma_h^2 = \frac{1}{M} \sum_{\mu=1}^M \left[v^h_\mu + (m^h_\mu)^2 \right]$ and $\sigma_x^2 = \frac{1}{N} \sum_{i=1}^N \left[v^x_i + (m^x_i)^2\right]$.
      In this problem, the $U$ operator of \cite{georges1991expand} at $\beta = 0$ is given by:
      \begin{align}\label{eq:U_bipartite}
         U(\beta=0,L) &= -\sum_{\mu,i} L_{\mu i} (h_\mu - m^h_\mu) (x_i - m^x_i).
      \end{align}
      Once again, as in Sec.~\ref{subsubsec:plefka_sym_spherical}, one can study all the diagrams that appear in the Plefka expansion.
      We show again the $L^2$ concentration of the simple cycles, and the negligibility of other strongly irreducible diagrams that can be constructed from the rectangular $L$ matrix.
      We state in more details these results for the bipartite case in Sec.~\ref{subsubsec:generalization_non_square}.
      We obtain the following result, a counterpart to eq.~(\ref{eq:spherical_full}) for this bipartite model:
      \begin{align}\label{eq:plefka_bipartite}
      \Phi_L(\beta) &= \frac{1+\alpha}{2} \left[1+ \log 2\pi \right] + \frac{\alpha}{2M} \sum_{\mu=1}^M \log v^h_\mu+ \frac{1}{2N} \sum_{i=1}^N \log v^x_i + \frac{\beta}{N} \sum_{\mu=1}^M \sum_{i=1}^N L_{\mu i} m^h_\mu m^x_i\\
      & + \frac{1}{N} \sum_{p=1}^\infty \frac{\beta^{2p} }{2p} \sum_{\substack{\mu_1,\cdots,\mu_p \\ \text{pairwise distincts }}} \sum_{\substack{i_1,\cdots,i_p \\ \text{pairwise  distincts }}} L_{\mu_1 i_1} L_{\mu_1 i_2} L_{\mu_2 i_2} \cdots L_{\mu_p i_p} L_{\mu_p i_1} \prod_{\alpha=1}^p v^h_{\mu_\alpha}v^x_{i_\alpha} + \smallO_N(1) \nonumber ,   
      \end{align}
      in which indices $\{\mu_l\}$ run from $1$ to $M$ and indices $\{i_l\}$ run from $1$ to $N$. 
      We make again an assumption of uniform variances at the maximum: $v^h_\mu = v^h, v^x_i = v^x$.
      Comparing to eq.~(\ref{eq:bipartite_all_temp}) in the paramagnetic
      $m^h_\mu, m^x_i = 0$ phase, we obtain the correspondence, similar
      to eq.~(\ref{eq:correspondance_spherical}) and valid \emph{a
        priori} for any given realization of $L$, in the high
      temperature phase:
      \begin{align}
         \frac{\alpha}{2} \log \sigma_h^2 &+ \frac{1}{2} \log \sigma_x^2 + \frac{1}{N} \sum_{p=1}^\infty \frac{\beta^{2p} \sigma_h^{2p}\sigma_x^{2p}}{2p} \sum_{\substack{\mu_1,\cdots,\mu_p \\ \text{pairwise distincts }}} \sum_{\substack{i_1,\cdots,i_p \\ \text{pairwise  distincts }}} L_{\mu_1 i_1} L_{\mu_1 i_2} L_{\mu_2 i_2} \cdots L_{\mu_p i_p} L_{\mu_p i_1} \nonumber \\   
      &= - \frac{1+\alpha}{2} + \frac{1}{2}\inf_{\gamma_h,\gamma_x}\left[\alpha \gamma_h \sigma_h^2 + \gamma_x \sigma_x^2 - (\alpha-1) \log \gamma_h -\int \,\rho_D(\mathrm{d}\lambda) \log(\gamma_h \gamma_x - \beta^2 \lambda)\right]. \label{eq:correspondance_bipartite}
      \end{align}

%% file: stat_models.tex
\section{Plefka expansion and Expectation Consistency approximations}\label{sec:stat_models}

In this section we perform Plefka expansions for generic models of pairwise interactions, both symmetric and bipartite.
In Sec.~\ref{subsec:ep_adatap} we recall some known facts on the Expectation Consistency (also called Expectation Propagation \cite{minka2001expectation}), adaTAP and VAMP approximations to compute 
the free entropy of such models.
In Sec.~\ref{subsec:plefka_sym_models} and Sec.~\ref{subsec:plefka_bipartite_models} we generalize the results of the Plefka expansions of Sec.~\ref{sec:spherical_bipartite} to these models and highlight the main 
differences and assumptions of our method. This yields a very precise and systematic justification of the TAP equations for rotationally invariant models.
We apply these results to retrieve the TAP free entropy of the Hopfield model, Compressed Sensing, as well
as different variations of high-dimensional inference models called Generalized Linear Models (GLMs). 
Sec.~\ref{subsec:plefka_replicas} is devoted to the study of a generic replicated system using these approximations, and the Plefka expansion. We show that they can be used 
in the celebrated replica method \cite{mezard1987spin} of theoretical physics, to compute the Gibbs free entropy of a generic pairwise inference model.

\subsection{Expectation Consistency, adaptive TAP, and Vector Approximate Message Passing approximations}\label{subsec:ep_adatap}

Expectation Consistency (EC) \cite{opper2005expectation,opper2005expectation2}, 
is an approximation scheme for a generic class of disordered systems that can also be applied to many inference problems. 
In this section we show how this scheme is derived and is closely related to the adaTAP approximation \cite{opper2001adaptive,opper2001tractable}, and 
the VAMP approximation \cite{rangan2017vector}. Let us shortly comment on the history of these methods.
The adaTAP scheme was developed and presented in $2001$ in \cite{opper2001tractable,opper2001adaptive}, 
and was discussed in details in the review \cite{opper2001advanced} for systems close to the SK model. 
The same year, Thomas Minka’s Expectation Propagation (EP) approach was presented \cite{minka2001expectation}. 
Opper and Winther used an alternative view of local-consistency approximations of the EP–type which they call Expectation Consistent (EC) approximations
in \cite{opper2005expectation,opper2005expectation2}, effectively rederiving their adaTAP scheme from this new point of view. 
The VAMP approach is more recent \cite{schniter2016vector}, and is again another EP approach for a different problem (compressed sensing) but it has the advantage that,
compared with other EP-like approaches~\cite{ccakmak2016self} it leads to a practical converging algorithm, and a rigorous treatement of its time evolution. 
The connection between these approaches and the Parisi-Potters formulation for inference problems \cite{jacquin2016resummed} was hinted several times for SK-like problems, 
see e.g. \cite{opper2016theory,ccakmak2019convergent}. 
We hope that our paper will help providing a unifying presentation of these works, generalizing them way beyond the SK model alone by leveraging random matrix theory.
We recall briefly the main arguments of these papers which are useful for our discussion.

\subsubsection{Expectation Consistency approximation}\label{subsubsec:expectation_consistency}

Consider a model in which the density of a vector $\bx \in \bbR^N$ is given by a probability distribution of the form:
\begin{align}\label{eq:def_model_ep}
	P(\bx) &= \frac{1}{Z} P_0(\bx) P_J(\bx).	
\end{align}
Such distributions typically appear in Bayesian approaches to inference problems. We will use the Bayesian language and denote $P_0$ as a \emph{prior} distribution on $\bx$, which will be typically factorized (all the components of $\bx$ are assumed to be independent under $P_0$) ;
The distribution $P_J$ is responsible for the interactions between the $\{x_i\}$. In this paper we are interested in pairwise interactions, which means that the $\log$ of $P_J$ is a quadratic form in the $\{x_i\}$ variables.
An example of such a model is the infinite-range Ising model of statistical physics at inverse temperature $\beta \geq 0$, with a binary prior and a quadratic interaction governed by a coupling 
matrix $J$. In this specific model, we have:
\begin{align}\label{eq:def_model_SK}
	P_0(\bx) &= \prod_{i=1}^N \left[\frac{1}{2 \cosh (\beta h_i)}\; \left(\delta(x_i-1) e^{-\beta h_i}+ \delta(x_i+1)e^{\beta h_i}\right)\right], \\
	P_J(\bx) &= \exp\left\{\frac{\beta}{2}\sum_{i,j} J_{ij} x_i x_j\right\},
\end{align}
for some $\{h_i\} \in \bbR^N$.
Our goal is to compute the large $N$ limit of the free entropy $\log Z$ in the model of eq.~(\ref{eq:def_model_ep}).
Each of the two distributions $P_0$ and $P_J$ allows for tractable computations of physical quantities (like averages), but
the difficulty arises when considering their product.
The idea behind EC is to simultaneously approximate $P_0$ and $P_J$ by a tractable family of distributions. For the sake 
of the presentation we will consider the family of Gaussian probability distributions, although this can be generalized to different families, 
see the general framework of \cite{opper2005expectation}. We define the first approximation as:
\begin{align}
	\mu_0(\bx) &\equiv \frac{1}{Z_0(\Gamma_0,\bm{\lambda}_0)} P_0(\bx) e^{-\frac{1}{2} \bx^\intercal \Gamma_0 \bx + \bm{\lambda_0}^\intercal \bx}. 
\end{align}
Here, the parameter $\Gamma_0$ is a symmetric positive matrix and $\bm{\lambda}_0$ is a vector.
We will denote $\braket{\cdot}_0$ the averages with respect to $\mu_0$.
We can write the trivial identity:
\begin{align*}
	Z &= Z \times \frac{Z_0(\Gamma_0,\bm{\lambda_0})}{Z_0(\Gamma_0,\bm{\lambda_0})} = Z_0(\Gamma_0,\bm{\lambda_0}) \braket{P_J(\bx) \, e^{\frac{1}{2} \bx^\intercal \Gamma_0 \bx - \bm{\lambda}_0^\intercal \bx}}_0.
\end{align*}
The idea of EC is to replace,  when one computes the average $\braket{P_J(\bx) \, e^{\frac{1}{2} \bx^\intercal \Gamma_0 \bx - \bm{\lambda}_0^\intercal \bx}}_0$, the distribution $\mu_0$ by an 
approximate Gaussian distribution, that we can write as:
\begin{align}
	\mu_S(\bx) &\equiv \frac{1}{Z_S}	e^{-\frac{1}{2} \bx^\intercal (\Gamma_J + \Gamma_0) \bx + (\bm{\lambda_0} + \bm{\lambda}_J)^\intercal \bx}.
\end{align}
Performing this replacement yields the \emph{expectation-consistency} approximation to the free entropy:
\begin{align}
	\log Z^{\rm EC}(\Gamma_0,\Gamma_J,\bm{\lambda}_0,\bm{\lambda}_J) &= \log \left[\int \mathrm{d}\bx P_0(\bx)e^{-\frac{1}{2} \bx^\intercal \Gamma_0 \bx + \bm{\lambda_0}^\intercal \bx} \right] + \log \left[\int \mathrm{d}\bx P_J(\bx)e^{-\frac{1}{2} \bx^\intercal \Gamma_J \bx + \bm{\lambda_J}^\intercal \bx} \right] \nonumber\\
	\label{eq:ZEC}
	&- \log \left[\int \mathrm{d}\bx \, e^{-\frac{1}{2} \bx^\intercal (\Gamma_0 + \Gamma_J) \bx + (\bm{\lambda_0}+\bm{\lambda}_J)^\intercal \bx}\right].
\end{align}
Note that all three parts of this free entropy are tractable. In order to symmetrize the result we can define a third measure:
\begin{align}
	\mu_J(\bx) &\equiv \frac{1}{Z_J(\Gamma_J,\bm{\lambda}_J)} P_J(\bx) e^{-\frac{1}{2} \bx^\intercal \Gamma_J \bx + \bm{\lambda_J}^\intercal \bx}. 
\end{align}
The final free entropy should not depend on the values of the parameters, so we expect that the best values for $\Gamma_0,\Gamma_J,\bm{\lambda}_0,\bm{\lambda}_J$
make $Z^{\rm EC}$ stationary. This is a strong hypothesis, and the reader can refer to \cite{opper2005expectation} for more details and justifications.
This yields the \emph{Expectation Consistency} conditions, giving their name to the procedure:
\begin{align}
\label{eq:ec_equations}
	\begin{cases}
	\braket{x_i}_0 &= \braket{x_i}_J = \braket{x_i}_S, \\	
	\braket{x_i x_j}_0 &= \braket{x_i x_j}_J = \braket{x_i x_j}_S.
	\end{cases}
\end{align}

\subsubsection{Adaptive TAP approximation}\label{subsubsec:adatap}

	The adaptive TAP approximation (or adaTAP) \cite{opper2001adaptive,opper2001tractable} provides an equivalent way to derive the free entropy of eq.~(\ref{eq:ZEC}) for models with pairwise interactions.
	Let us briefly sketch its derivation and the main arguments behind it. 
	We follow the formulation of \cite{huang2013adaptive} and we consider again the infinite-range Ising model of eq.~(\ref{eq:def_model_SK}).
	The extensive Gibbs free entropy $N\Phi = \log Z$ at fixed values of the magnetizations $m_i = \braket{x_i}$ and $v_{ij} = \braket{x_i x_j}_c$ can be written
	using Lagrange parameters: a vector $\bm{\lambda}$ and a symmetric matrix $\Gamma$.
	\begin{align}
		\Phi(\beta,\bm{m},\bv) &= \mathrm{extr}_{\bm{\lambda},\Gamma}\, \left[-\bm{\lambda}^\intercal \bm{m} + \frac{1}{2} \sum_{i,j} \Gamma_{ij}(v_{ij} + m_i m_j) + \log \int \mathrm{d}\bx \, P_0(\bx) \, e^{\frac{\beta}{2} \bx^\intercal J \bx - \frac{1}{2} \bx^\intercal \Gamma \bx+ \bm{\lambda}^\intercal \bx}\right].
	\end{align}
	The adaTAP approximation consists in writing:
	\begin{align}
	N\Phi(\beta,\bm{m},\bv) &= \Phi(0,\bm{m},\bv) + \int_0^\beta \mathrm{d}l \, \frac{\partial \Phi(l,\bm{m},\bv)}{\partial l}, \nonumber\\
	\label{eq:adatap}
					&\simeq \Phi(0,\bm{m},\bv) + \Phi_{g}(\beta,\bm{m},\bv) - \Phi_{g}(0,\bm{m},\bv).
	\end{align}
	In this expression, $\Phi_{g}(\beta,\bm{m},\bv)$ denotes the free entropy of the same system, but where the spins have a  Gaussian statistics. 
	The idea behind the adaTAP approximation is as follows. The derivative $\partial_l \Phi(l,\bm{m},\bv) = \frac{1}{2N} \sum_{ij} J_{ij}\braket{x_i x_j}$ is an expectation of a sum over a large number of terms; therefore  it is reasonable to assume that this expectation
	is the same as if the underlying variables were Gaussian. This assumption of adaTAP, although reasonable, is \emph{a priori} hard to justify more rigorously and systematically.
	It is important to notice that the free entropy \eqref{eq:adatap} of adaTAP is equivalent to the one derived using Expectation Consistency in eq.~\eqref{eq:ZEC}. 
	Indeed, using Lagrange parameters we can write the three terms of eq.~(\ref{eq:adatap}) as:
	\begin{align}
	N	\Phi^{\rm adaTAP}(\beta,\bm{m},\bv) &= \underset{\bm{\lambda_0},\Gamma_0}{\mathrm{extr}} \left[\log \left\{\int \mathrm{d}\bx P_0(\bx)\, e^{-\frac{1}{2} \bx^\intercal \Gamma_0 \bx + \bm{\lambda}_0^\intercal \bx} \right\}- \bm{\lambda_0}^\intercal \bm{m} + \frac{1}{2} \sum_{i,j} (\Gamma_0)_{ij}(v_{ij} + m_i m_j)\right] \nonumber \\
		& + \underset{\bm{\lambda_J},\Gamma_J}{\mathrm{extr}} \left[\log \left\{\int \mathrm{d}\bx P_J(\bx) \, e^{-\frac{1}{2} \bx^\intercal \Gamma_J \bx + \bm{\lambda}_J^\intercal \bx} \right\}- \bm{\lambda_J}^\intercal \bm{m} + \frac{1}{2} \sum_{i,j} (\Gamma_J)_{ij}(v_{ij} + m_i m_j)\right] \nonumber \\
		& - \underset{\bm{\lambda_S},\Gamma_S}{\mathrm{extr}} \left[\log \left\{\int \mathrm{d}\bx \, e^{-\frac{1}{2} \bx^\intercal \Gamma_S \bx + \bm{\lambda}_S^\intercal \bx} \right\}- \bm{\lambda_S}^\intercal \bm{m} + \frac{1}{2} \sum_{i,j} (\Gamma_S)_{ij}(v_{ij} + m_i m_j)\right].
	\end{align}
	Once written in this form, the extremization over $\bm{m}$ and $\bv$ of the free entropy 
	implies that $\Gamma_S = \Gamma_0 + \Gamma_J$ and $\bm{\lambda}_S = \bm{\lambda}_0 + \bm{\lambda}_J$.
	It is then clear that we found back $\log Z^{\rm EC}$ of eq.~(\ref{eq:ZEC}).

\subsubsection{Vector Approximate Message Passing approximation}\label{subsubsec:derivation_vamp}

The Vector Approximate Message Passing (VAMP) algorithm \cite{rangan2017vector} extends previous message-passing approaches 
like the GAMP algorithm \cite{rangan2011generalized} (that we will describe in more details in Sec.~\ref{subsec:gamp}) to a class of correlated interaction matrices, 
namely matrices that satisfy a right-rotation invariance property, similarly to Model~\ref{model:sym_rot_inv} and Model~\ref{model:nsym_rot_inv}. 
The algorithm itself can be derived in several ways (see \cite{rangan2017vector}). Here we briefly recall the use of belief-propagation equations on a ``duplicated''  factor graph and their Gaussian projection.
As we shall see, the Bethe free entropy, given as a function of the BP messages, is then equivalent to the expectation-consistency free entropy.
For simplicity, we consider again the problem of eq.~(\ref{eq:def_model_SK}) with a pairwise interaction involving a matrix $J$ following Model~\ref{model:sym_rot_inv}.

The idea behind VAMP is to consider two vector spin variables ${\bf x_1}$, with measure $P_0$ and ${\bf x_2}$ with measure~$P_J$, and to impose that they are equal. The partition function can be written using a trivial decomposition:
\begin{align}\label{eq:Z_VAMP}
Z \equiv e^{N \Phi} = \int_{\bbR^N} \  {\rm d} \bx_1 \  {\rm d} \bx_2 \ P_0(\bx_1) P_J(\bx_2) \ \delta(\bx_1-\bx_2).
\end{align}
This partition function can be represented as a ``duplicated'' factor graph involving two vector nodes, see Fig.~\ref{fig:vamp}.
\begin{figure}[t]
 \centering
\captionsetup{justification=centering}
\begin{tikzpicture}[scale=1.]
\node[draw,circle,line width = 0.5mm,minimum size=0.7cm,inner sep=0pt] (x1) at (0,0) {$\bx_1$};
\node[draw,circle,line width = 0.5mm,minimum size=0.7cm,inner sep=0pt] (x2) at (4,0) {$\bx_2$};
\node[draw,rectangle,line width = 0.5mm,minimum size=1cm,inner sep=0pt] (delta) at (2,0) {$\delta_{\bx_1,\bx_2}$};
\node[draw,rectangle,line width = 0.5mm,minimum size=1cm,inner sep=0pt] (P0) at (-2,0) {$P_0$};
\node[draw,rectangle,line width = 0.5mm,minimum size=1cm,inner sep=0pt] (PJ) at (6,0) {$P_J$};
\draw [line width = 0.3mm] (P0) to (x1);
\draw[line width = 0.3mm,->] (delta) -- node[above=1.5mm] {$m_0$} (x1);
\draw[line width = 0.3mm,->] (delta) -- node[above=1.5mm] {$m_J$} (x2);
\draw [line width = 0.3mm] (PJ) to (x2);
\end{tikzpicture}
\caption{Duplicated factor graph for the VAMP approximation. Circles represent vector nodes and squares factor nodes. 
We represent the two messages $m_0$ and $m_J$ in terms of which we can write the full BP equations.}\label{fig:vamp}
\end{figure}
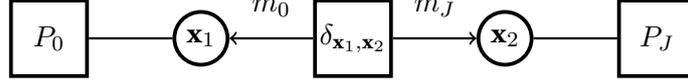
One then writes the BP equations for this problem using this factor graph representation (the 
 reader who is not familiar with BP equations and factor graph representations can consult \cite{mezard2009information}).
The BP equations can be written in terms of two ``messages'' $m_0({\bf x_2})$ and $m_J({\bf x_1})$. They can be obtained by looking at the stationarity conditions of the following ``Bethe free entropy'':
\begin{align}\label{eq:phi_Bethe_VAMP}
\Phi_{\rm Bethe} &\equiv \log\left( \int {\rm d} \bx \ P_0(\bx) m_J(\bx) \right) + \log\left( \int {\rm d} \bx \ P_J(\bx) m_0(\bx) \right) - \log\left( \int {\rm d} \bx \ m_0(\bx) m_J(\bx) \right).
\end{align}
As the factor graph has a tree structure, these BP equations are an exact representation of the original problem, but they are in general intractable.
In order to make the computation tractable one can make a Gaussian approximation, which is at the core of the VAMP algorithm:
the messages $m_0$ and $m_J$ on the factor graph of Fig.~\ref{fig:vamp} are assumed to be Gaussian, and thus are only characterized 
by their mean and their covariance. 
We can thus write:
\begin{align}\label{eq:Gaussian-Ansatz-vamp}
m_0(\bx) &\propto e^{-\frac12 \bx^\intercal \Gamma_0 \bx + \blambda_0^\intercal \bx}, \qquad 
m_J(\bx) \propto e^{-\frac12 \bx^\intercal \Gamma_J \bx + \blambda_J^\intercal \bx}.
\end{align}
Writing the BP update rule with this assumption yields the VAMP algorithm of \cite{rangan2017vector}. In the present case it reads:
\begin{subnumcases}{\label{eq:BPVamp}}
\Gamma_0^{t}=\langle \bx \bx^T\rangle_{\mu_0,c}^{t-1}-\Gamma_J^{t-1}, &\\
\blambda_0^{t}=(\Gamma_0^t+\Gamma_J^{t-1})\langle \bx \rangle_{\mu_0}^{t-1}-\blambda_J^{t-1} , & \\
\Gamma_J^{t}=\langle \bx \bx^T\rangle_{\mu_J,c}^{t}-\Gamma_0^{t}, &\\
\blambda_J^{t}=(\Gamma_0^{t}+\Gamma_J^{t})\langle \bx \rangle_{\mu_J}^{t}-\blambda_0^{t},& 
\end{subnumcases}
where the measures $\mu_0$ and $\mu_J$ are respectively
\begin{align}
\mu_0(\bx)&\propto P_0(\bx)m_J(\bx),\qquad \mu_J(\bx)\propto P_J(\bx)m_0(\bx).
\end{align}
Plugging the ansatz of eq.~(\ref{eq:Gaussian-Ansatz-vamp}) into eq.~(\ref{eq:phi_Bethe_VAMP}) immediately gives back the EC free entropy of eq.~(\ref{eq:ZEC}). 
And the BP equations of eq.~(\ref{eq:BPVamp}) are identical to the Expectation Consistency conditions of eq.~(\ref{eq:ec_equations}).

In the end, the three approximation schemes, Expectation Consistency, adaptive TAP and VAMP give the same expression for the free entropy.
However, an important advantage of the VAMP approach is that it ``naturally'' gives an iterative scheme to solve the fixed point equations because it was derived via the belief propagation equations. 
This iterative scheme turns the fixed point equations into an efficient algorithm, as noticed by 
\cite{rangan2017vector} and as we show later in Sec.~\ref{subsec:vamp}.

\subsection{Plefka expansion for models of symmetric pairwise interactions}\label{subsec:plefka_sym_models}

	\subsubsection{A symmetric model with generic priors}\label{subsubsec:sym_generic_prior}

		In this subsection we consider a generic model of $N$ ``spin'' variables $\bx = \{x_1,\cdots,x_N\} \in \bbR^N$. 
		They interact via a pairwise interaction and are subject to a possible external field, which is modeled by the following Hamiltonian:
		\begin{align}\label{eq:def_hamiltonian_ising}
			H_J(\bx) &= - \frac{1}{2} \sum_{i,j} J_{ij} x_i x_j + \sum_i h_i x_i.
		\end{align}
		We will consider random symmetric coupling matrices $\{J_{ij}\}$, generated from a rotationally 
		invariant ensemble satisfying Model~\ref{model:sym_rot_inv}.
		We assume that the variables $\{x_i\}$ have a prior distribution under which they are \emph{independent} variables, each with 
		a distribution $P_i$. For instance, this includes Ising (binary) spins by choosing $P_i = \frac{1}{2}(\delta_1 + \delta_{-1})$.
		At a given inverse temperature $\beta \geq 0$ and a fixed realization of the coupling matrix $J$ we define
		the Gibbs-Boltzmann distribution of the spins, and the partition function, as:
		\begin{align}
			\label{eq:def_Gibbs_ising}
			P_{\beta,J}(\mathrm{d}\bx) &\equiv \frac{1}{Z_{J}(\beta)} \prod_i P_i(\mathrm{d}x_i) \, \exp \left\{\frac{\beta}{2} \sum_{i,j} J_{ij} x_i x_j - \beta \sum_i h_i x_i\right\}, \\
			\label{eq:def_Z_ising}
			Z_{J}(\beta) &= \int \prod_i P_i(\mathrm{d}x_i) \, \exp \left\{\frac{\beta}{2} \sum_{i,j} J_{ij} x_i x_j - \beta \sum_i h_i x_i\right\}.
		\end{align}
		We will compute the large $N$ limit of the free entropy $\Phi_{J}(\beta) \equiv \frac{1}{N} \log Z_{J}(\beta)$ at fixed 
		values of the magnetizations $m_i = \braket{x_i}$ and variances $v_i = \braket{(x_i-m_i)^2}$ using the Plefka expansion.
		We will fix these variables using Lagrange multipliers $\{\lambda_i\}$ for the $\{m_i\}$ variables, and $\{\gamma_i\}$ for the $\{v_i\}$. 
		The goal of this section is to show how, and under which assumptions, the 
		calculation of Sec.~\ref{subsubsec:plefka_sym_spherical} can be generalized in this context. 
		Clearly the zeroth order term is different from the spherical case and it is given by:
		\begin{align}\label{eq:phi_order0_symmetric}
			\Phi_{J}(\beta = 0) &= \frac{1}{N} \sum_i \lambda_i m_i + \frac{1}{2N} \sum_i \gamma_i (v_i+m_i^2) + \frac{1}{N} \log \int\, \prod_{i} P_i(\mathrm{d}x_i) \, e^{-\frac{1}{2} \sum_i \gamma_i x_i^2 -\sum_i \lambda_i x_i }.
		\end{align} 
		As we underlined in Sec.~\ref{subsubsec:plefka_sym_spherical}, in the Georges-Yedidia method the Lagrange parameters are always considered at $\beta=0$.
		At order $1$ in $\beta$ we obtain at leading order: 
		\begin{align}
			\left(\frac{\partial \Phi_{J}}{\partial \beta}\right)_{\beta=0} &= - \frac{1}{N} \braket{H_J}_0 = \frac{1}{2N} \sum_{i,j} J_{ij} m_i m_j + \frac{1}{2N} \sum_i J_{ii} v_i - \frac{1}{N}\sum_i h_i m_i.
		\end{align}
		The operator $U$ of eq.~(\ref{eq:def_U}), defined in \cite{georges1991expand} and taken at $\beta=0$, is thus exactly the same as the one of eq.~(\ref{eq:U_spherical}). 
		Using this remark we can see that many of the results obtained in Sec.~\ref{subsubsec:plefka_sym_spherical} will apply to the present case. 
		For instance the second order term is identical and given in eq.~(\ref{eq:spherical_order2}). 
		We then conjecture, backed by our diagrammatic results in Sec.~\ref{sec:diagrammatics}, that the higher order terms are different from the spherical model only in terms which are sub-leading in $N$.
		 For instance at third order we obtain:
		\begin{align}\label{eq:phi_order3_sym_prior}
			\frac{1}{3!} \left(\frac{\partial^3 \Phi_{J}}{\partial \beta^3}\right)_{\beta=0} &= -\frac{1}{6N} \braket{U^3}_0 = \frac{1}{6N} \hspace{-0.5cm}\sum_{\substack{i_1,i_2,i_3 \\ \text{pairwise distincts}}} \hspace{-0.5cm}J_{i_1 i_2} J_{i_2 i_3} J_{i_3 i_1} v_{i_1} v_{i_2} v_{i_3} + \frac{1}{6N} \sum_{i \neq j} J_{ij}^3 \kappa^{(3)}_i \kappa^{(3)}_j.
		\end{align}
		In this equation, we denoted $\kappa^{(p)}_i$ the \emph{cumulant} 
		 of order $p$ of the distribution of $x_i$ at $\beta=0$.
		Note that the rotation invariance  of Model~\ref{model:sym_rot_inv} implies that if $i \neq j$, typically $J_{ij} \sim \frac{1}{\sqrt{N}}$. Therefore 
		a term like $\sum_{i \neq j} J_{ij}^3$ gives a negligible contribution to the free entropy. 
		We shall therefore assume that the second part of 
		the RHS of eq.~(\ref{eq:phi_order3_sym_prior}) is negligible as $N \to \infty$. This is correct provided that the possible correlations of the third order cumulants~$\kappa^{(3)}_i$
		with the matrix $J$ do not change the scaling of this term sufficiently to make it thermodynamically relevant (see Sec.~\ref{subsec:higher_order_symmetric} for more details on this particular point). 
		The first term corresponds to a simple cycle of order $3$ and is the same term that appeared in Sec.~\ref{subsubsec:plefka_sym_spherical}.

		\paragraph{Higher orders}
		We can carry on the computation of the derivatives $\frac{\partial^p}{\partial \beta^p}\Phi_{J}(\beta=0)$.
		We explain in more details
		in Sec.~\ref{subsec:higher_order_symmetric} under which precise results and assumptions we can leverage
		the results of Sec.~\ref{subsubsec:plefka_sym_spherical}, summarized in eq.~(\ref{eq:spherical_full}), to conjecture the following value of the
		free entropy at all orders of perturbation and at leading order in $N$:
		\begin{align}\label{eq:conjecture_phi_symmetric}
			\Phi_{J}(\beta) = \Phi_{J}(0) + \frac{\beta}{2N} \sum_{i,j} J_{ij} m_i m_j - \frac{\beta}{N} \sum_i h_i m_i + \frac{1}{N} \sum_{p=1}^\infty \frac{\beta^p}{2p} 
			\hspace{-0.5cm}\sum_{\substack{i_1,\cdots,i_p \\\text{pairwise distincts}}}\hspace{-0.5cm} J_{i_1 i_2} \cdots J_{i_p i_1} \prod_{\alpha=1}^p v_{i_\alpha}.
		\end{align}
		\paragraph{Homogeneous variances} For the remainder of this section we assume that the maximum of the free entropy of eq.~(\ref{eq:conjecture_phi_symmetric})
		is attained for variables $\{v_i\}$ such that $v_i = v$. This hypothesis can be argued as reasonable for many models, but we postpone this argumentation 
		to the applications of eq.~(\ref{eq:conjecture_phi_symmetric}) to specific models.
		We obtain a resummation of the Plefka free entropy using the correspondence of eq.~(\ref{eq:correspondance_spherical}):
		\begin{align}\label{eq:resummed_phi_symmetric}
			\Phi_{J}(\beta) = \Phi_{J}(0) + \frac{\beta}{2N} \sum_{i,j} J_{ij} m_i m_j - \frac{\beta}{N} \sum_i h_i m_i + \frac{1}{2} \int_0^{\beta v} \mathcal{R}_{\rho_D}(u) \, \mathrm{d}u.
		\end{align}
		Recall finally that $\Phi_{J}(0)$ is given by eq.~(\ref{eq:phi_order0_symmetric}). We were able to perform this expansion and its resummation almost only by
		applying our results on the spherical models of Sec.~\ref{sec:spherical_bipartite}. The study of the large-$N$ behavior of 
		diagrams made out of matrix elements of $J$, performed in Sec.~\ref{sec:diagrammatics}, proves to be of crucial importance both in the expansion and
		its resummation.
		As discussed in Sec.~\ref{subsubsec:plefka_sym_spherical},
		we expect this Plefka expansion of the free entropy to hold for $\beta < \beta_c$, in which $\beta_c \equiv -v^{-1} \ {\cal S}_{\rho_D}(\lambda_{\rm max})$.

	\subsubsection{Connection of the Plefka expansion to EC approximations}\label{subsubsec:connection_plefka_ep_adatap}

	Although not obvious at first, the result of the Plefka expansion in eq.~(\ref{eq:conjecture_phi_symmetric})
	provides a systematic and precise analysis of asymptotic exactness for rotationally invariant models of the Expectation Consistency approximations.
	The more straightforward way to see how the Plefka expansion relates to these approximations is 
	to start from the adaTAP approximation, see eq.~(\ref{eq:adatap}).
	In the language of the Plefka expansion, adaTAP amounts to assuming that at every order $p \geq 1$ of perturbation in $\beta$, 
	one can perform the calculation \emph{as if} the statistics of the variables were Gaussian. An equivalent formulation is that all the terms of order $p \geq 1$ in the low-$\beta$ expansion  
	of the free entropy
	should be the same for the model with a generic prior of Sec.~\ref{subsubsec:sym_generic_prior} and for the spherical model of Sec.~\ref{subsec:sym_spherical_model}.
	This statement, which generalizes the Parisi-Potters result of \cite{parisi1995mean}, is exactly what we argued in the calculation of Sec.~\ref{subsubsec:sym_generic_prior}, using the diagrammatic analysis of Sec.~\ref{sec:diagrammatics}.
	Therefore, the diagrammatic analysis `à la Plefka' provides a clear meaning to the EC approximations: the class of diagrams that are neglected in these approximations are explicited in Sec.~\ref{sec:diagrammatics}. 
	We believe that these diagrams are actually negligible in the large $N$ limit, so that the EC approximations are actually exact asymptotically for rotationally invariant models, in the high temperature phase as captured by the resummation of the Plefka expansion, which we summarized in Conjecture~\ref{conj:main_conj_approximation}. We believe that this asymptotic exactness extends beyond the high temperature phase to any model in the replica symmetric phase. 
	The diagrammatic analysis provides a route to proving this statement rigorously.

	\subsubsection{Application to the Hopfield model}\label{subsubsec:plefka_hopfield}

In the Hopfield model \cite{hopfield1982neural} we consider binary spins $\bx \in \{\pm 1\}^N$ and the coupling matrix $J$ is constructed out of $P$ \emph{patterns}, which
are spin configurations $\xi^l \in \{\pm 1\}^N$, for $l \in \{1,\cdots,P\}$. The coupling constants are defined as:
\begin{align}
	\label{eq:def_coupling_hopfield}
	\begin{cases}
	J_{ij} &= \frac{1}{N} \sum_{l=1}^P \xi^l_i \xi^l_j \qquad (i \neq j), \\
	J_{ii} &= 0,
	\end{cases}
\end{align}
and we assume that the $\{\xi^l_i\}$ are i.i.d.\ variables with equal probability in $\{\pm 1\}$, so that $\EE J_{ij} = 0$ and $\EE J_{ij}^2 = P/N^2$. We study this system
in the limit in which both $P,N \to \infty$ with a fixed ratio $P/N \to \alpha$.
The derivation of the TAP free energy for these models has been performed in \cite{nakanishi1997mean,mezard2017mean} via the Plefka expansion, and via the cavity method in \cite{mezard1987spin}.
If the random matrix ensemble of eq.~(\ref{eq:def_coupling_hopfield}) is \emph{a priori} not rotationally invariant, one can show that since the variables $\{\xi^l_i\}$ are i.i.d., 
only the first and second moment of their distributions will contribute to the thermodynamic limit of the free entropy, so that we can assume 
that they are actually standard centered Gaussian variables without changing the free entropy.
This is strengthened by the classical results of \cite{marchenko1967distribution}, who only need to consider i.i.d.\ variables $\{\xi^l_\mu\}$
to obtain that the spectral law of the covariance matrix written in eq.~(\ref{eq:def_coupling_hopfield}) converges weakly to the celebrated Marchenko-Pastur distribution.

The ensemble of eq.~(\ref{eq:def_coupling_hopfield}) is thus for our purposes essentially a Wishart matrix model in which the diagonal has been removed. 
It is then a known result of random matrix theory (see for instance \cite{tulino2004random}) that its asymptotic $\mathcal{R}$-transform reads:
\begin{align}\label{eq:R_transform_hopfield}
	\mathcal{R}_J(x) &= \frac{\alpha}{1-x} - \alpha = \alpha \frac{x}{1-x}.
\end{align}
The term $-\alpha$ in the first equality accounts for the ``removal'' of the diagonal of the Wishart matrix.
Let us apply eq.~(\ref{eq:conjecture_phi_symmetric}) and eq.~(\ref{eq:resummed_phi_symmetric}) for this model. 
At $\beta=0$, the free entropy is given by eq.~(\ref{eq:phi_order0_symmetric}). For an Ising prior on the spins, it reads:
\begin{align}
	\label{eq:ising_order0}
	\Phi_{J}(\beta=0) &= - \frac{1}{N} \sum_i \left[\frac{1+m_i}{2} \ln \frac{1+m_i}{2} + \frac{1-m_i}{2} \ln \frac{1-m_i}{2}\right]. 
\end{align}
Note that because $x_i^2 = 1$ the variances $v_i$ of the variables $x_i$ are fixed by the magnetizations $m_i$ by the relation $v_i = 1 - m_i^2$.
Eq.~(\ref{eq:conjecture_phi_symmetric}) becomes:
\begin{align}
	\label{eq:plefka_hopfield}
	\Phi_{J}(\beta) &=  \Phi_{J}(0) +\frac{\beta}{2} \sum_{i,j} J_{ij} m_i m_j + \frac{1}{N} \sum_{p=1}^\infty  \frac{\beta^p}{2p}\hspace{-0.5cm} \sum_{\substack{i_1,\cdots,i_p \\ \text{pairwise distincts}}}\hspace{-0.5cm} J_{i_1 i_2} \cdots J_{i_p i_1} \prod_{\alpha=1}^p (1-m_{i_\alpha}^2) + \smallO_N(1).
\end{align}
Substituting the squared means $m_i^2$ by the spin glass order parameter $q \equiv (1/N) \sum_i m_i^2$ in eq.~(\ref{eq:plefka_hopfield}) 
and using the resummed form of eq.~(\ref{eq:resummed_phi_symmetric}), we obtain:
\begin{align}
	\label{eq:plefka_hopfield_resummed}
	\Phi_{J}(\beta) &=  \Phi_{J}(0) + \frac{\beta}{2} \sum_{i,j} J_{ij} m_i m_j + \frac{1}{2} \int_0^{\beta(1-q)} \mathcal{R}_{\rho_D}(u) \,\mathrm{d}u + \smallO_N(1).
\end{align}
Starting from eq.~(\ref{eq:plefka_hopfield_resummed}) and eq.~(\ref{eq:R_transform_hopfield}), we reach the final form for the free entropy:
\begin{align} \label{eq:plefka_hopfield_resummed2}
	\Phi_{J}(\beta) = - \frac{1}{N} \sum_i &\left[\frac{1+m_i}{2} \ln \frac{1+m_i}{2} + \frac{1-m_i}{2} \ln \frac{1-m_i}{2}\right] + \frac{\beta}{2} \sum_{i,j} J_{ij} m_i m_j \\
	& - \frac{\alpha \beta (1-q)}{2} - \frac{\alpha}{2} \ln \left[1-\beta(1-q)\right]  + \smallO_N(1)\nonumber.
\end{align}
Maximizing it over the continuous set of magnetizations $\{m_i\}$ yields the TAP equations:
\begin{align}\label{eq:TAP_hopfield}
 \beta^{-1}\tanh^{-1}(m_i) &= \sum_{j (\neq i)} J_{ij} m_j - \alpha \beta \frac{1-q}{1-\beta (1-q)} m_i. 	
\end{align}
This is in agreement with the findings of \cite{mezard1987spin,nakanishi1997mean, mezard2017mean}.
However, our framework and results allowed us to treat this kind of model in a very generic way.

	\subsection{The Replica approach}\label{subsec:plefka_replicas}
		
	Interestingly, the same approaches can be used to study the partition function and derive the free entropy of rotationally invariant models using the replica approach. 
	The essence of the replica method is well known \cite{mezard1987spin}: one studies the $n$-th moment of the partition function $Z$ 
	by introducing $n$ 'replicas' of each original variable, and one then studies the average of $Z^n$ in the limit $n\to 0$, 
	which allows to reconstruct the average of $\log Z$, and therefore the free entropy.
	We shall illustrate in this section how the three approximation schemes (EC, adaTAP and VAMP) can be used to derive the replica free entropy, 
	and how the high-temperature Plefka expansion justifies these approximations.

	Let us consider again the generic model of eq.~(\ref{eq:def_hamiltonian_ising}), with $h_i = 0$ for simplicity.
	The $n$-th moment of the partition function reads ($a,b$ indices always denote replica indices running from $1$ to $n$):
	\begin{align}
			Z_{J}^n(\beta) &= \left[\int \prod_{a}\prod_i P_i(\mathrm{d}x^a_i)\right] \, \exp \left\{\frac{\beta}{2} \sum_{a} \sum_{i,j} J_{ij} x^a_i x^a_j\right\}.
\label{replicatedZ_def}
	\end{align}

This ``replicated partition function'' should then be averaged over the disorder. 
Here, we deal with a $J$ matrix generated from Model~\ref{model:sym_rot_inv},
and the disorder average means an average over the orthogonal matrix $O$ in the decomposition $J=ODO^T$ (keeping the eigenvalues $D$ fixed). 
Let us see how it can be analyzed using the four (equivalent) approximation schemes. 
Interestingly the average over $O$ is in fact not needed in these approaches, as they show that the repliacted partition function is ``self-averaging'',
meaning that it gives the same free entropy density for almost all realizations of $J$.
The order parameter we will fixe here is the $n\times n$ symmetric matrix $Q$ with elements $Q_{ab}=\frac{1}{N}\sum_i \braket{x^a_i x^b_i}_c $, called the \emph{overlap} matrix
in the statistical physics language.

\subsubsection{Expectation-Consistency approximation} 

As we have seen in Sec.~\ref{subsubsec:expectation_consistency}, in the Expectation-Consistency (EC) approximation, 
we decompose the replicated free entropy $\Phi = \frac{1}{N} \log Z_J^n(\beta)$ as a function of three auxiliary free entropies:
\begin{align}\label{eq:ec_replica_1}
	\Phi = \underset{\Gamma_0, \Gamma_J}{\rm extr} &\left\{\frac{1}{N} \log \left[\int \prod_{a,i} P_i({\rm d} x_i^a) e^{-\frac{1}{2} \sum_{i} \bx_i^\intercal \Gamma_0 \bx_i}\right] + \frac{1}{N} \log \left[\int \prod_{a,i} {\rm d} x_i^a \, e^{\frac{\beta}{2} \sum_{i,j,a} J_{ij} x^a_i x^a_j -\frac{1}{2} \sum_{i}\bx_i^\intercal \Gamma_J \bx_i}\right] \right. \nonumber \\
	& \left. \, - \frac{1}{N} \log \left[\int \prod_{a,i} {\rm d} x_i^a e^{-\frac{1}{2} \sum_{i} \bx_i^\intercal (\Gamma_0 + \Gamma_J) \bx_i}\right] \right\}.
\end{align}
The $\Gamma_0, \Gamma_J$ matrices are symmetric $n \times n$ matrices, and $\bx_i = (x^a_i)_{a=1}^n$.
Indeed, as we will only fix the $Q_{ab}$, we expect the first moments to vanish and the second moments to be uniform in space (but not in replica space).
The extremization over $\Gamma_0, \Gamma_J$ yields the Expectation-Consistency equations, as described in eq.~\eqref{eq:ec_equations}. In particular, the average $\frac{1}{N} \sum_i \braket{x^a_i x^b_i}$ 
is constrained to be the same under the three different measures appearing in eq.~\eqref{eq:ec_replica_1}.
As we want to impose this average to be $Q_{ab}$, 
we can introduce a Lagrange parameter $\Gamma_S$ to fix this average for instance in the third part of the r.h.s.\ of eq.~\eqref{eq:ec_replica_1}. This third term becomes
\begin{align*}
	- \, \underset{\Gamma_S}{\rm extr} \left\{\frac{1}{2} \mathrm{Tr}\ (\Gamma_S Q) + \frac{1}{N} \log \left[\int \prod_{a,i} {\rm d} x_i^a e^{-\frac{1}{2} \sum_{i} \bx_i^\intercal (\Gamma_0 + \Gamma_J + \Gamma_S) \bx_i}\right]  \right\}.
\end{align*}
Changing $\Gamma_S \to \Gamma_S + \Gamma_0 + \Gamma_J$ gives:
\begin{align*}
	&\frac{1}{2} \mathrm{Tr} \left[Q(\Gamma_0 + \Gamma_J)\right] - \underset{\Gamma_S}{\rm extr} \left\{ \frac{1}{2} \mathrm{Tr}\, \left[\Gamma_S Q\right] + \frac{1}{N} \log \left[\int \prod_{a,i} {\rm d} x_i^a e^{-\frac{1}{2} \sum_{i} \bx_i^\intercal  \Gamma_S \bx_i}\right]  \right\} \\
	&=\frac{1}{2} \mathrm{Tr} \left[Q(\Gamma_0 + \Gamma_J)\right]  - \frac{n}{2} (1 + \log 2 \pi) - \frac{1}{2} \mathrm{Tr} \log Q.
\end{align*}
In the end, performing explicitely the Gaussian integration in the basis of eigenvectors of $J$ in eq.~\eqref{eq:ec_replica_1}, we obtain the free entropy at fixed overlap $Q$
 (recall that $\rho_D$ is the asymptotic eigenvalue density of $J$):
\begin{align}\label{eq:ec_replica_2}
	\Phi(Q) &= \underset{\Gamma_0}{\rm extr} \left\{\frac{1}{2} \mathrm{Tr}(\Gamma_0 Q) + \frac{1}{N} \sum_i \log \left[\int \prod_{a} P_i({\rm d} x^a) e^{-\frac{1}{2} \bx^\intercal \Gamma_0 \bx}\right] \right\} \nonumber \\
	& + \underset{\Gamma_J}{\rm extr} \left\{\frac{1}{2} \mathrm{Tr}(\Gamma_J Q) -\frac{1}{2} \int \rho_D({\rm d}\lambda) \mathrm{Tr} \log (\Gamma_J - \beta \lambda {\rm I}_n) \right\} - \frac{1}{2} \mathrm{Tr} \log Q - \frac{n}{2}.
\end{align}
The total free entropy is obtained simply as $\underset{Q}{\mathrm{extr}} \, \Phi(Q)$. 
Assuming an ultrametric structure in replica space \cite{mezard1987spin} 
(for instance in the replica symmetric case the matrix elements of $Q$ take only two values, 
one on the diagonal and one out of the diagonal, and the same holds for $\Lambda$ and $R$), 
one can perform an explicit computation. For random orthogonal models, this expression was first derived
using Plefka expansion in \cite{marinari1994replica, marinari1994replica2}, and
written in exactly the same form as ours in \cite{cherrier2003role}.

\subsubsection{The adaTAP approximation}
Let us now describe the adaTAP approach to this problem.
We introduce Lagrange parameters $\Lambda_{ab}$ which act as external fields, fixing the values of the order parameter $Q_{ab}$.
The matrix $\Lambda$ is a $n\times n$ symmetric matrix. The free entropy at fixed overlap matrix $\Phi(Q)$ is expressed as:
\begin{align}
\Phi(Q) &=\underset{\Lambda}{\mathrm{extr}}\left[\frac{1}{2}\sum_{a,b}\Lambda_{ab}Q_{ab}+\frac{1}{N}\log
\left[\int \prod_{a,i} P_i(\mathrm{d}x^a_i)\right] \, \exp \left\{\frac{\beta}{2} \sum_{a,i,j} J_{ij} x^a_i x^a_j-\frac{1}{2}\sum_{a,b,i}\Lambda_{ab} x_i^a x_i^b\right\}
\right].
\end{align}
The total free entropy will thus be obtained as $\underset{Q}{\mathrm{extr}} \, \Phi(Q)$.
We use the adaTAP approximation, which gives:
\begin{align}
	\Phi(Q)=\Phi(Q,\beta=0)+\Phi_G(Q,\beta)-\Phi_G(Q,\beta=0),
\end{align}
where $\Phi_G$ is the Gibbs free energy for a model with Gaussian
spins, at fixed overlap between spins.
The three pieces of $\Phi$ can be written as:
\begin{align}
\Phi(Q,\beta=0) &= \underset{\Lambda}{\mathrm{extr}}\left[\frac{1}{2}\mathrm{Tr}\,(\Lambda Q) +
\frac{1}{N}\sum_i\log\left(\int \prod_a P_i(\mathrm{d}x^a) e^{-\frac{1}{2}\sum_{a,b}\Lambda_{ab} x^a x^b}\right)\right], \\
\Phi_G(Q,\beta) &= \underset{\Lambda}{\mathrm{extr}}\left[\frac{1}{2}\mathrm{Tr} \, (\Lambda Q) +\frac{1}{N}\sum_{\lambda \in {\rm Sp}(J)}\log \int\prod_a {\rm d} x^a e^{-\frac{1}{2}\sum_{a,b}\Lambda_{ab}x^ax^b+\frac{\beta}{2}\lambda\sum_a (x^a)^2}\right],\\
\Phi_G(Q,\beta=0) &= \underset{\Lambda}{\mathrm{extr}}\left[\frac{1}{2}\mathrm{Tr} \, (\Lambda Q) +\log \int\prod_a {\rm d} x^a e^{-\frac{1}{2}\sum_{a,b}\Lambda_{ab}x^ax^b}\right].
\end{align}
In $\Phi_G(Q,\beta)$ we carry the integral over $s_i^a$ in an (orthonormal) basis of eigenvectors of $\{J_{ij}\}$.
$\Phi_G(Q,\beta=0)$ is easily evaluated: the extremization gives
$\Lambda = Q^{-1}$, and therefore
\begin{align}
\Phi_G(Q,\beta=0) &= n\frac{1 + \log 2 \pi}{2}+\frac{1}{2}\Tr\log Q.
\end{align}
We change notation and denote by $R$ the matrix $\Lambda$ that appears in $\Phi_G(Q,\beta)$. This gives 
finally:
\begin{align}
\Phi(Q)=\underset{\Lambda,R}{\mathrm{extr}} \, \Phi(Q,\Lambda,R),
\end{align}

with
\begin{align}
\Phi(Q,\Lambda,R)&=\frac{1}{2}\Tr(\Lambda Q) +\frac{1}{N}\sum_i \log\left(\int
    \prod_a P_i(\mathrm{d}x^a) e^{-\frac{1}{2}\sum_{a,b}\Lambda_{ab} x^a
x^b}\right)\nonumber\\
&+\frac{1}{2}\Tr(R Q)
  - \frac{1}{2} \int \rho_D({\rm d}\lambda) \Tr\log(R-\beta \lambda {\rm I}_n)
  -\frac{1}{2}\Tr\log Q -\frac{n}{2}.
\end{align}
We found back exactly the EC free entropy of eq.~\eqref{eq:ec_replica_2}.

\subsubsection{The VAMP approach}
Following eq.~(\ref{eq:Z_VAMP}), the replicated partition function can be written as a ``duplicated'' integral:

\begin{align} 
Z_{J}^n(\beta) = \int_{\bbR^{nN} \times \bbR^{nN}} P_0(\rm{d}\bx) P_J(\rm{d} \by) \delta(\bx - \by),
\end{align}
where
\begin{align}
P_0({\rm d}\bx) &= \prod_i P_i(\mathrm{d}x^a_i),\\
P_J({\rm d}\by) &= \exp \left\{\frac{\beta}{2} \sum_{a} \sum_{i,j} J_{ij} y^a_i y^a_j\right\} \,{\rm d}\by.
\end{align}
We now write the Bethe free entropy $\Phi_{\rm Bethe}$ as in eq.~(\ref{eq:phi_Bethe_VAMP}), and project it onto the space of Gaussian messages 
which are assumed to be proportional to identity in space, but with an
arbitrary replica structure:
\begin{align}
m_0(\bx) &\propto \exp\left(-\frac{1}{2} \sum_{a,b} A^0_{ab}\sum_i
  x_i^a x_i^b\right),\\
m_J(\bx) &\propto \exp\left(-\frac{1}{2} \sum_{a,b} A^J_{ab}\sum_i
  x_i^a x_i^b\right).
\end{align}
We did not include first moments because we expect all the first
moments to vanish in the replicated system.
We can now write the
stationarity equations of $\Phi_{\rm Bethe}$ with respect to $A^0$ and $A^J$: 
\begin{align}
\left([A^0+A^J]^{-1}\right)_{ab}&=-2\frac{\partial}{\partial A^J_{ab}} \frac{1}{N}\sum_i\log\left( \int \prod_a P_i(\mathrm{d}x^a) \exp\left(-\frac{1}{2} \sum_{a,b} A^J_{ab}
  x^a x^b\right) \right), \\
\left([A^0+A^J]^{-1}\right)_{ab}&= \int \rho_D({\rm d}\lambda) \left([A^0-\beta\lambda]^{-1}\right)_{ab}.
\label{eqQA}
\end{align}
It is easy to derive the replicated free entropy $\Phi$. One finds
\begin{align}
\Phi &= \underset{Q,A^0,A^J}{\mathrm{extr}} \, [\Phi(Q,A^0,A^J)],
\end{align}
where
\begin{align}
\Phi(Q,A^0,A^J)&= \frac{1}{N}\sum_i \log\left[ \int \prod_a P_i(\mathrm{d}x^a) \exp\left(-\frac{1}{2} \sum_{a,b} A^J_{ab}
  x^a x^b\right) \right]
-\frac{1}{2}\int \rho_D({\rm d}\lambda) \Tr\log (A^0-\beta\lambda {\rm I}_n)
\nonumber
\\
&
- \frac{1}{2}\Tr\log(Q)
+\frac{1}{2} \Tr(A^J Q) +\frac{1}{2}\Tr(A^0 Q)-\frac{n}{2}.
\label{PhiQrep}
\end{align}
This is again exactly the same result as the replicated free entropy found with EC or adaTAP in the previous sections.

\subsubsection{Plefka expansion}

As before, we fix only the overlap $Q_{ab}$, via Lagrange parameters $\gamma^{ab}$.
	Let us denote $\Phi(Q,\beta)$ the corresponding free entropy. At $\beta=0$ the replicas are not coupled so that we obtain
	\begin{align}\label{eq:phi_n_order0}
		\Phi(Q,\beta = 0) &= \frac{1}{2} \sum_{a,b} \gamma^{ab} Q_{ab} + \frac{1}{N} \sum_i \log \int\, \prod_{a=1}^n P_i(\mathrm{d}x^a) \, e^{-\frac{1}{2} \sum_{a,b} \gamma^{ab} x^a x^b}.
	\end{align} 
	One can then compute the order $1$ perturbation and the $U$ operator of Georges-Yedidia:
	\begin{align}
		\left(\frac{\partial \Phi}{\partial \beta}\right)_{\beta=0} &= \frac{1}{2N} \sum_{a,i} J_{ii} q_{aa}, \\
		U(\beta=0) &= - \frac{1}{2} \sum_{a}\sum_{i \neq j} J_{ij} x^a_i x^a_j. 	
	\end{align}
	Note that at $\beta=0$ we have $\braket{x^a_i}_0 = 0$. We obtain the order $2$ correction in the same way as before:
	\begin{align}
		\frac{1}{2} \left(\frac{\partial^2 \Phi}{\partial \beta^2}\right)_{\beta=0} &= \frac{1}{8N} \sum_{i \neq j} J_{ij}^2 \mathrm{Tr}\,[Q^2].
	\end{align}
	Here the trace is taken in the replica space. One can continue the Plefka expansion at any order, and one obtains very similar results to
	the non-replicated free entropy, simply the product of variances is replaced by traces of the overlap matrix $Q$. 
	Indeed, the diagrams constructed from the matrix indices $\{J_{ij}\}$ that appear in the replicated free entropy are exactly the same as in the non-replicated 
	case, so that all the results of Sec.~\ref{sec:diagrammatics} stay valid for this replicated calculation.
	In the end, the resummation yields the \emph{single-graph} replicated free entropy as a function of the overlap matrix $Q$:
	\begin{align}
		\Phi(Q,\beta) &= \Phi(Q,\beta=0) + \sum_{p=1}^\infty \frac{\beta^p}{2p} c_p(\rho_D) \mathrm{Tr}\,\left[Q^p\right]\ .
	\end{align}
The series in the second part of this equation is nothing but $\mathrm{Tr}\,\left[G_{\rho_D}(Q)\right]$ (see Appendix~\ref{sec:appendix_rmt} for the definition of this function). Recalling the expression (\ref{eq:phi_n_order0}) for the first, $\beta=0$, piece, we get finally:
	\begin{align}\label{eq:phi_replica_plefka}
		\Phi(Q,\beta) &= \underset{\gamma}{\rm extr} \left\{ \frac{1}{2} \mathrm{Tr} (\gamma Q) + \frac{1}{N} \sum_i \log \int\, \prod_{a=1}^n P_i(\mathrm{d}x^a) \, e^{-\frac{1}{2} \sum_{a,b} \gamma^{ab} x^a x^b} + \mathrm{Tr}\,\left[G_{\rho_D}(Q)\right] \right\},
	\end{align}
	and recall that the $\beta=0$ term is given by eq.~(\ref{eq:phi_n_order0}) and that the $G_{\rho_D}$ function is the integrated $\mathcal{R}$-transform 
	defined in Appendix~\ref{sec:appendix_rmt}, which verifies:
	\begin{align}
		G_{\rho_D}(z) &= \underset{x}{\rm extr} \left[\frac{z x}{2} - \frac{1}{2}\int \rho_D({\rm d}\lambda) \log(\lambda-x) \right] - \frac{1 + \log x}{2}.	
	\end{align}
	Thus we recognize once again in eq.~\eqref{eq:phi_replica_plefka} the expression obtained via EC and adaTAP, see eq.~\eqref{eq:ec_replica_2}.

\subsection{Plefka expansion for models of bipartite pairwise interactions}\label{subsec:plefka_bipartite_models}

	\subsubsection{A bipartite model with generic priors}\label{subsubsec:bipartite_generic_prior}

We consider here another generic model, which is an extension of the bipartite model studied in Sec.~\ref{subsec:bipartite_spherical_model}, the difference being that
we assume a generic prior on the variables rather than a spherical constraint.
As we will see, this model is closely related to the symmetric model of Sec.~\ref{subsubsec:sym_generic_prior}. 
Let $M,N \geq 1$. We consider two types of variables: a vector $\bx \in \bbR^N$ and a vector $\bh \in \bbR^M$. 
These two fields are assumed to follow prior distributions under which they are independent and the distribution of all their components decouples.
 For instance, the prior $P_X$ on $\bx$ can be written as $P_X(\mathrm{d}\bx) = \prod_i P_{i}(\mathrm{d}x_i)$. The two fields $\bh,\bx$ interact via the following Hamiltonian:
\begin{align}
	H_F(\bh,\bx) &= - \sum_{\mu,i}  F_{\mu i} h_\mu x_i.	
\end{align}
As in Sec.~\ref{subsec:bipartite_spherical_model}, greek indices $\mu,\nu$ will always run from $1$ to $M$ while latin indices $i,j$ run from $1$ to $N$.
We assume that the coupling matrix $F \in \bbR^{M \times N}$ satisfies the rotation invariance property described in Model~\ref{model:nsym_rot_inv}.
For a fixed $\beta \geq 0$ and a realization of $F$ we define the 
Gibbs-Boltzmann distribution and the partition function:
\begin{align}
	P_{\beta,F}(\bh,\bx) &\equiv \frac{1}{Z_{\beta,F}} \prod_\mu P_\mu(\mathrm{d}h_\mu)\, \prod_i P_i(\mathrm{d}x_i) \, \exp \left\{\beta \sum_{\mu,i} F_{\mu i} h_\mu x_i\right\}, \\
	Z_{\beta,F} &= \int \prod_\mu P_\mu(\mathrm{d}h_\mu) \, \prod_i P_i(\mathrm{d}x_i) \, \exp \left\{\beta \sum_{\mu,i} F_{\mu i} h_\mu x_i\right\}.
\end{align}
As in the previous section we will compute the large $N$ limit of the free entropy $\Phi_{F}(\beta) \equiv \frac{1}{N} \log Z_{\beta,F}$. 
We look at this problem in the thermodynamic limit, with $N,M \to \infty$ and a fixed ratio $M/N \to \alpha > 0$.
We constraint the first and second moments of $\{x_i\}$ and $\{h_\mu\}$ under the Gibbs measure to be $\braket{x_i} = m^x_i$, $\braket{h_\mu} = m^h_\mu$, 
$\braket{(x_i-m_i^x)^2} = v^x_i$, $\braket{(h_\mu-m_\mu^h)^2} = v^h_\mu$. The Lagrange multipliers introduced to enforce these conditions will be denoted 
$\lambda^x_i$, $\lambda^h_\mu$ for the first moments and $\gamma^x_\mu$, $\gamma^h_\mu$ for the second moments.
At order $0$, we obtain:
\begin{align}\label{eq:phi_order0_bipartite_prior}
	\Phi_{F}(\beta=0) &= \frac{1}{N} \sum_i \left[\lambda^x_i m^x_i + \frac{1}{2} \gamma^x_i (v^x_i + (m^x_i)^2)\right] + \frac{1}{N} \sum_\mu \left[\lambda^h_\mu m^h_\mu + \frac{1}{2} \gamma^h_\mu (v^h_\mu + (m^h_\mu)^2)\right] \nonumber  \\
	&+ \frac{1}{N} \log \left[\int \prod_\mu P_\mu(\mathrm{d}h_\mu) \, \prod_i P_i(\mathrm{d}x_i) \, e^{-\sum_i \left[\lambda^x_i x_i + \frac{1}{2} \gamma^x_i x_i^2\right] - \sum_\mu \left[\lambda^h_\mu h_\mu + \frac{1}{2} \gamma^h_\mu h_\mu^2\right]} \right].
\end{align}
The calculations at order $1$ and $2$ are very similar to the ones of the spherical model of Sec.~\ref{subsubsec:plefka_bipartite_spherical}. One can also refer 
to the symmetric case of Sec.~\ref{subsubsec:sym_generic_prior}.
 We obtain the first four orders:
\begin{align}
	\left(\frac{\partial \Phi_{F}}{\partial \beta}\right)_{\beta=0} &= \frac{1}{N} \sum_{\mu,i} F_{\mu i} m^h_\mu m^x_i,\\
	\frac{1}{2} \left(\frac{\partial^2 \Phi_{F}}{\partial \beta^2}\right)_{\beta=0} &= \frac{1}{2N} \sum_{\mu,i} F_{\mu i}^2 v^h_\mu v^x_i,\\
	\frac{1}{3!} \left(\frac{\partial^3 \Phi_{F}}{\partial \beta^3}\right)_{\beta=0} &= \frac{1}{6N} \sum_{\mu,i} F_{\mu i}^3 \kappa^{(3,h)}_\mu \kappa^{(3,x)}_i,\\
	\label{eq:order4_bipartite_generic}
	\frac{1}{4!} \left(\frac{\partial^4 \Phi_{F}}{\partial \beta^4}\right)_{\beta=0} &= \frac{1}{8N} \sum_{\mu_1 \neq \mu_2} \sum_{i_1 \neq i_2} F_{\mu_1 i_1} F_{\mu_1 i_2}F_{\mu_2 i_2}F_{\mu_2 i_1} v^h_{\mu_1} v^h_{\mu_2}v^x_{i_1}v^x_{i_2} + \smallO_N(1).
\end{align}
Recall that $\kappa^{(p,x)}_i$ is the $p$-th order cumulant of $x_i$ at $\beta=0$, and we define $\kappa^{(p,h)}_\mu$ in the same way for $h_\mu$.
\paragraph{On higher orders} The discussion on the higher orders in perturbation is very similar to the one we made in Sec.~\ref{subsubsec:sym_generic_prior}.
In the end, we only retain the simple cycles made of matrix elements $\{F_{\mu i}\}$ in our Plefka expansion. As we explain in more details in Sec.~\ref{sec:diagrammatics}, and 
in particular in Sec.~\ref{subsec:diagrammatics_bipartite_models} for the bipartite case, we make two crucial statements to obtain this result:
\begin{itemize}
	\item Let us forget for the moment about the factors of variances or higher-order moments of the distributions of $x_i,h_\mu$ at $\beta=0$. 
		We can then study all the possible terms appearing in the Plefka expansion at every order as \emph{diagrams} made of matrix elements $\{F_{\mu i}\}$, and show that 
		the only non-vanishing diagrams are the \emph{simple cycles}. This is shown in more details in Sec.~\ref{subsec:diagrammatics_bipartite_models}.
	\item We assume that the factors arising from the variances $v^x_i,v^h_\mu$ or higher-order cumulants of the variables $x_i$ and $h^\mu$ do not change significantly the scaling of the diagrams, 
		that is we can still only retain the simple cycles in the thermodynamic limit. More details are given in Sec.~\ref{subsubsec:higher_moments}.
\end{itemize}
For instance, at order $3$ these statements yield $\partial^3_\beta \Phi_{N,F} = \smallO_N(1)$ since one can not construct a simple cycle 
for bipartite models at order $3$. This also explains the $\smallO_N(1)$ term in the order $4$, see eq.~(\ref{eq:order4_bipartite_generic}).
In the end, we obtain the following value of the free entropy at leading order in $N$:
\begin{align}\label{eq:plefka_bipartite_prior}
\Phi_{F}(\beta) &= \Phi_{F}(0) + \frac{\beta}{N} \sum_{\mu,i} F_{\mu i} m^h_\mu m^x_i + \frac{1}{N} \sum_{p=1}^\infty \frac{\beta^{2p} }{2p} \sum_{\substack{\mu_1,\cdots,\mu_p \\ i_1,\cdots,i_p}}  F_{\mu_1 i_1} F_{\mu_1 i_2}  \cdots F_{\mu_p i_p} F_{\mu_p i_1} \prod_{\alpha=1}^p v^h_{\mu_\alpha}v^x_{i_\alpha}.
\end{align}
In the summation all indices $\mu_1,\cdots,\mu_p$ are pairwise distinct, and so are $i_1,\cdots,i_p$.

	\paragraph{Homogeneous variances} 
Let us assume that the maximum of the free entropy of eq.~(\ref{eq:plefka_bipartite_prior})
	is attained for variables $\{v^h_\mu,v^x_i\}$ such that $v^h_\mu = v^h$ and $v^x_i = v^x$.
	As in the symmetric case this hypothesis can be justified for many models that we will analyze later on.
	Using eq.~(\ref{eq:correspondance_bipartite}) the free entropy of eq.~(\ref{eq:plefka_bipartite_prior}) can be resummed:
	\begin{align}\label{eq:resummed_bipartite_prior}
		\Phi_{F}(\beta) = &\Phi_{F}(0) + \frac{\beta}{N} \sum_{\mu,i} F_{\mu i} m^h_\mu m^x_i - \frac{1+\log v^x}{2} - \alpha \frac{1+\log v^h}{2} \nonumber \\
		&+ \frac{1}{2} \inf_{\zeta^x,\zeta^h}\left[\alpha \zeta^h v^h + \zeta^x v^x -(\alpha-1) \log \zeta^h - \int \rho_D(\mathrm{d}\lambda) \log(\zeta^x \zeta^h-\beta^2 \lambda)\right]\ ,
	\end{align}
where $\rho_D$ is the spectral distribution of $F^TF$.
	At $\beta=0$, $\Phi_{F}(0)$ is given by eq.~(\ref{eq:phi_order0_bipartite_prior}). 
	As in the symmetric model of Sec.~\ref{subsec:plefka_sym_models} we were able to perform this expansion and its resummation by 
	applying our results on the spherical models of Sec.~\ref{sec:spherical_bipartite}. 
	The diagrammatic study performed in Sec.~\ref{subsec:diagrammatics_bipartite_models}, which is a generalization of the diagrammatic for symmetric matrices, 
	plays a decisive role in this analysis.

	\paragraph{A remark on Restricted Boltzmann machines}

	In the case of an i.i.d.\ matrix $F$ (see Sec.~\ref{subsec:iid_matrices} for a remark on how to apply our results to this class of matrices), we 
	recognize in particular in eq.~\eqref{eq:plefka_bipartite_prior} the result obtained in eq.~(36) of \cite{tramel2018deterministic} for Restricted Boltzmann Machines (RBMs).
	For generic rotationally invariant $F$, the fixed point equations corresponding to the extremization of $\Phi$ in eq.~\eqref{eq:resummed_bipartite_prior}
	also correspond to the fixed point of Algorithm~3 of \cite{tramel2018deterministic} (which was described there as ``adaTAP inference algorithm'').
	This generic property of the fixed points of the Plefka-expanded free entropy will be investigated in Sec.~\ref{sec:applications_algorithms}.

	\subsubsection{Generalized Linear Models with correlated matrices}\label{subsubsec:plefka_glm}
	
	Generalized Linear Models (GLMs) \cite{nelder1972generalized,mccullagh2018generalized} arise as a generalization of the Compressed Sensing problem (see below).
	They are of primary importance in a very wide variety of scientific and engineering fields, such as phase retrieval in optics \cite{fienup1982phase}
	and classification problems in statistics. 
	GLMs can also be thought of as the building blocks of fully connected neural networks \cite{lecun2015deep}.
	Let us now define more precisely the model we will study.
	Consider $M,N \geq 1$ both going to infinity with a fixed ratio $M/N \to \alpha > 0$. We are given a (random)
	measurement matrix $F \in \bbR^{M \times N}$ which comes from the ensemble of Model~\ref{model:nsym_rot_inv}. 
	Given $F$, data samples $\{Y_\mu\}$ are generated as:
	\begin{align}\label{eq:def_glm}
	\forall \mu \in \{1,\cdots,M\}, \qquad Y_\mu \sim P_{\rm out} \left( \cdot \, \Big|\,\left(F \bX\right)_\mu \right),
	\end{align}
	in which $\bX \in \bbR^N$ is the vector we try to recover from the observation of $\{Y_\mu\}_{\mu=1}^M$, and 
	$P_{\rm out}$ is a fixed probabilistic channel. Recall that greek indices $\mu,\nu$ will always run between $1$ and $M$ and latin indices $i,j$ between $1$ and $N$.
	The vector $\bX \in \mathbb{R}^N$ is assumed to be drawn with i.i.d.\ coordinates $\{X_i\}_{i=1}^N$ according to a prior $P_X$ with zero mean and variance $\rho > 0$. 

	\paragraph{Compressed Sensing, the Gaussian channel case}

	Compressed Sensing \cite{donoho2006compressed} arises as a particular case of channel distribution in eq.~\eqref{eq:def_glm}, 
	in which the channel is taken to be a Gaussian distribution with zero mean and variance $\Delta$. Equivalently, it can be formulated as:
	\begin{align}\label{eq:def_cs}
		Y_\mu &= \sum_i F_{\mu i} X_i + \sqrt{\Delta} z_\mu.	
	\end{align}
	Our aim is to infer the vector $\bX$ from these observations.
	In this equation we modeled the noise by a standard Gaussian variable $z_\mu$, the strength of the noise being $\Delta > 0$. 
	In \cite{krzakala2012probabilistic} the authors have considered a subclass of matrices $F$, namely i.i.d.\ matrices. 
	We follow here the same probabilistic inference approach, as we aim to study this problem by sampling from the following distribution: 
	\begin{align}\label{eq:posterior_cs}
		P(\bx|\bY) &= \frac{1}{Z_{\bY,F}} P_X(\bx) \exp \left\{-\frac{1}{2 \Delta} \sum_\mu \left(Y_\mu - \sum_i F_{\mu i} \, x_i\right)^2\right\}.
	\end{align} 
	As the parameters $(P_X,\Delta)$ of the signal are known, we could use them in our model, a setting which is known as the 
	\emph{Bayes-optimal} setting. Using the matrix $J = -F^\intercal F$ (which follows Model~\ref{model:sym_rot_inv}), we can rewrite the posterior distribution, up to a normalization, as:
	\begin{align}\label{eq:posterior_cs_2}
		P(\bx|\bY) &= \frac{1}{Z_{\bY,F}} P_X(\bx) \exp \left\{\frac{1}{2 \Delta} \sum_{i,j} J_{ij} \, x_i x_j + \frac{1}{\Delta} \sum_{\mu,i} F_{\mu i} Y_\mu x_i\right\}.
	\end{align} 
	Defining $\beta \equiv \Delta^{-1}$ it becomes clear that this can be written as a Gibbs-Boltzmann distribution of the model we studied in Sec.~\ref{subsubsec:sym_generic_prior}, 
	with $J = -F^\intercal F$ and $h_i = -\sum_\mu F_{\mu i} Y_\mu$.
	Assuming that the variance variables at the maximum of the free entropy are homogeneous ($v_i = v$) we can use eq.~(\ref{eq:resummed_phi_symmetric}) to directly obtain 
	the free entropy of this problem as a function of $\mathcal{R}_J$, the $\mathcal{R}$-transform of the asymptotic spectrum of the $J$ matrix:
	\begin{align}\label{eq:phi_cs}
		\Phi_{\bY,F}(\beta) &= \Phi_{\bY,F}(0) - \frac{\beta}{2 N}	\sum_{i,j} (F^\intercal F)_{ij} m_i m_j + \frac{\beta}{N} \sum_{\mu,i} F_{\mu i} Y_\mu m_i + \frac{1}{2} \int_0^{\beta v} \mathcal{R}_{J}(u) \mathrm{d}u.
	\end{align}
	We postpone the analysis of the corresponding fixed point equations to our algorithmic discussion in Sec.~\ref{subsec:gamp} and Sec.~\ref{subsec:vamp}.

	\paragraph{Generic channel distributions}

	We now turn to generic $P_{\rm out}$ distributions.
	We assume that both $P_{\rm out}$ and $P_X$ are \emph{known} (this is the Bayes-optimal setting, known in statistical physics as the \emph{Nishimori line}), so that we can use them
	in the posterior distribution:
	\begin{align}\label{eq:posterior_glm}
		P(\bx|\bY) &= \frac{1}{Z(\bY,F)} \prod_{i} P_X(x_i) \, \prod_\mu P_{\rm out}\left[Y_{\mu}| (F\bx)_\mu\right],
	\end{align}
	from which we will sample to obtain an estimate of $\bX$.
	While in the compressed sensing setting $\beta = \Delta^{-1}$ played naturally the role of an inverse temperature, 
	in the general setting of eq.~(\ref{eq:def_glm}) there is \emph{a priori} no way to perform a Plefka expansion.
	As it turns out, there is a way to introduce an auxiliary parameter in terms of which we will perform the expansion, similarly to what is done in~\cite{altieri2016jamming,altieri2018higher}.
	Introducing the usual Lagrange parameters
	to fix the mean and variance of $\{x_i\}$, we obtain the free entropy:
	\begin{align}
	\Phi_{\bY,F} &\equiv \frac{1}{N} \sum_{i} \lambda_i m_i + \frac{1}{2 N} \sum_{i} \gamma_i \left(v_i + m_i^2 \right) + \frac{1}{N}\log \left[\int_{\mathbb{R}^N} \mathrm{d} \bx \, e^{-S\left[\bx\right]} \right],
	\end{align}
	in which we introduced an \emph{action} $S[\bx]$: 
	\begin{align}
	S[\bx] &\equiv \sum_{i} \lambda_i x_i + \frac{1}{2} \sum_{i} \gamma_i x_i^2 - \sum_{\mu} \log P_{\rm out}\left(Y_\mu \Big| \sum_{i} F_{\mu i} x_i\right) - \sum_{i} \log P_X(x_i).
	\end{align}
	As before $\{\lambda_i,\gamma_i\}$ are Lagrange parameters used to enforce the condition $\braket{x_i} = m_i$ and $\braket{x_i}^2 = v_i + m_i^2$.
	Introducing an auxiliary field $\bh \equiv \bF \bx \in \bbR^M$, and using the Fourier representation of the Dirac distribution, we reach: 
	\begin{align}
	\Phi_{\bY,F} &= - \alpha \log 2\pi +\frac{1}{N} \sum_{i} \lambda_i m_i + \frac{1}{2N} \sum_{i} \gamma_i \left(v_i + m_i^2 \right) + \frac{1}{N}\log \left[\int\mathrm{d} \bx \, \mathrm{d} \bh \, \mathrm{d} \tilde{\bh}  \, e^{-S_{\rm eff}\left[\bx,\bh,\widetilde \bh\right]} \right],
	\end{align}
	with a new effective action $S_{\rm eff}$:
	\begin{align}
	S_{\rm eff}\left[\bx,\bh,\tilde{\bh}\right] \equiv \sum_i \left[\lambda_i x_i + \frac{1}{2} \gamma_i x_i^2\right] & -\sum_i \log P_X(x_i) - \sum_{\mu}\left[\log P_{\rm out}(Y_\mu|h_\mu) + h_\mu (i\tilde{h}_\mu) \right]  \nonumber \\
	\label{eq:Seff_CS} &+  \sum_{\mu,i}(i\tilde{h}_\mu) F_{\mu i} x_i.
	\end{align}
	The key idea is to treat $\bx$ and $i\tilde{\bh}$ as two independent non-Gaussian fields that interact via the last (quadratic) term of eq.~(\ref{eq:Seff_CS}) and to perform 
	a Plefka expansion in terms of this \emph{effective} Hamiltonian, which is exactly the bipartite Hamiltonian of the general model of Sec.~\ref{subsubsec:bipartite_generic_prior}.
	This mapping of a generalized linear model to a bipartite model using Fourier transformation has already been 
	successfully applied in the context of the replica method, see \cite{kabashima2008inference, kabashima2008integral}.
	We will call $\eta$ the ``inverse temperature'', that is 
	in eq.~(\ref{eq:Seff_CS}) we substitute:
	\begin{align}
		\sum_{\mu,i}F_{\mu i} \, x_i \, (i\tilde{h}_\mu) &\to  \eta\sum_{\mu,i} F_{\mu i} \, x_i \, (i\tilde{h}_\mu), 
	\end{align}
	and at the end of the expansion we will set $\eta=1$.
	Similarly as for the field $\bx$ we will fix the first and second moments of the field $i \tilde{\bh}$ as $\braket{i \tilde{h}_\mu}_\eta = f_\mu$ and 
	$\braket{(i \tilde{h}_\mu)^2}_\eta = -r_\mu + f_\mu^2$, conditions that will be enforced by new Lagrange parameters $\{\omega_\mu,b_\mu\}$.
	Although a bit tedious, this is straightforward, and we obtain a free entropy in which we will perform a low-$\eta$ expansion:
	\begin{align}
	\Phi_{\bY,F}(\eta) = - \alpha \log 2\pi +& \frac{1}{N}\sum_{i} \lambda_i m_i + \frac{1}{2N} \sum_{i} \gamma_i \left(v_i + m_i^2 \right) + \frac{1}{N} \sum_{\mu} \omega_\mu f_\mu - \frac{1}{2N} \sum_{\mu} b_\mu \left(-r_\mu + f_\mu^2 \right) \nonumber \\
	& + \frac{1}{N}\log \left[\int \mathrm{d} \bx \, \mathrm{d} \bh \, \mathrm{d} \tilde{\bh} \, e^{-S_{\rm eff}\left[\bx,\bh,\tilde{\bh}\right]} \right]. \label{eq:Phi_eta_CS}
	\end{align}
	The effective action and Hamiltonian are expressed as follows:
	\begin{align}
	S_{\rm eff}\left[\bx,\bh,\tilde{\bh}\right] &\equiv \sum_{i} \lambda_i x_i + \frac{1}{2} \sum_{i} \gamma_i x_i^2 + \sum_{\mu} \omega_\mu (i \tilde{h}_\mu) - \frac{1}{2} \sum_{\mu} b_\mu (i\tilde{h}_\mu)^2 \nonumber \\
	\label{eq:Seff_eta_CS}&-\sum_i \log P_X(x_i) -\sum_{\mu} \log P_{\rm out}\left(Y_\mu | h_\mu \right) -  \sum_{\mu} h_\mu (i\tilde{h}_\mu) + \eta  \,H_{\rm eff}\left[\bx,\tilde{\bh}\right], \\
	\label{eq:Heff_CS_eta}
	H_{\rm eff}[\bx, \tilde{\bh}] &\equiv \sum_{\mu,i} F_{\mu i} \, x_i \, (i \tilde{h}_\mu).
	\end{align}
	From these equations it is clear that:
	\begin{itemize}
		\item[$(i)$] The priors on the variables $\{x_i\}$ and $\{(i\tilde{h}_\mu)\}$ decouple. The prior on $x_i$ is $P_X(x_i)$, while the prior distribution on
		$(i \tilde{h})_\mu$ is related to the Fourier transform of the channel distribution:
		\begin{align}\label{eq:P_mu}
			P_{\tilde{H}}(i\tilde{h}_\mu) &= \int \frac{\mathrm{d}h}{2 \pi} e^{i h \tilde{h}_\mu} P_{\rm out}(Y_\mu | h).
		\end{align}
		\item[$(ii)$] The interaction Hamiltonian of eq.~\eqref{eq:Heff_CS_eta} is a bipartite Hamiltonian of the type of the model we studied in Sec.~\ref{subsubsec:bipartite_generic_prior}, in 
		terms of the variables $x_i$ and $(i \tilde{h}_\mu)$.
	\end{itemize}
	Points $(i)$ and $(ii)$ allow us to use the general results of Sec.~\ref{subsubsec:bipartite_generic_prior}.
	We can therefore directly conjecture the final form of the free entropy we were seeking:
	\begin{align}\label{eq:phi_glm_fullplefka}
		\Phi_{\bY,F}(\eta = 1) &= \Phi_{\bY,F}(\eta=0) - \frac{1}{N}\sum_{\mu i} F_{\mu i} f_\mu m_i \\
		&+ \frac{1}{N} \sum_{p=1}^\infty \frac{(-1)^p}{2p} \hspace{-0.5cm}\sum_{\substack{\mu_1,\cdots,\mu_p \\ \text{pairwise distincts}}}\hspace{-0.cm} \sum_{\substack{i_1,\cdots,i_p \\ \text{pairwise distincts}}}\hspace{-0.5cm} F_{\mu_1 i_1} F_{\mu_2 i_1} \cdots F_{\mu_p i_p} F_{\mu_1 i_p} \prod_{\alpha=1}^p r_{\mu_\alpha} v_{i_\alpha} + \smallO_N(1). \nonumber
	\end{align}

\paragraph{Homogeneous variances} Once again we can assume that the maximum of the free entropy written in eq.~\eqref{eq:phi_glm_fullplefka} will be attained
for variance variables $\{v_i,r_\mu\}$ that are \emph{homogeneous}: they satisfy $v_i = v$ and $r_\mu = r$.
Using the resummation of eq.~\eqref{eq:resummed_bipartite_prior} this leads to a simplified expression for eq.~\eqref{eq:phi_glm_fullplefka}:
\begin{align}\label{eq:phi_glm_resummed}
	\Phi_{\bY,F}(\eta = 1) &= \Phi_{\bY,F}(\eta=0) - \frac{1}{N}\sum_{\mu i} F_{\mu i} f_\mu m_i - \alpha \, \frac{1 + \log r}{2} - \frac{1+ \log v}{2}\\
	&+\frac{1}{2} \inf_{\zeta,\zeta'} \left[\alpha \zeta r + \zeta' v - (\alpha-1) \log \zeta -\frac{1}{N} \log \det \left[\zeta \zeta' {\rm I}_N + F^\intercal F\right]\right] + \smallO_N(1). \nonumber
\end{align}
We will study the fixed point equations corresponding to the free entropies of eq.~\eqref{eq:phi_glm_fullplefka}
and eq.~\eqref{eq:phi_glm_resummed} in Sec.~\ref{subsec:gamp} and Sec.~\ref{subsec:gvamp}.

%% file: fixed_point.tex
\section{Consequences for iterative algorithms}\label{sec:applications_algorithms}

In the previous section we showed how to use the Plefka expansion to derive the 
single-graph free entropy of a large class of systems with pairwise interactions. 
One then needs to maximize this free entropy, which yields fixed point equations.
Iterating these fixed point equations is in itself a challenge since different choices for the 
iteration scheme can lead to drastically different convergence properties.
In the context of the Plefka expansion, or equivalently in the adaTAP or EC approximation, 
several iterations schemes for the TAP equations have been studied, see for instance \cite{opper2016theory,ccakmak2016self}.

On a parallel point of view, message-passing algorithms have been extensively studied in the 
statistical physics literature. In particular the belief-propagation equations \cite{mezard2009information} 
can be shown in the large $N$ limit to reduce to a simpler algorithm called Approximate Message Passing (AMP) \cite{donoho2009message}, 
and derived initially for i.i.d.\ coupling matrices. It is well understood that for these matrices the stationary limit of the AMP equations
is directly related to the fixed point equations of the Plefka free entropy (stopping here at order $2$ in the couplings). Extending these algorithms to 
correlated matrices has been the subject of extensive studies \cite{cakmak2014s,ma2017orthogonal}.
The  generic (and appealing by its simplicity) iteration scheme called Vector Approximate Message Passing (VAMP) \cite{rangan2017vector}
has proven to be very successful both numerically and in its theoretical justification in the case of rotationally-invariant sensing matrices.
 It has then been generalized to the broader class of generalized linear models with generic output channels \cite{schniter2016vector}.

 In Sec.~\ref{subsec:gamp} we describe the connection between AMP equations and the Plefka expansion in the context of generalized linear models with i.i.d.\ matrices,
 retrieving the GAMP algorithm \cite{rangan2011generalized} and the analysis of \cite{krzakala2012probabilistic}. 
 In Sec.~\ref{subsec:vamp} we relate the VAMP 
 algorithm to the Expectation-Consistency
 equations by showing that the stationary limit of the
 algorithm yields the TAP equations. 
We note that in the naïve Plefka expansion performed in terms of
$\Delta^{-1}$ for the compressed sensing  problem there is no clear
way to iterate the TAP equations to find back an asymptotically exact algorithm.
 Finally, 
 in Sec.~\ref{subsec:gvamp} we extend this analysis to Generalized Linear Models with correlated matrices and the G-VAMP algorithm.
 In contrast with the naïve $\Delta^{-1}$-expansion in compressed sensing, the mapping to a bipartite model that we performed in Sec.~\ref{subsubsec:plefka_glm} allows to retrieve the stationary limit of the more 
 general G-VAMP algorithm as our TAP equations.

\subsection{Generalized Approximate Message Passing for a GLM with i.i.d.\ matrices}\label{subsec:gamp}

We consider in this subsection the Generalized Linear Model (GLM) of Sec.~\ref{subsubsec:plefka_glm}.
We will concentrate on a particular subclass of sensing matrices, namely matrices $F \in \bbR^{M\times N}$ 
that are generated with i.i.d.\ centered standard Gaussian matrix elements $\{F_{\mu i}\}$. Note that the remarks 
of Sec.~\ref{subsec:iid_matrices} (coherently with the analysis of \cite{rangan2011generalized,barbier2019optimal}) show  
that the particular distribution of the elements does not need to be Gaussian, as long as the $\{F_{\mu i}\}$ are distributed independently 
and identically.
Generalized linear estimation with such i.i.d.\ matrices $F$ and generic prior and channel distributions has received a lot of attention recently.
In particular Generalized Approximate Message Passing (GAMP), an algorithm first developed in \cite{rangan2011generalized}, has been shown to be optimal
among all polynomial-time algorithms for this problem, see \cite{barbier2019optimal}.
A description of the algorithm in our case can be found in eqs.~(171)-(177) of \cite{zdeborova2016statistical}. We state here the iterative GAMP equations with our notations:
\begin{align}\label{eq:gamp_CS}
\begin{cases}
b^t &= \frac{1}{N} \sum_i v^{t-1}_i, \\
\omega^t_\mu &= \sum_i F_{\mu i} m^{t-1}_i + f^{t-1}_\mu b^t, \\
f^t_\mu &= - g_{\rm out}(Y_\mu,\omega^t_\mu,b^t),\\
\gamma^t &= -\frac{1}{N}\sum_\mu \partial_\omega g_{\rm out}(Y_\mu,\omega^t_\mu,b^t),  \\
\lambda^t_i &= -\gamma^t m^{t-1}_i + \sum_\mu F_{\mu i} f^t_\mu, \\
m^t_i &= \mathbb{E}_{P_X(\lambda^t_i,\gamma^t)} \left[x\right], \\
v^t_i &= \mathbb{E}_{P_X(\lambda^t_i,\gamma^t)} \left[(x-m^t_i)^2\right].
\end{cases}
\end{align}
In these equations $P_X(\lambda_i,\gamma_i)$ is the probability measure with density: 
\begin{align}\label{eq:def_PX}
	P_X(\lambda_i,\gamma_i)(x) \propto	P_X(x) e^{- \frac{1}{2} \gamma_i x^2-\lambda_i x },
\end{align}
and $g_{\rm out}$ is defined from the channel distribution $P_{\rm out}$ as: 
\begin{align}\label{eq:def_gout}
g_{\rm out}(y,\omega,b) &\equiv \frac{1}{b} \frac{\int \mathrm{d}z \, P_{\rm out}(y|z) \, (z-\omega)\, e^{-\frac{(z-\omega)^2}{2b}}}{\int \mathrm{d}z \, P_{\rm out}(y|z) \, e^{-\frac{(z-\omega)^2}{2b}}}.
\end{align}
We now turn to the TAP equations that we can derive from the extremization of the Plefka-expanded free entropy of eq.~\eqref{eq:phi_glm_fullplefka}.
Note that as $F$ is an i.i.d.\ matrix, all the terms with $p \geq 2$ in eq.~\eqref{eq:phi_glm_fullplefka} will be negligible.
More discussion on this can be found in Sec.~\ref{subsec:iid_matrices}. We obtain:
\begin{align}\label{eq:phi_glm_iid}
	\Phi_{F}(\eta = 1) &= \Phi_{F}(\eta = 0) - \frac{1}{N}\sum_{\mu i} F_{\mu i} f_\mu m_i - \frac{1}{2N} \sum_{\mu,i} F_{\mu i}^2 r_{\mu} v_{i} + \smallO_N(1).
\end{align}
Recall that $\Phi_{F}(\eta=0)$ is given by eq.~\eqref{eq:phi_order0_bipartite_prior}.
Extremizing $\Phi_{F}(\eta = 0)$ over the Lagrange parameters yields the moments conditions that we wish to enforce: 
\begin{align}\label{eq:eta0_glm_fixedpoint}
	\begin{cases}
		m_i &= \mathbb{E}_{P_X(\lambda_i,\gamma_i)} \left[x\right],\\
		v_i &= \mathbb{E}_{P_X(\lambda_i,\gamma_i)} \left[(x-m_i)^2\right],\\
		f_\mu &= -g_{\rm out} (y_\mu,\omega_\mu,b_\mu), \\
		r_\mu &= -\partial_\omega g_{\rm out} (y_\mu,\omega_\mu,b_\mu).
	\end{cases}
\end{align}
On the other hand, the maximization of eq.~\eqref{eq:phi_glm_iid} with respect to the physical parameters $\{m_i,v_i,f_\mu,r_\mu\}$ 
leads to four additional equations: 
\begin{align}\label{eq:fixedpoint_glm_iid}
	\begin{cases}
		\lambda_i &= - \gamma_i m_i + \sum_\mu F_{\mu i} f_\mu,\\
		\omega_\mu &= f_\mu b_\mu + \sum_i F_{\mu i}m _i, \\
		\gamma_i &= \frac{1}{N} \sum_\mu r_\mu, \\
		b_\mu &= \frac{1}{N} \sum_i v_i.
	\end{cases}
\end{align}
Combining eq.~\eqref{eq:eta0_glm_fixedpoint} and eq.~\eqref{eq:fixedpoint_glm_iid}
we see easily that the equations that one has to solve are exactly the ones
of the GAMP algorithm of eq.~\eqref{eq:gamp_CS} (without time indices).

\paragraph{The Gaussian channel case} In the case of an additive gaussian channel with variance $\Delta$ the problem reduces
to the compressed sensing studied in Sec.~\ref{subsubsec:plefka_glm} and Sec.~\ref{subsec:vamp}, with a Gaussian i.i.d.\ sensing matrix.
In this case, the function $g_{\rm out}$ of eq.~\eqref{eq:def_gout} is computable explicitly, 
and leads to $f_{\mu} = (\omega_{\mu}-y_{\mu})/(\Delta+b_{\mu})$ and $r_{\mu} = (\Delta+b_{\mu})^{-1}$.
In this particular case one recovers the AMP algorithm of \cite{krzakala2012probabilistic}, which is also compatible with the VAMP algorithm of Sec.~\ref{subsec:vamp}.

\subsection{Vector Approximate Message Passing (VAMP) in Compressed Sensing}\label{subsec:vamp}
\subsubsection{The TAP equations in Compressed sensing}\label{subsubsec:fixed_point_cs}

We analyze here the fixed point equations for the Compressed Sensing (CS) problem, see Sec.~\ref{subsubsec:plefka_glm} for 
the corresponding free entropy derivation. Our starting point 
is the Plefka-expanded free entropy written in eq.~\eqref{eq:phi_cs}. 
The extremization over the Lagrange parameters $\{\lambda_i,\gamma_i\}$ (which are considered at $\beta = 0$) yields:
	\begin{subnumcases}{\label{eq:fixed_point_cs_1}}
	m_i = \frac{\int \mathrm{d}x P_X(x) \, x \, e^{-\frac{1}{2}\gamma_i x^2 - \lambda_i x}}{\int \mathrm{d}x \, P_X(x) \, e^{-\frac{1}{2}\gamma_i x^2 - \lambda_i x}} \overset{(a)}{\equiv} F_m(\lambda_i,\gamma_i),  & \\
	v_i = \frac{\int \mathrm{d}x \, P_X(x) \, (x-m_i)^2 \, e^{-\frac{1}{2}\gamma_i x^2 - \lambda_i x}}{\int \mathrm{d}x \, P_X(x) \, e^{-\frac{1}{2}\gamma_i x^2 - \lambda_i x}}  \overset{(b)}{\equiv} F_v(\lambda_i,\gamma_i), &
	\end{subnumcases}
where $(a)$ and $(b)$ respectively define the functions $F_m$ and $F_v$. 
Maximizing the free entropy with respect to the physical parameters $\{m_i,v\}$ results in the following equations (recall that $\beta = \Delta^{-1}$):
	\begin{subnumcases}{\label{eq:fixed_point_cs_2}}
	\gamma = - \frac{1}{\Delta} \mathcal{R}_{-F^\intercal F}\left(\frac{v}{\Delta}\right) =  \mathcal{R}_{F^\intercal F/ \Delta}\left(- v\right), &\\	
	\lambda_i = - \gamma m_i + \frac{1}{\Delta} \sum_j (F^\intercal F)_{ij} m_j - \frac{1}{\Delta}\sum_\mu F_{\mu i} Y_\mu.&
	\end{subnumcases}
	Eq.~\eqref{eq:fixed_point_cs_1} and eq.~\eqref{eq:fixed_point_cs_2} define a set of fixed point equations that one has to solve 
	in order to retrieve the maximum of the free entropy of eq.~\eqref{eq:phi_cs}. 

\subsubsection{TAP equations and the fixed point of the VAMP algorithm}

\paragraph{A remark on i.i.d.\ matrices}
We start with a remark on the case of an i.i.d.\ matrix $F$. 
Remarkably, eqs.~\eqref{eq:fixed_point_cs_1} and
\eqref{eq:fixed_point_cs_2} are compatible with the fixed points of
AMP, see eqs.~(22) and (23) in~\cite{krzakala2014variational} with $R_i = - \lambda_i/\gamma$ and $\Sigma^{-2}=\gamma$, since in
this case $\mathcal{R}_{F^\intercal F/\Delta}(- v) = \alpha/(\Delta+ v)$, see for instance \cite{tulino2004random}. 
We now turn to the VAMP algorithm for a general rotationally invariant matrix $F$.
Applying the VAMP derivation of Sec.~\ref{subsubsec:derivation_vamp} to the Compressed Sensing problem of Sec.~\ref{subsubsec:plefka_glm}, 
the VAMP algorithm reads\footnote{Note that instead of fixing all the correlations $\braket{x_i x_j}$, we only fix the `diagonal' second moments $\braket{x_i^2}$.}:
\begin{align}
 (\bmm^{t}_{1})_i &= F_m((\blambda_J^t)_i,a_J^t), &  v^t_1 &= \frac{1}{N} \sum_i F_v((\blambda^t_J)_i,\gamma_J^t), \\
\label{EqA2CS}
  \blambda_0^{t} &= - \frac{\bmm^t_1}{v^t} - \blambda_J^t, &  \gamma_0^{t} &= \frac{1}{v^t} - \gamma^t_J, \\
\label{Eqa2CS}
 \bmm^{t}_2 &= - (\gamma_0^{t}+ F^\intercal F/\Delta)^{-1} (\blambda_0^{t}  - \frac{1}{\Delta} F^\intercal \bY), & v^{t}_2 &= \frac{1}{N} \text{Tr} \frac{1}{\gamma_0^{t}+ F^\intercal F/\Delta}, \\
\label{Eq_aJ_CS}
 \blambda_J^{t+1} &= - \frac{\bmm^{t}_2}{v^{t}_2} - \blambda_0^{t}, &  \gamma_J^{t+1} &= \frac{1}{v^{t}_2} - \gamma^{t}_0,
\end{align}
where $F_m$ and $F_v$ were defined in eq.~\eqref{eq:fixed_point_cs_1}.
Note that in Compressed Sensing the matrix $\gamma_0^{t} + F^\intercal F/\Delta$ has only strictly positive eigenvalues since $\gamma_0^t \geq 0$, 
so the previous iterative equations are always well defined.
 At the fixed point, we expect $\bmm_1=\bmm_2 = \bmm$ and $v_1=v_2 = v$.
In the stationary limit eq.~\eqref{Eqa2CS} yields:
\begin{equation}
v = \mathcal{S}_{F^\intercal F/\Delta}(-\gamma_0) \qquad  \Rightarrow \qquad \gamma_0 = - \mathcal{S}_{F^\intercal F/\Delta}^{-1}(v).
\end{equation}
From eq.~\eqref{Eq_aJ_CS}, one has
\begin{equation}
\gamma_J = \frac{1}{v} +  S_{F^\intercal F/\Delta}^{-1}(v)  =  R_{F^\intercal F/\Delta}(- v).
\end{equation}
And from eq.~\eqref{Eqa2CS}, we obtain
\begin{equation}
(\gamma_0 +F^\intercal F/\Delta ) \bmm = - \blambda_0 +  \frac{1}{\Delta}  F^\intercal \bY,
\end{equation}
and
\begin{equation}
\Big(\frac{1}{v} - \gamma_J + \frac{F^\intercal F}{\Delta} \Big) \bmm = \frac{\bmm}{v}  + \blambda_J +  \frac{1}{\Delta}  F^\intercal \bY \ ,
\end{equation}
which gives
\begin{equation}\label{Eq_B}
\blambda_J = \Big(- \gamma_J +  \frac{F^\intercal F}{\Delta} \Big) \bmm -  \frac{1}{\Delta} F^\intercal \bY = 
- \gamma_J \bmm - \frac{1}{\Delta} F^\intercal \Big( \bY - F \bmm \Big) \ .
\end{equation}
One now recognizes easily the fixed points obtained with the Plefka expansion in Sec.~\ref{subsubsec:fixed_point_cs}, namely eq.~\eqref{eq:fixed_point_cs_1} and eq.~\eqref{eq:fixed_point_cs_2},
with $\blambda_J =  \blambda$ and $\gamma_J=\gamma$.

\paragraph{A remark on iterating the TAP equations in the i.i.d.\ case}
Note that in the i.i.d.\ case (Sec.~\ref{subsec:gamp}), doing the Plefka expansion in terms of the $\eta$ parameter
after having mapped the GLM to a bipartite problem allows us not only to retrieve the fixed point of the GAMP algorithm 
(and even the G-VAMP for non-i.i.d.\ matrices, as we will see in Sec.~\ref{subsec:gvamp}), but there is a simple iterating scheme of the TAP equations that exactly yields the GAMP algorithm.
We insist that this is not true when making the correspondence of the VAMP algorithm 
with the Plefka expansion in $\Delta^{-1}$ for compressed sensing with an i.i.d.\ matrix.
This underlines one of the possible limitations of the EC, adaTAP and Plefka methods for these problems, as iterating the TAP equations 
with an algorithmic scheme that guarantees convergence is a very involved task, while the VAMP derivation provides 
an iteration scheme of the equations.

\subsection{Generalized Vector Approximate Message Passing (G-VAMP) for Generalized Linear Models}\label{subsec:gvamp}

 We focus in this section on Generalized Linear Models with a correlated matrix $F$ that satisfies rotation invariance (Model~\ref{model:nsym_rot_inv}).
 We first derive the TAP equations from the Plefka expansion we performed in Sec.~\ref{subsubsec:plefka_glm}, before stating
 the G-VAMP algorithm for this problem following \cite{schniter2016vector}. 
 We then analyze how the stationary limit of G-VAMP is equivalent to these TAP equations.

 \subsubsection{The TAP equations from the Plefka expansion}

 Recall that the Plefka-expanded free entropy was computed in Sec.~\ref{subsubsec:plefka_glm}. 
 Following the assumptions of the VAMP and G-VAMP algorithms \cite{schniter2016vector,rangan2017vector} we assume that the 
 variances $\{v_i,r_\mu\}$ are \emph{homogeneous}, that is $r_\mu = r$ and $v_i = v$. We can then use the resummed expression of the Plefka 
 free entropy expressed in eq.~\eqref{eq:phi_glm_resummed}.
 We first extremize this expression with respect to
 the Lagrange parameters $\{\lambda_i,\gamma_i,\omega_\mu,b_\mu\}$ and we obtain an equivalent expression to eq.~\eqref{eq:fixed_point_cs_1}.
 We reach more precisely:
\begin{align}\label{eq:fixed_point_glm_1}
	\begin{cases}
		m_i &= \mathbb{E}_{P_X(\lambda_i,\gamma)} \left[x\right],\\
		v_i &= \mathbb{E}_{P_X(\lambda_i,\gamma)} \left[(x-m_i)^2\right],\\
		f_\mu &= -g_{\rm out} (y_\mu,\omega_\mu,b), \\
		r &= -\frac{1}{M} \sum_\mu \partial_\omega g_{\rm out} (y_\mu,\omega_\mu,b).
	\end{cases}
\end{align}
 Recall the definitions of $P_X(\lambda,\gamma)$ and $g_{\rm out}(y,\omega,b)$ from eq.~\eqref{eq:def_PX}
 and eq.~\eqref{eq:def_gout}. The remaining equations are obtained by maximizing eq.~\eqref{eq:phi_glm_resummed}
 with respect to the physical parameters. We make use of the Jacobi formula for a symmetric positive definite matrix 
 $J \in {\cal S}_N^{++}$: $\frac{\partial}{\partial J_{ij}} \log \det J =  (J^{-1})_{ij}$. 
 We reach:
	\begin{subnumcases}{\label{eq:fixed_point_glm_2}}
		\lambda_i = - \gamma m_i + \sum_\mu F_{\mu i} f_\mu,& \label{eq:fixed_point_glm_lambda}\\
		\omega_\mu = f_\mu b + \sum_i F_{\mu i}m_i, & \label{eq:fixed_point_glm_omega}\\
		\zeta \mathcal{S}_{F^\intercal F}\left(-\zeta \zeta'\right) = v,& \label{eq:fixed_point_glm_zeta}\\
		\zeta' \mathcal{S}_{F^\intercal F}\left(-\zeta \zeta'\right) = \alpha r - \frac{\alpha-1}{\zeta} ,&\label{eq:fixed_point_glm_zetap}\\
		\gamma = \frac{1}{v} - \zeta',& \label{eq:fixed_point_glm_gamma}\\
		b = \frac{1}{r} - \zeta. & \label{eq:fixed_point_glm_b}
	\end{subnumcases}

	\paragraph{Remark: Additive gaussian channel} In the case of an additive Gaussian channel with variance $\Delta$
	 we find $r = (\Delta+b)^{-1}$, which gives $\zeta=\Delta$ and $\gamma = \mathcal{R}_{F^TF/\Delta}(-v)$.
	 We thus coherently recover the TAP equations for the compressed sensing problem (see Sec.~\ref{subsubsec:fixed_point_cs}) 
	 even though these equations were derived with a ``naïve'' Plefka expansion in powers of $\beta \equiv \Delta^{-1}$.

	 \subsubsection{The G-VAMP algorithm for Generalized Linear Models}\label{subsubsec:gvamp_glm}

		With a similar reasoning that we used to derive the VAMP algorithm for a symmetric pairwise model, we can write a VAMP algorithm for a bipartite model.
		We do not describe its full derivation here, and we simply report the G-VAMP algorithm for the GLM as 
		stated in \cite{schniter2016vector}.
		We define a set of functions:
	\begin{subnumcases}{}
		\tilde{F}_m(r,\gamma) \equiv  \frac{\int \mathrm{d}x P_X(x) \, x \, e^{-\frac{1}{2}\gamma (x - r)^2}}{\int \mathrm{d}x \, P_X(x) \, e^{-\frac{1}{2}\gamma (x - r)^2}} = 
		F_m(- \gamma r,\gamma) , & \\
		\tilde{F}_v(r,\gamma) \equiv  \frac{\int \mathrm{d}x P_X(x) \, x^2 \, e^{-\frac{1}{2}\gamma (x - r)^2}}{\int \mathrm{d}x \, P_X(x) \, e^{-\frac{1}{2}\gamma (x - r)^2}} - (\tilde{F}_m(r,\gamma))^2 = F_v(- \gamma r,\gamma), & \\
		\tilde{F}_z(\omega,\tau) \equiv  \frac{\int \mathrm{d}z P_{\rm out}(y|z) \, z \, e^{-\frac{1}{2} \tau (z - \omega)^2}}{\int \mathrm{d}z P_{\rm out}(y|z) \, e^{-\frac{1}{2} \tau (z - \omega)^2}} = g_{\rm out} (y,\omega,\tau^{-1})  \tau^{-1} + \omega, &\\
		\label{Eq_F_kappa}
		\tilde{F}_\kappa(\omega,\tau) \equiv \frac{\int \mathrm{d}z P_{\rm out}(y|z) \, z^2 \, e^{-\frac{1}{2} \tau (z - \omega)^2}}{\int \mathrm{d}z P_{\rm out}(y|z) \, e^{-\frac{1}{2} \tau (z - \omega)^2}} - (\tilde{F}_z(\omega,\tau))^2 =  \partial_\omega g_{\rm out} (y,\omega,\tau^{-1}) \tau^{-2} + \tau^{-1}. &
	\end{subnumcases}
		The full algorithm then amounts to iterate the following equations:
		\begin{subequations}
			\label{eq:gvamp_full}
		\begin{align}
		\bmm_{1i}^t &= \tilde{F}_m((\br_J^t)_i,\gamma^t_J), & v_1^t &= \frac{1}{N} \sum_i \tilde{F}_v((\br_J^t)_i,\gamma^t_J), \\
		\label{Eq_vamp_GLM_r2}
		\br_{0}^{t} &= \frac{(\bmm^t - \gamma^t_J v^t \br^t_J )}{(1- \gamma_J^t v^t)},  & \gamma_{0}^{t} &=  \frac{1}{v^t} - \gamma_{J}^t, \\
		\bz_{1\mu}^t &= \tilde{F}_z((\bomega_J^t)_\mu,\tau^t_J), & \kappa^t_1 &= \frac{1}{M} \sum_{\mu} \tilde{F}_\kappa((\bomega_J^t)_\mu,\tau^t_J),\\
		\label{Eq_vamp_GLM_p2}
		\bomega_{0}^{t} &= \frac{(z^t - \tau_{J}^t\kappa^t \bomega_{J}^t)}{(1-\tau_{J}^t \kappa^t)}, & \tau_{0}^{t+1} &= \frac{1}{\kappa^t}_1 -  \tau_{J}^t, \\
		\label{Eq_vamp_GLM_x2}
		\bmm^{t}_2 &= \frac{1}{\tau_{0}^{t+1} F^T F + \gamma_{0}^{t}} \left( \gamma_{0}^{t} \br_0^{t} + F^T \tau_{0}^{t} \bomega_0^{t} \right), & v^{t}_2 &= \frac{1}{N} \text{Tr} \frac{1}{\tau_0^{t} F^T F + \gamma_0^{t}}, \\
		\label{Eq_vamp_GLM_r}
		\br_J^{t+1} &= \frac{(\bmm^{t}_2 - \gamma_{0}^{t} v_2^{t} \br_{0}^{t})}{(1- \gamma_{0}^{t} v_2^{t})}, & \gamma_J^{t+1} &= \frac{1}{v_2^{t}} - \gamma_0^{t}, \\
		\label{Eq_vamp_GLM_z2}
		\bz_2^{t} &= F \frac{1}{\tau_{0}^{t} F^\intercal F + \gamma_{0}^{t}} \left( \gamma_{0}^{t} \br_0^{t} + F^\intercal \tau_{0}^{t} \bomega_0^{t} \right),& \kappa^{t}_2 &= \frac{1}{M} \text{Tr} F^\intercal F  \frac{1}{\tau_0^{t} F^\intercal F + \gamma_0^{t}}, \\
		\label{Eq_vamp_GLM_p}
		\bomega_{J}^{t+1} &= \frac{(\bz^{t}_2 - \tau_{0}^{t}\kappa^{t}_2 \bomega_{0}^{t})}{(1-\tau_{0}^{t}\kappa^{t}_2)}, & \tau_{J}^{t+1} &= \frac{1}{\kappa^{t}_2} - \tau_{0}^{t}.
		\end{align}
		\end{subequations}
		
		\subsubsection{TAP equations and fixed points of G-VAMP}

		 We want to see if the stationary limit of G-VAMP, that is the G-VAMP equations without time indices, is related to the TAP equations.
		At the fixed points of the G-VAMP algorithm written in eq.~\eqref{eq:gvamp_full},
		 one expects the following equalities to take place: $\bmm_1=\bmm_2=\bmm$, $\bz_1=\bz_2=\bz$, $v_1=v_2=v$ and $\kappa_1=\kappa_2=\kappa$.
		 We start from the TAP equations, eq.~\eqref{eq:fixed_point_glm_1} and eq.~\eqref{eq:fixed_point_glm_2}, and we will try to recover every equation in eq.~\eqref{eq:gvamp_full}.
		 \begin{itemize}
			\item From eq.~\eqref{eq:fixed_point_glm_b} and eq.~\eqref{Eq_F_kappa} we can write
			\begin{equation}\label{Eq_b_glm_vamp}
			\frac{1}{b} = \frac{1}{\kappa} - \frac{1}{\zeta} \ ,
			\end{equation}
			which can be identified with eq.~\eqref{Eq_vamp_GLM_p}, with $b=\tau_J^{-1}$ and $\zeta=\tau_0^{-1}$.
			\item Using eq.~\eqref{Eq_F_kappa} we write eq.~\eqref{eq:fixed_point_glm_zetap} 
			as
			\begin{equation}
			\frac{r}{\zeta'} = - \frac{\kappa}{\zeta' b^2} + \frac{1}{\zeta' b} = \frac{1}{M} \text{Tr} \left[\frac{1}{\zeta\zeta'+F^\intercal F}\right] + \frac{1-\alpha^{-1}}{\zeta'\zeta}.
			\end{equation}
			Finally from eq.~\eqref{Eq_b_glm_vamp} we obtain
			\begin{equation}
			\frac{\kappa}{\zeta}  =  \frac{1}{\alpha} - \frac{1}{M} \text{Tr} \frac{\zeta\zeta'}{\zeta\zeta'+F^\intercal F},
			\end{equation}
			which is compatible with the second part of eq.~\eqref{Eq_vamp_GLM_z2}, with $\zeta=\tau_0^{-1}$  and $\zeta'=\gamma_0$.
			\item Eq.~\eqref{eq:fixed_point_glm_zeta} and eq.~\eqref{eq:fixed_point_glm_gamma} are equivalent to the second parts of eq.~\eqref{Eq_vamp_GLM_x2}
			and eq.~\eqref{Eq_vamp_GLM_r},
			with $\zeta'=\gamma_0$, $\zeta=\tau_0^{-1}$  and $\gamma=\gamma_J$.
			\item We write eq.~\eqref{Eq_vamp_GLM_x2} as
			\begin{equation}
			(\tau_{0} F^\intercal F + \gamma_{0}) \bmm = \left( \gamma_{0} \br_0 + F^\intercal  \tau_{0} \bomega_0 \right),
			\end{equation}
			and using that $F \bmm = \bz$, as well as eq.~\eqref{Eq_vamp_GLM_r2} and eq.~\eqref{Eq_vamp_GLM_p2}, we arrive at
			\begin{equation}
			\gamma_J \br_J = \gamma_J \bmm + \tau_{J} F^T( \bz - \bomega_J),
			\end{equation}
			which is exactly eq.~\eqref{eq:fixed_point_glm_lambda} with $\omega=\omega_J$, $ \blambda = - \br_J \gamma_J$ and
			$\tau_J = b^{-1}$.
			\item Finally we note that eq.~\eqref{eq:fixed_point_glm_omega}
			at the fixed point is nothing but $\bz = F \bmm$, which gives eq.~\eqref{Eq_vamp_GLM_z2}.
		 \end{itemize}
		All these relations show the equivalence between the stationary limit of the G-VAMP algorithm of \cite{schniter2016vector} and the (TAP)
		maximization equations of the free entropy that we derived with our Plefka expansion in Sec.~\ref{subsubsec:plefka_glm}.

%% file: diagrammatics.tex
\section{The diagrammatics of the Plefka expansion} \label{sec:diagrammatics}

The goal of this section is to precise how the different diagrams arising in our Plefka expansions in Sec.~\ref{sec:stat_models}
can be computed. Recall that for symmetric random matrices $J$ we construct diagrams as described in Fig.~\ref{fig:diagrams_example}. 
\begin{figure}[t]
 \centering
\captionsetup{justification=centering}
\begin{subfigure}[b]{0.42\textwidth}
   \centering
\begin{tikzpicture}[scale=1.]
\node (i0) at (0,3) {$\bullet$};
\node (i1) at (1.5,3) {$\bullet$};
\node (i2) at (3,3) {$\bullet$};
\node (i3) at (0.725,3.725) {$\bullet$};
\node (i4) at (0.725,2.225) {$\bullet$};
\node (i5) at (2.225,3.725) {$\bullet$};
\node (i6) at (2.225,2.225) {$\bullet$};
\node (i7) at (4.5,3) {$\bullet$};
\node (i8) at (0.725,5.225) {$\bullet$};
\draw [line width = 0.5mm] (i0.center) to [out=90,in=-180] (i3.center);
\draw [line width = 0.5mm] (i3.center) to [out=-0,in=90] (i1.center);
\draw [line width = 0.5mm] (i1.center) to [out=-90,in=0] (i4.center);
\draw [line width = 0.5mm] (i4.center) to [out=180,in=-90] (i0.center);
\draw [line width = 0.5mm] (i1.center) to [out=90,in=-180] (i5.center);
\draw [line width = 0.5mm] (i5.center) to [out=-0,in=90] (i2.center);
\draw [line width = 0.5mm] (i2.center) to [out=-90,in=0] (i6.center);
\draw [line width = 0.5mm] (i6.center) to [out=180,in=-90] (i1.center);
\draw [line width = 0.5mm] (i2.center) to [out=45,in=135] (i7.center);
\draw [line width = 0.5mm] (i2.center) to [out=-45,in=-135] (i7.center);
\draw [line width = 0.5mm] (i3.center) to [out=45,in=-45] (i8.center);
\draw [line width = 0.5mm] (i3.center) to [out=135,in=-135] (i8.center);
\end{tikzpicture}
\caption{A possible diagram at perturbation order $p=12$.}\label{fig:diagram_ex1_appendix}
\end{subfigure}
\begin{subfigure}[b]{0.42\textwidth}
   \centering
\begin{tikzpicture}[scale=0.8]
\node (i1) at (0,3) {$\bullet$};
\node (i2) at (3,3) {$\bullet$};
\node (i3) at (1.5,5.598) {$\bullet$};
\node (i4) at (4.5,5.598) {$\bullet$};
\draw [line width = 0.5mm] (i1.center) to [out=45,in=135] (i2.center);
\draw [line width = 0.5mm] (i1.center) to [out=-45,in=-135] (i2.center);
\draw [line width = 0.5mm] (i1.center) to [out=-20,in=-160] (i2.center);
\draw [line width = 0.5mm] (i2.center) to [out=100,in=-40] (i3.center);
\draw [line width = 0.5mm] (i3.center) to [out=-140,in=80] (i1.center);
\draw [line width = 0.5mm] (i2.center) to [out=40,in=-100] (i4.center);
\draw [line width = 0.5mm] (i4.center) to [out=160,in=20] (i3.center);
\draw [line width = 0.5mm] (i3.center) to [out=-80,in=140] (i2.center);
\end{tikzpicture}
\caption{Another diagram at order $p = 8$}\label{fig:diagram_ex2_appendix}
\end{subfigure}
\caption{Cactus and non-cactus diagrams. Each vertex represents an index $i$ over which we sum, and each edge is a factor $J_{ij}$. Each connected component of the diagrams carries a global $\frac{1}{N}$ factor.}\label{fig:diagrams_example}
\end{figure}
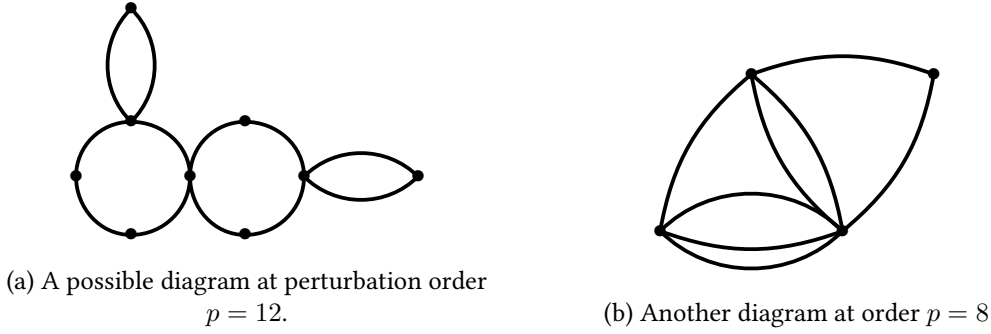
For instance the diagram depicted in Fig.~\ref{fig:diagram_ex1_appendix} is equal to:
\begin{align}
  \frac{1}{N} \sum_{\substack{i_1,\cdots,i_9 \\ \text{pairwise distincts}}} \left(J_{i_1 i_2} J_{i_2 i_3} J_{i_3 i_4} J_{i_4 i_1}\right)\left(J_{i_3 i_5} J_{i_5 i_6} J_{i_6 i_7} J_{i_7 i_3}\right) J_{i_6 i_8}^2 J_{i_2 i_9}^2.
\end{align}
The perturbation order of any diagram is equal to its number of edges, since each of them represents a factor $J_{ij}$.
In this whole section we will only consider \emph{connected} diagrams (unless stated otherwise).
The structure of the section is the following:
\begin{itemize}
  \item In Sec.~\ref{subsec:freecum_expectation} we prove a first rigorous result on 
  the `simple cycles' arising in the Plefka expansion of Sec.~\ref{sec:spherical_bipartite}, namely we study 
  these diagrams \emph{in expectation over $J$} and show a weaker version of Theorem~\ref{thm:free_cum}.
  \item In Sec.~\ref{subsec:generic_diagrams_expectation} we extend this study to all possible diagrams,
  in expectation over $J$.
  \item In Sec.~\ref{subsec:concentration_J} we show how the results of Sec.~\ref{subsec:freecum_expectation} and 
  Sec.~\ref{subsec:generic_diagrams_expectation} can be extended to study the second moments of these diagrams, and use
  it to show concentration results. This will in particular imply the full statement of Theorem~\ref{thm:free_cum}.
  \item In Sec.~\ref{subsec:higher_order_symmetric} we explain how to handle the higher-order moments that can appear as additional factors 
  in these diagrams for the statistical models studied in Sec.~\ref{sec:stat_models}.
  \item In Sec.~\ref{subsec:diagrammatics_bipartite_models} we explain how to generalize all these techniques and results to diagrams made of rectangular matrices, that arise
  in the Plefka expansion for bipartite models.
  \item Finally, in Sec.~\ref{subsec:iid_matrices} we show that if one considers an i.i.d.\ coupling matrix, all the diagrams of 
  order greater than $3$ will not contribute in the thermodynamic limit and that one can effectively consider the distribution of the 
  matrix elements to be Gaussian.
\end{itemize}
Some technicalities, as well as side results and generalizations of these diagrammatics for Hermitian matrices and diverging-size diagrams, which are not directly useful for our expansions, 
are detailed in Appendix~\ref{sec:generalizations_expansions}.
We finally note that some of our results are similar to the recent independent work of \cite{bauer2018equilibrium} that was recently brought to our attention.

\subsection{A weaker version of Theorem~\ref{thm:free_cum}}\label{subsec:freecum_expectation}

We will consider the random matrix ensemble defined by Model~\ref{model:sym_rot_inv}. 
In the following, $J \in {\cal S}_N$ is a random matrix from this ensemble.
Recall as well the random matrix tools defined in Appendix~\ref{sec:appendix_rmt}, in particular the free cumulants $\{c_p(\rho_D)\}$.
We first show a weaker version of Theorem~\ref{thm:free_cum}:
 \begin{theorem}[Expectation of simple cycles and free cumulants]\label{thm:free_cum_expectation}
 For $J$ following Model~\ref{model:sym_rot_inv}, for any $p \geq 1$, and any set of pairwise distinct indices $i_1,\cdots,i_r \in \bbN^p$, one has:
 \begin{align}\label{eq:freecum_expectation_nonsummed}
 \lim_{N \to \infty} \EE \left[N^{p-1} J_{i_1 i_2} J_{i_2 i_3} \cdots J_{i_{p-1} i_p} J_{i_p i_1}\right] &= c_p(\rho_D).
 \end{align}
 A stronger result actually takes place, that is we only need to average over $O$ to obtain the result:
 \begin{align}\label{eq:conjecture}
 \lim_{N\to \infty} N^{p-1}  \int_{\mathcal{O}(N)} \mathcal{D}O \left[\left(O D O^\intercal\right)_{i_1 i_2} \left(O D O^\intercal\right)_{i_2 i_3} \cdots \left(O D O^\intercal\right)_{i_p i_1} \right]&=  c_p(\rho_D).
 \end{align}
 This last equality is true a.s.\ with respect to the law of $D$.
 \end{theorem}
 Note that in the Plefka expansions we perform in Sec.~\ref{subsubsec:plefka_sym_spherical} and Sec.~\ref{subsec:plefka_sym_models} we consider sums over all distinct pairwise indices of eq.~\eqref{eq:freecum_expectation_nonsummed}. 
 The expectation of these sums over $O$ is an immediate consequence of Theorem~\ref{thm:free_cum_expectation}:
\begin{align*}
\forall p \in \bbN^\star, \quad  \lim_{N \to \infty} \EE_O \left[\frac{1}{N} \sum_{\substack{i_1,\cdots,i_p \\ \mathrm{pairwise}\,\mathrm{distincts}}} J_{i_1 i_2} J_{i_2 i_3} \cdots J_{i_{p-1} i_p} J_{i_p i_1}\right] &= c_p(\rho_D).
 \end{align*}
 We now turn to the proof of Theorem~\ref{thm:free_cum_expectation}.
 \begin{proof}[Proof of Theorem~\ref{thm:free_cum_expectation}]
 A first pedestrian way to show eq.~\eqref{eq:conjecture} for small values of $p$ is to use explicit integration of polynomials over the Haar measure of the orthogonal or unitary group, see for instance \cite{collins2006integration}. This can be used
 to check eq.~\eqref{eq:conjecture} for the first values of $p$.
 Since we aim at a generic proof we will choose a different path, leveraging from HCIZ-type integrals \cite{harish1957differential} \cite{itzykson1980planar}, in the particular case
 in which one matrix has finite rank. In our setting, the computation of these integrals has been made rigorous in \cite{guionnet2005fourier}.
Let us denote:
 \begin{align}\label{eq:def_Lp}
  L_p^{(N)} &\equiv N^{p-1} \int_{\mathcal{O}(N)} \mathcal{D}O \left[\left(O D O^\intercal\right)_{i_1 i_2} \left(O D O^\intercal\right)_{i_2 i_3} \cdots \left(O D O^\intercal\right)_{i_p i_1} \right].
 \end{align}
 In order to simplify the following calculation, we assume that $(i_1,\cdots,i_p) = (1,\cdots,p)$. 
 Since the sought result does not depend on the particular choice of indices (as is clear by rotational invariance), this does not remove any generality. 
 We first note that the case $p=1$ and $p=2$ are trivial to show by an explicit computation, so we will assume $p \geq 3$ in the following.
 One can rewrite eq.~\eqref{eq:def_Lp} as:
 \begin{align}
 L_p^{(N)} &= \frac{1}{N} \prod_{l=1}^p \frac{\partial}{\partial b_l} \left[\int_{\mathcal{O}(N)} \mathcal{D}O \, e^{\frac{N}{2} \mathrm{Tr}\, \left[M(\bb) O D O^\intercal \right]}\right]_{\bb=0},
 \end{align}
 in which we denoted $\bb \equiv (b_1,\cdots,b_p)$ and $M(\bb) \in {\cal S}_N$ the following symmetric block matrix of rank $p$:
 \begin{align}
 M(\bb) \equiv \begin{pmatrix}
 M_1(\bb) & (0) \\
 (0) & (0) \\
 \end{pmatrix},
 \end{align}
 in which $M_1(\bb) \in {\cal S}_p$ with:
 \begin{align}
 M_1(\bb) \equiv  \begin{pmatrix}
 0 & b_1 & 0 & \cdots & 0 & b_p \\
 b_1  & 0 & b_2 & \cdots & 0& 0 \\
 0 & b_2 & 0 & \cdots & 0 &0 \\
 \vdots & \vdots & \vdots & \ddots & \vdots & \vdots \\
 0 & 0 & 0 & \cdots & 0 & b_{p-1} \\
 b_p & 0 & 0 & \cdots & b_{p-1}  & 0
 \end{pmatrix}.
 \end{align}
 Now we can apply Theorem~2 of \cite{guionnet2005fourier}. Recall that $G_{\rho_D}$ is (up to a factor)
  the integrated $\mathcal{R}$-transform of $\rho_D$.
 We obtain:
 \begin{align}
  \lim_{N \to \infty} L_p^{(N)} &= \lim_{N \to \infty} \frac{1}{N} \left[\prod_{l=1}^p \frac{\partial}{\partial b_l}\right]\left[\exp\left\{N \mathrm{Tr} \, G_{\rho_D}\left[M(\bb)\right] \right\}\right]_{\bb=0}, \nonumber \\
 \label{eq:Lp_guionnet}
 &= \lim_{N \to \infty}  \frac{1}{N} \left[\prod_{l=1}^p \frac{\partial}{\partial b_l}\right]\left[\exp\left\{\frac{N}{2}\sum_{n=1}^\infty \frac{c_n(\rho_D)}{n} \mathrm{Tr} \, [M(\bb)^n] \right\}\right]_{\bb=0}.
 \end{align}
 Let us denote $Z(\bb) \equiv\exp\left\{\frac{N}{2}\sum_{n=1}^\infty \frac{c_n(\rho_D)}{n} \mathrm{Tr} \, [M(\bb)^n] \right\}$.
 Note that differentiating $Z(\bb)$ with respect to $b_1$ yields (by cyclicity of the trace):
 \begin{align}
 \frac{1}{Z(\bb)}\frac{\partial}{\partial b_1} Z(\bb)&= \frac{N}{2} \sum_{n=1}^\infty c_n(\rho_D) \mathrm{Tr} \, \left[\left\{\frac{\partial}{\partial b_1} M(\bb) \right\}  M(\bb)^{n-1}\right] , \\
 \label{eq:1}
 &= \frac{N}{2}\sum_{n=1}^\infty c_n(\rho_D) \mathrm{Tr} \, \left[E_{12} M(\bb)^{n-1}\right],
 \end{align}
 with elementary symmetric matrices $(E_{ab})_{ll'} \equiv \delta_{l,a} \delta_{l',b} + \delta_{l',a} \delta_{l,b}$. These matrices are such that for 
 each $a < b$ and $c < d$, $E_{ab} E_{cd} = 0$ if $\{c,d\} \cap \{a,b\} =\emptyset$.
The only way to obtain a matrix of non-zero trace with a product of matrices $\{E_{ab}\}$ is to have a cycle structure in the indices of the matrices.
Recall that the indices are symmetric, that is $E_{ba} = E_{ab}$.
For instance:
\begin{align*}
  \mathrm{Tr} \, \left[E_{12}^2 E_{13}E_{23}E_{12}\right] &= \mathrm{Tr} \, \left[E_{12} E_{21} E_{13}E_{32}E_{21}\right]  \neq 0 ,\\
  \mathrm{Tr} \, \left[E_{12}^2 E_{24}E_{23}E_{12}\right] &= 0 .
\end{align*}
 Because of this and the fact that $M(\bb = 0) = 0$, 
 it is easy to see that the only term that will survive after taking all the
 successive derivatives and taking $\bb = 0$ will be the derivatives of the
 right-hand-side of eq.~\eqref{eq:1}, and not other derivatives of $Z(\bb)$.
Let us analyze what differentiating this term yields.
As we saw, taking derivative with respect to $b_1$ yields a matrix $E_{12}$. 
When differentiating with respect to $b_2$ this yields a matrix $E_{23}$.
Note that \emph{a priori}, one would have:
\begin{align}\label{eq:diff_M}
  \frac{\partial}{\partial b_2} \mathrm{Tr}\left[E_{12} M(\bb)^{n-1}\right] &=  \mathrm{Tr} \, \left[E_{12} \sum_{k=0}^{n-2} M(\bb)^{k} E_{23} M(\bb)^{n-2-k}\right].
\end{align}
However, the following differentiations with respect to $b_3,\cdots,b_p$ will never yield a matrix $E_{ab}$ with one
of the indices being equal to $2$. So in eq.~\eqref{eq:diff_M} it is clear that only two terms of the sum, the term $k=0$ and $k=n-2$, 
will yield a non-zero contribution.
In the end, after taking all the $p$ successive derivatives, 
only two terms will remain, which correspond to the two possible orientations of the simple cycle:
\begin{align}
   \label{eq:cycles_sym}
  \lim_{N \to \infty}L_p^{(N)} &= \frac{1}{2}\sum_{n=p}^\infty  c_n(\rho_D) \mathrm{Tr}\left[\left(E_{12} E_{23} \cdots E_{p1} + E_{1p} E_{p p-1} \cdots E_{32} E_{21} \right) \, M(0)^{n-p} \right], \\
      &= c_p(\rho_D), \nonumber
\end{align}
using that $M(0) = 0$, which finishes the proof.
 \end{proof}

\subsection{The expectation of generic diagrams}\label{subsec:generic_diagrams_expectation}

\begin{figure}[t]
 \centering
\captionsetup{justification=centering}
\begin{subfigure}[b]{0.32\textwidth}
   \centering
\begin{tikzpicture}[scale=1.]
\node (i1) at (0,2) {$\bullet$};
\node (i2) at (1,3) {$\bullet$};
\node (i3) at (2,2) {$\bullet$};
\node (i4) at (1,1) {$\bullet$};
\draw [line width = 0.5mm] (i1.center) to [out=90,in=-180] (i2.center);
\draw [line width = 0.5mm] (i2.center) to [out=-0,in=90] (i3.center);
\draw [line width = 0.5mm] (i3.center) to [out=-90,in=0] (i4.center);
\draw [line width = 0.5mm] (i4.center) to [out=180,in=-90] (i1.center);
\draw [line width = 0.5mm] (i4.center) to [out=90,in=-90] (i2.center);
\draw [line width = 0.5mm] (i1.center) to [out=0,in=180] (i3.center);
\end{tikzpicture}
\caption{A non-Eulerian diagram.}\label{fig:ncycle}
\end{subfigure}
\begin{subfigure}[b]{0.32\textwidth}
   \centering
\begin{tikzpicture}[scale=1.]
\node (i1) at (0,2) {$\bullet$};
\node (i2) at (1,3) {$\bullet$};
\node (i3) at (2,2) {$\bullet$};
\node (i4) at (1,1) {$\bullet$};
\draw [line width = 0.5mm] (i1.center) to [out=90,in=-180] (i2.center);
\draw [line width = 0.5mm] (i2.center) to [out=-0,in=90] (i3.center);
\draw [line width = 0.5mm] (i3.center) to [out=-90,in=0] (i4.center);
\draw [line width = 0.5mm] (i4.center) to [out=180,in=-90] (i1.center);
\draw [line width = 0.5mm] (i4.center) to [out=45,in=-45] (i2.center);
\draw [line width = 0.5mm] (i4.center) to [out=135,in=-135] (i2.center);
\end{tikzpicture}
\caption{An Eulerian strongly irreducible diagram.}\label{fig:ncactus_appendix}
\end{subfigure}
\begin{subfigure}[b]{0.32\textwidth}
   \centering
\begin{tikzpicture}[scale=1.]
\node (i0) at (0,3) {$\bullet$};
\node (i1) at (1.5,3) {$\bullet$};
\node (i2) at (3,3) {$\bullet$};
\node (i3) at (0.725,3.725) {$\bullet$};
\node (i4) at (0.725,2.225) {$\bullet$};
\node (i5) at (2.225,3.725) {$\bullet$};
\node (i6) at (2.225,2.225) {$\bullet$};
\node (i7) at (4.5,3) {$\bullet$};
\node (i8) at (0.725,5.225) {$\bullet$};
\draw [line width = 0.5mm] (i0.center) to [out=90,in=-180] (i3.center);
\draw [line width = 0.5mm] (i3.center) to [out=-0,in=90] (i1.center);
\draw [line width = 0.5mm] (i1.center) to [out=-90,in=0] (i4.center);
\draw [line width = 0.5mm] (i4.center) to [out=180,in=-90] (i0.center);
\draw [line width = 0.5mm] (i1.center) to [out=90,in=-180] (i5.center);
\draw [line width = 0.5mm] (i5.center) to [out=-0,in=90] (i2.center);
\draw [line width = 0.5mm] (i2.center) to [out=-90,in=0] (i6.center);
\draw [line width = 0.5mm] (i6.center) to [out=180,in=-90] (i1.center);
\draw [line width = 0.5mm] (i2.center) to [out=45,in=135] (i7.center);
\draw [line width = 0.5mm] (i2.center) to [out=-45,in=-135] (i7.center);
\draw [line width = 0.5mm] (i3.center) to [out=45,in=-45] (i8.center);
\draw [line width = 0.5mm] (i3.center) to [out=135,in=-135] (i8.center);
\end{tikzpicture}
\caption{A cactus diagram.}\label{fig:cactus_appendix}
\end{subfigure}
\caption{Cactus and non-cactus diagrams. Each vertex represents an index $i$ over which we sum, and each edge is a factor $J_{ij}$. Each diagram carries a global $\frac{1}{N}$ factor.}\label{fig:diagrams_cactus_ncactus}
\end{figure}
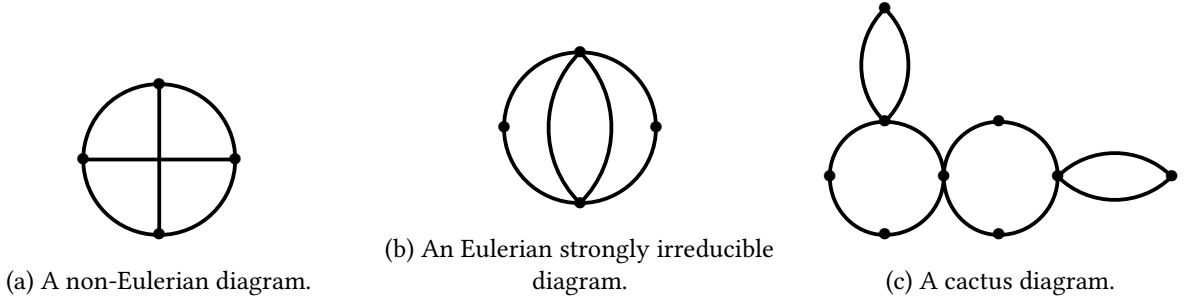
Following the remarks of \cite{georges1991expand} and \cite{parisi1995mean}, we can separate some of the diagrams constructed as in Fig.~\ref{fig:diagrams_example} in
three disjoint categories or types:
\begin{enumerate}[label=\textbf{T.\arabic*},ref=T.\arabic*]
  \item \label{type:non_eulerian} \emph{Non-Eulerian diagrams.} By definition, a diagram is Eulerian if one can construct a cyclic path in the graph 
   that goes through each edge exactly once. Note that this is a classic result of graph theory (the Euler–Hierholzer theorem) that these graphs are exactly the connected graphs with even degree
   in each vertex. For instance, the graph depicted in eq.~\eqref{fig:ncycle} is not Eulerian, whereas the one of Fig.~\ref{fig:ncactus_appendix} is Eulerian.
   \item \label{type:strongly_irreducible}
   \emph{Eulerian diagrams that are strongly irreducible but not simple cycles}. By strongly irreducible, we mean \cite{georges1991expand} that one can 
   not make it disconnected by removing any single vertex. For instance, the diagram of Fig.~\ref{fig:ncactus_appendix} is strongly irreducible, whereas 
   the diagram of Fig.~\ref{fig:cactus_appendix} is not.
   \item \label{type:cactus} \emph{Cactus diagrams}. These diagrams, like the one of Fig.~\ref{fig:cactus_appendix}, are trees made of simple cycles joining at 
   their vertices. Among them are of course the \emph{simple cycles}.
\end{enumerate}
We are not interested in Eulerian diagrams that are not strongly irreducible. Indeed, as argued in \cite{georges1991expand}, only strongly irreducible
diagrams will appear in the Plefka expansions. This is an important hypothesis of the Plefka expansion, somehow a bit hidden by the formalism.
We give precise descriptions of the large $N$ limit of the \emph{expectation} of all these diagrams in the following.
When we write ``expectation'' we will always mean expectation over the orthogonal matrix of Model~\ref{model:sym_rot_inv}.
More precisely, we will show:
\begin{enumerate}
   \item[$(i)$] All non-Eulerian graphs of type \ref{type:non_eulerian} have a vanishing expectation in the $N \to \infty$ limit.
   \item[$(ii)$] All strongly irreducible diagrams of type \ref{type:strongly_irreducible} also have a vanishing expectation in the $N \to \infty$ limit.
   \item[$(iii)$] We already showed that the expectation of a simple cycle of size $p$ converges
    to the $p$-th free cumulant of $\rho_D$ in Sec.~\ref{subsec:freecum_expectation}. We show that the expectation of a cactus diagram converges to the 
    product of the expectations of all its constituent simple cycles. For instance, for the diagram $\mathcal{C}$ of Fig.~\ref{fig:cactus_appendix} we obtain
    that its expectation converges to: 
   \begin{align}
     \lim_{N \to \infty} \EE \, \mathcal{C} &= c_2(\rho_D)^2 c_4(\rho_D)^2. 
   \end{align}
\end{enumerate}
Results $(i)$ and $(ii)$ are justified in Sec.~\ref{subsubsec:neulerian_ncactus_expectation}, and are directly useful for our diagrammatic expansions.
Result $(iii)$ on the other hand is a side result that is not used in our expansions, as we argued that only strongly irreducible diagrams come up in our expansions \cite{georges1991expand}.
It is justified in Sec.~\ref{subsubsec:cactus_expectation}. 

\subsubsection{Eulerian diagrams, strongly irreducible diagrams and simple cycles}\label{subsubsec:neulerian_ncactus_expectation}

Let us consider a connected diagram $G$ with $V$ vertices and $E$ edges. We will show that:
\begin{itemize}
  \item If $G$ is not Eulerian, its expectation goes to $0$ as $N \to \infty$.
  \item If $G$ is Eulerian and strongly irreducible, but is not a simple cycle, its expectation
  also goes to $0$ as $N \to \infty$.
\end{itemize}
Once averaged over the orthogonal matrices, the permutation invariance of the indices allows us to write
\begin{align}
  \EE \, G &= N^{V-1} \int_{\mathcal{O}(N)} \mathcal{D}O \prod_{1 \leq l < l' \leq V} \left(O D O^\intercal\right)_{ll'}^{\epsilon_{ll'}}, 
\end{align}
in which the $\epsilon_{ll'}$ are positive integers such that $\sum_{l < l'} \epsilon_{ll'} = E$. We can now use the results of 
\cite{guionnet2005fourier}, as we did in Sec.~\ref{subsec:freecum_expectation}, to write this diagram as (in the $N \to \infty$ limit):
\begin{align}\label{eq:decomposition_ncactus}
\EE \,G = N^{V-E-1}\left[ \prod_{l < l'}\frac{\partial^{\epsilon_{ll'}}}{\partial b_{ll'}^{\epsilon_{ll'}}}\right] \left[\exp \left\{\frac{N}{2}\sum_{n=1}^\infty \frac{c_n(\rho_D)}{n} \mathrm{Tr}\, [M(\bb)]^n\right\}\right]_{\bb = 0}. 
\end{align} 
In this expression, $M(\bb)_{ll'} \equiv b_{ll'} = M(\bb)_{l'l}$ for $l < l'$,and the diagonal is zero: $M(\bb)_{ll} = 0$. 
Exactly as in Sec.~\ref{subsec:freecum_expectation}, the elementary matrices $\{E_{ll'}\}$ will appear in eq.~\eqref{eq:decomposition_ncactus} by successive
derivatives of the exponential, using the fact that $\frac{\partial}{\partial b_{ll'}} M(\bb) = E_{ll'}$ and then using $M(\bb = 0) = 0$.
As we explained in Sec.~\ref{subsec:freecum_expectation}, a trace of the products of the $\{E_{ll'}\}$ matrices will only be non-zero if and only if the indices in the products form a cycle.
Moreover, as is clear in eq.~\eqref{eq:decomposition_ncactus}, the terms corresponding to the decomposition of $\EE \,G$ into the maximum number of such 
 cycles will dominate in the large $N$ limit, as each derivation of the exponential term adds a multiplicative factor $N$ 
\footnote{There might be a confusion, so we emphasize that this ``decomposition'' of $\EE \, G$ is a decomposition of the \emph{graph} representing $\EE \, G$.}. 
These two facts together imply that:
\begin{itemize}
  \item If $G$ is not Eulerian, as in Fig.~\ref{fig:ncycle}, its expectation will be $0$ in the limit $N \to \infty$
  since it is not possible to decompose it into disjoint cycles by definition.
  \item If $G$ is Eulerian, strongly irreducible, but not a simple cycle, 
      the dominant contribution to $\EE \,G$ in eq.~\eqref{eq:decomposition_ncactus} will arise from decomposing the graph $G$ 
      into simple cycles, as this decomposition maximizes the number of cycles, and we already showed that each simple cycle has a $\mathcal{O}_N(1)$ contribution.
      For the graph of Fig.~\ref{fig:ncactus_appendix}, we show two such possible decompositions in Fig.~\ref{fig:ncactus_decomposition_appendix}.
\end{itemize}
\begin{figure}[t]
 \centering
\captionsetup{justification=centering}
\begin{subfigure}[b]{0.42\textwidth}
   \centering
   $
   \left(
\begin{tikzpicture}[baseline={([yshift=-.5ex]current bounding box.center)},scale=1.]
\node (i0) at (0,0) {$\bullet$};
\node (i1) at (1.5,0) {$\bullet$};
\draw [line width = 0.5mm] (i0.center) to [out=30,in=150] (i1.center);
\draw [line width = 0.5mm] (i0.center) to [out=-30,in=-150] (i1.center);
\end{tikzpicture}
;
\begin{tikzpicture}[baseline={([yshift=-.5ex]current bounding box.center)},scale=1.]
\node (i0) at (0,0) {$\bullet$};
\node (i1) at (1,1) {$\bullet$};
\node (i2) at (2,0) {$\bullet$};
\node (i3) at (1,-1) {$\bullet$};
\draw [line width = 0.5mm] (i0.center) to [out=90,in=180] (i1.center);
\draw [line width = 0.5mm] (i1.center) to [out=0,in=90] (i2.center);
\draw [line width = 0.5mm] (i2.center) to [out=-90,in=0] (i3.center);
\draw [line width = 0.5mm] (i3.center) to [out=180,in=-90] (i0.center);
\end{tikzpicture}
\right)
$
\end{subfigure}
\begin{subfigure}[b]{0.42\textwidth}
   \centering
   $
   \left(
\begin{tikzpicture}[baseline={([yshift=-.5ex]current bounding box.center)},scale=1.]
\node (i0) at (0,0) {$\bullet$};
\node (i1) at (0.866,1.5) {$\bullet$};
\node (i2) at (1.732,0) {$\bullet$};
\draw [line width = 0.5mm] (i0.center) to [out=80,in=-140] (i1.center);
\draw [line width = 0.5mm] (i1.center) to [out=-40,in=100] (i2.center);
\draw [line width = 0.5mm] (i2.center) to [out=-160,in=-20] (i0.center);
\end{tikzpicture}
;
\begin{tikzpicture}[baseline={([yshift=-.5ex]current bounding box.center)},scale=1.]
\node (i0) at (0,0) {$\bullet$};
\node (i1) at (0.866,1.5) {$\bullet$};
\node (i2) at (1.732,0) {$\bullet$};
\draw [line width = 0.5mm] (i0.center) to [out=80,in=-140] (i1.center);
\draw [line width = 0.5mm] (i1.center) to [out=-40,in=100] (i2.center);
\draw [line width = 0.5mm] (i2.center) to [out=-160,in=-20] (i0.center);
\end{tikzpicture}
\right)
$
\end{subfigure}
\caption{Two possible decompositions of the diagram of Fig.~\ref{fig:ncactus_appendix} into simple cycles.}\label{fig:ncactus_decomposition_appendix}
\end{figure}
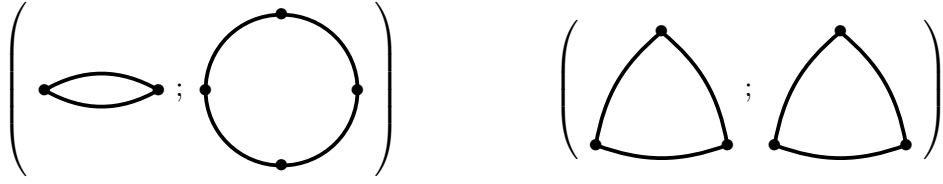
Given the remarks above we assume now that $G$ is Eulerian and strongly irreducible.
Let us denote $P$ the maximal number of simple cycles in such a decomposition of the graph $G$.
Then one can see that the scaling of eq.~\eqref{eq:decomposition_ncactus} will be:
\begin{align*}
  \EE \, G &\sim N^{V+P-E-1}. 
\end{align*}
One can easily be convinced that for a strongly irreducible diagram $G$ we have $V+P-E-1 \leq 0$, and we have equality only if $G$ is a simple cycle.
This implies that all the strongly irreducible diagrams that are not simple cycles and that appear in our Plefka expansions in Sec.~\ref{sec:spherical_bipartite} and
Sec.~\ref{sec:stat_models} will not contribute in the $N \to \infty$ limit.

\subsubsection{Cactus diagrams}\label{subsubsec:cactus_expectation}

As a side result, although it's not directly useful for our Plefka expansions, 
we show that we can compute the large $N$ limit of any ``cactus'' \cite{parisi1995mean} diagram (like the one of Fig.~\ref{fig:cactus_appendix})
as a function of the free cumulants of $\rho_D$.
The argument is straightforward and uses the same technique as in Sec.~\ref{subsubsec:neulerian_ncactus_expectation}.
Consider a cactus diagram $G$ with $V$ vertices and $E$ edges.
One can write the same equation as eq.~\eqref{eq:decomposition_ncactus}:
\begin{align}\label{eq:decomposition_cactus}
\EE \, G = N^{V-E-1}\left[ \prod_{l < l'}\frac{\partial^{\epsilon_{ll'}}}{\partial b_{ll'}^{\epsilon_{ll'}}}\right] \left[\exp \left\{\frac{N}{2}\sum_{n=1}^\infty \frac{c_n(\rho_D)}{n} \mathrm{Tr}\, [M(\bb)]^n\right\}\right]_{\bb = 0}. 
\end{align} 
Again, the dominant contribution is obtained by decomposing $G$ in as many simple cycles as possible. For a cactus diagram 
it is easy to see that there is only \emph{one} such decomposition, which corresponds to its natural decomposition into its
constituent simple cycles, and that the number of such cycles is $P = E+V-1$. Let us denote $\{r_1,\cdots,r_P\}$ the number of vertices
in each of these $P$ simple cycles.
The dominant contribution corresponds to differentiating $P$ times inside the exponential of eq.~\eqref{eq:decomposition_cactus}. Using exactly 
the argument of Sec.~\ref{subsec:freecum_expectation} for each of the $P$ simple cycles we finally obtain:
\begin{align}
 \EE \, G &= N^{P+V-E-1} \prod_{\alpha=1}^P c_{r_\alpha}(\rho_D) + \smallO_N(1), \nonumber\\
  \EE \, G &= \prod_{\alpha=1}^P c_{r_\alpha}(\rho_D) + \smallO_N(1).
\end{align}
This justifies the point $(iii)$ that we gave in the introductory part of the section:
the expectation of the cactus diagrams decouple into the products of their simple cycles constituents.

\subsection{Concentration of the diagrams: a second moment analysis}\label{subsec:concentration_J}

Using our first moment results of Sec.~\ref{subsec:freecum_expectation} and Sec.~\ref{subsec:generic_diagrams_expectation}, 
we will show the following results:
\begin{itemize}
  \item[$(i)$] If ${\cal C}_p$ is the simple cycle of order $p$, then we have that $\lim_{N \to \infty} {\cal C}_p \overset{L^2}{=} c_p(\rho_D)$, which 
   implies directly Theorem~\ref{thm:free_cum} and thus ends its proof.
  Moreover, if $G$ is a cactus diagram then it converges in $L^2$ to the products of the free cumulants corresponding to its constituent
  simple cycles.
  \item[$(ii)$] If $G$ is of the type \ref{type:non_eulerian} or \ref{type:strongly_irreducible}, we have:
  \begin{align}
    \lim_{N \to \infty} \EE\, G^2 &= 0. 
  \end{align}
  This implies that the diagram $G$ will be negligible in the $N \to \infty$ limit.
\end{itemize}
Note that following the arguments of \cite{georges1991expand}, one can convince oneself that \emph{only} strongly irreducible diagrams
will contribute in general to the expansion in our models. 
Together with point $(ii)$
this shows in more detail why only the simple cycles contribute in our Plefka expansions, 
like in eq.~\eqref{eq:spherical_full} for the spherical model of Sec.~\ref{subsec:sym_spherical_model}.
In order to show $(i)$ and $(ii)$ we will establish the following fact. 
Consider a diagram $G$ with $V$ vertices and $E$ edges, of any of the types \ref{type:non_eulerian}, \ref{type:strongly_irreducible}, or \ref{type:cactus}.  
Then one has:
\begin{align}\label{eq:decomposition_second_moment}
 \EE \, G^2 &= \left(\EE G\right)^2 + \frac{1}{N} \sum_\alpha \EE \, \mathcal{C}_\alpha + \smallO_N(1).
\end{align}
In this formula, the sum $\sum_\alpha \mathcal{C}_\alpha$ represents \emph{all the possible diagrams} that one can obtain by
`gluing' together two replicas of the diagram $G$. Indeed, one can write the generic form of a diagram $G$ as:
\begin{align*}
  G  &= \frac{1}{N} \sum_{\substack{i_1,\cdots,i_V \\ \text{pairwise distincts}}} \prod_{1 \leq l < l' \leq V} J_{i_l i_{l'}}^{\epsilon_{ll'}}, 
\end{align*}
in which the integers $\epsilon_{ll'}$ verify $\sum_{l < l'} \epsilon_{ll'} = E$.
Thus one has:
\begin{align*}
  \EE \, G^2 &= \EE \, \left[\frac{1}{N^2} \sum_{\substack{i_1,\cdots,i_V \\ \text{pairwise distincts}}} \sum_{\substack{j_1,\cdots,j_V \\ \text{pairwise distincts}}}\prod_{1 \leq l < l' \leq V}  J_{i_l i_{l'}}^{\epsilon_{ll'}} J_{j_l j_{l'}}^{\epsilon_{ll'}} \right].
\end{align*}
In this expression, one can see that two types of terms have to be taken into account:
\begin{itemize}
  \item A term for which \emph{all indices} $\{i_1,\cdots,i_V,j_1,\cdots,j_V\}$ are pairwise distinct. Diagrammatically, 
  this corresponds to a graph with two disconnected components that are identical and equal to $G$. Therefore, one can repeat the arguments of Sec.~\ref{subsec:freecum_expectation} and 
  Sec.~\ref{subsec:generic_diagrams_expectation} straightforwardly. Indeed, as all the indices are distincts, the decomposition of this diagram into the maximum number of simple cycles will be two 
  copies of the maximal decomposition of $G$.
  This yields that this term is equal in the $N \to \infty$ limit to $(\EE \, G)^2$.
  \item Terms for which there is at least one equality of the type $i_l = j_{l'}$ for $1 \leq l,l' \leq V$. Such a term thus corresponds to a diagram
  with a \emph{single} connected component and constructed by `gluing' some of the vertices of two identical copies of $G$. Since these diagrams have a single
  connected component, they carry a single $\frac{1}{N}$ factor, which explains the term $\frac{1}{N} \sum_\alpha \EE \, \mathcal{C}_\alpha$ in eq.~\eqref{eq:decomposition_second_moment}, if we
  denote $\mathcal{C}_\alpha$ each of these possible terms.
\end{itemize}
We give a schematic representation of eq.~\eqref{eq:decomposition_second_moment} for a simple cycle in Fig.~\ref{fig:decomposition_second_moment}.
\begin{figure}[t]
 \centering
\captionsetup{justification=centering}
   $\EE \left(
\begin{tikzpicture}[baseline={([yshift=-.5ex]current bounding box.center)},scale=0.5]
\node (i0) at (0,0) {$\bullet$};
\node (i1) at (0.866,1.5) {$\bullet$};
\node (i2) at (1.732,0) {$\bullet$};
\draw [line width = 0.5mm] (i0.center) to [out=80,in=-140] (i1.center);
\draw [line width = 0.5mm] (i1.center) to [out=-40,in=100] (i2.center);
\draw [line width = 0.5mm] (i2.center) to [out=-160,in=-20] (i0.center);
\end{tikzpicture}\right)^2
 =  \left( \EE \,
\begin{tikzpicture}[baseline={([yshift=-.5ex]current bounding box.center)},scale=0.5]
\node (i0) at (0,0) {$\bullet$};
\node (i1) at (0.866,1.5) {$\bullet$};
\node (i2) at (1.732,0) {$\bullet$};
\draw [line width = 0.5mm] (i0.center) to [out=80,in=-140] (i1.center);
\draw [line width = 0.5mm] (i1.center) to [out=-40,in=100] (i2.center);
\draw [line width = 0.5mm] (i2.center) to [out=-160,in=-20] (i0.center);
\end{tikzpicture}\right)^2
+ \frac{1}{N} \left[9 \, \EE \, 
\begin{tikzpicture}[baseline={([yshift=-.5ex]current bounding box.center)},scale=0.5]
\node (i0) at (0,0) {$\bullet$};
\node (i1) at (0.866,1.5) {$\bullet$};
\node (i2) at (1.732,0) {$\bullet$};
\node (i3) at (2.598,1.5) {$\bullet$};
\node (i4) at (3.464,0) {$\bullet$};
\draw [line width = 0.5mm] (i0.center) to [out=80,in=-140] (i1.center);
\draw [line width = 0.5mm] (i1.center) to [out=-40,in=100] (i2.center);
\draw [line width = 0.5mm] (i2.center) to [out=-160,in=-20] (i0.center);
\draw [line width = 0.5mm] (i2.center) to [out=80,in=-140] (i3.center);
\draw [line width = 0.5mm] (i3.center) to [out=-40,in=100] (i4.center);
\draw [line width = 0.5mm] (i4.center) to [out=-160,in=-20] (i2.center);
\end{tikzpicture}
+ 18 \, \EE \, 
\begin{tikzpicture}[baseline={([yshift=-.5ex]current bounding box.center)},scale=0.5]
\node (i0) at (0,0) {$\bullet$};
\node (i1) at (0.866,1.5) {$\bullet$};
\node (i2) at (1.732,0) {$\bullet$};
\node (i3) at (2.598,1.5) {$\bullet$};
\draw [line width = 0.5mm] (i0.center) to [out=80,in=-140] (i1.center);
\draw [line width = 0.5mm] (i1.center) to [out=-40,in=100] (i2.center);
\draw [line width = 0.5mm] (i2.center) to [out=-160,in=-20] (i0.center);
\draw [line width = 0.5mm] (i2.center) to [out=40,in=-100] (i3.center);
\draw [line width = 0.5mm] (i3.center) to [out=160,in=20] (i1.center);
\draw [line width = 0.5mm] (i1.center) to [out=-80,in=140] (i2.center);
\end{tikzpicture}
+ 6 \, \EE \, 
\begin{tikzpicture}[baseline={([yshift=-.5ex]current bounding box.center)},scale=0.5]
\node (i0) at (0,0) {$\bullet$};
\node (i1) at (0.866,1.5) {$\bullet$};
\node (i2) at (1.732,0) {$\bullet$};
\draw [line width = 0.5mm] (i0.center) to [out=80,in=-140] (i1.center);
\draw [line width = 0.5mm] (i0.center) to [out=40,in=-100] (i1.center);
\draw [line width = 0.5mm] (i1.center) to [out=-40,in=100] (i2.center);
\draw [line width = 0.5mm] (i1.center) to [out=-80,in=140] (i2.center);
\draw [line width = 0.5mm] (i2.center) to [out=-160,in=-20] (i0.center);
\draw [line width = 0.5mm] (i2.center) to [out=160,in=20] (i0.center);
\end{tikzpicture}
\right]
$
\caption{Second moment decomposition of the simple cycle of order $3$. We detail the combinatorial factors.}\label{fig:decomposition_second_moment}
\end{figure}
It is now possible to see why it implies our results $(i)$ and $(ii)$. Indeed, all the diagrams $\mathcal{C}_\alpha$ have 
an expectation that is $\mathcal{O}_N(1)$ by the first moment analysis we performed in Sec.~\ref{subsec:freecum_expectation} and Sec.~\ref{subsec:generic_diagrams_expectation}.
So very generically, for every kind of diagram we described we have:
\begin{align}
  \EE \, G^2 &= \left(\EE \, G\right)^2 + \smallO_N(1). 
\end{align}
Given our  previous computations of the first moments this implies results $(i)$ and $(ii)$.

\subsection{The higher-order moments and their influence on the diagrammatics in the symmetric model}\label{subsec:higher_order_symmetric}

All the results of Sec.~\ref{subsec:freecum_expectation}, Sec.~\ref{subsec:generic_diagrams_expectation} and Sec.~\ref{subsec:concentration_J}
that we derived for the diagrammatics of the Plefka expansion in this context
were valid for diagrams solely made out of the matrix elements $\{J_{i j}\}$, without any additional factors.
However in the Plefka expansions there
generically are possible factors that are the cumulants (or the moments) of the variables $x_i$ at $\beta=0$, see Sec.~\ref{subsec:plefka_sym_models}.
Recall that we denote $\kappa^{(p)}_{i}$ the cumulant of order $p$ of $x_i$ at $\beta=0$.
As an example, consider the diagram of Fig.~\ref{fig:ncactus_appendix}. 
Two possible contributions to the free entropy in our Plefka expansion at order $6$ would be:
\begin{align}
  \label{eq:ex_diag_1}
 &\frac{1}{N}\sum_{\substack{i_1, i_2,i_3 \\ \text{pairwise distincts}}}  J_{i_1 i_2} J_{i_2 i_3} J_{i_3 i_4} J_{i_4 i_1} J_{i_1 i_3}^2\, v_{i_1}^2 v_{i_2} v_{i_3}^2 v_{i_4}, \\
  \label{eq:ex_diag_2}
 &\frac{1}{N}\sum_{\substack{i_1, i_2,i_3 \\ \text{pairwise distincts}}}  J_{i_1 i_2} J_{i_2 i_3} J_{i_3 i_4} J_{i_4 i_1} J_{i_1 i_3}^2\, \kappa^{(4)}_{i_1} v_{i_2} \kappa^{(4)}_{i_3} v_{i_4}.
\end{align}
Note that both these contributions are represented by the diagram of Fig.~\ref{fig:ncactus_appendix}.
One can now clearly see that in order to apply the diagrammatic results
of Sec.~\ref{subsec:freecum_expectation}, Sec.~\ref{subsec:generic_diagrams_expectation} and Sec.~\ref{subsec:concentration_J}
to our Plefka expansion, and justify eq.~\eqref{eq:conjecture_phi_symmetric}, we 
need to make some additional assumptions that we detail here:
\begin{enumerate}[label=\textbf{A.\arabic*},ref=A.\arabic*]
  \item\label{assumption:non_eulerian} From the construction of the diagrams, odd cumulants of order greater or equal to $3$ 
  only appear in \emph{non-Eulerian} graphs. By the results of Sec.~\ref{subsec:generic_diagrams_expectation} and Sec.~\ref{subsec:concentration_J}
   we know that these diagrams, without the moments or cumulants as factors, are negligible. 
  We assume that the possible correlations of the higher order moments of $x_i$ with the matrix elements $\{F_{\mu i}\}$
  are not strong enough to yield thermodynamically relevant corrections to the free entropy.
  \item\label{assumption:irreducible} Eulerian strongly irreducible diagrams that are not simple cycles are negligible by our previous result. We assume 
  that the higher order (even) moments that appear as additional factors do not change their scaling, so that they remain negligible in the thermodynamic limit.
\end{enumerate}
For instance, \ref{assumption:irreducible} implies that the contributions of both eq.~\eqref{eq:ex_diag_1} and eq.~\eqref{eq:ex_diag_2} are 
negligible in the $N \to \infty$ limit, as the diagram of Fig.~\ref{fig:ncactus_appendix} is strongly irreducible but is not a simple cycle.
Concerning the simple cycles, we already know that they are not thermodynamically negligible. So we do not need to assume anything additional
regarding them. Note however that in order to ``resum'' the free entropy of the Plefka expansion, as we did in Sec.~\ref{subsec:plefka_sym_models}, we will need to assume that all the variance 
factors appearing in these simple cycles will be the same, that is $v_i = v$ (at the maximum of the free entropy).

\subsection{Extension to bipartite models}\label{subsec:diagrammatics_bipartite_models}

We detail here how we can treat the diagrams that arise in the Plefka expansion of bipartite models with pairwise interactions
(as the generalized linear models)
that we perform in Sec.~\ref{subsec:plefka_bipartite_models}.
The structure of this section is the following:
\begin{itemize}
  \item We show in Sec.~\ref{subsubsec:generalization_non_square} how we can generalize all the techniques and results 
  already seen in the rest of Sec.~\ref{sec:diagrammatics} to diagrams constructed from a random rectangular matrix $L$ drawn from the rotationally invariant
  ensemble given by Model~\ref{model:nsym_rot_inv}.
  \item In Sec.~\ref{subsubsec:higher_moments} we transpose the assumptions of Sec.~\ref{subsec:higher_order_symmetric}
  to this bipartite case, to deal with the higher-order moments of the fields that can arise in the high-temperature Plefka expansions.
\end{itemize}

\subsubsection{Generalization of the previous results to rectangular matrices}\label{subsubsec:generalization_non_square}

Consider a random matrix $F \in \bbR^{M \times N}$ drawn from a rotation invariant ensemble satisfying Model~\ref{model:nsym_rot_inv}. 
We are interested in the limit $M,N \to \infty$ with a finite ratio $M/N \to \alpha > 0$.
In the Plefka expansions performed for bipartite models in Sec.~\ref{subsubsec:plefka_bipartite_spherical} and Sec.~\ref{subsec:plefka_bipartite_models} 
they appear some quantities that we can represent as \emph{diagrams}.
In this subsection, we construct diagrams as explained Fig.~\ref{fig:nonsquare_diagram}.
For instance, the diagram depicted in this figure represents the quantity:
\begin{align}
 \frac{1}{N} \sum_{\substack{\mu_1, \mu_2, \mu_3 \\ \text{pairwise distincts}}} \sum_{\substack{i_1, i_2,i_3 \\ \text{pairwise distincts}}}  F_{\mu_1 i_1} F_{\mu_1 i_2} F_{\mu_2 i_2} F_{\mu_2 i_3} F_{\mu_3 i_3} F_{\mu_3 i_1} F_{\mu_1 i_3}^2.
\end{align}
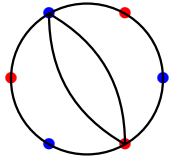
\begin{figure}[t]
  \centering
\captionsetup{justification=centering}
\begin{tikzpicture}[scale=1]
\node[color=red] (i1) at (2,0) {$\bullet$};
\node[color=blue] (i2) at (2.5,0.866) {$\bullet$};
\node[color=red] (i3) at (3.5,0.866) {$\bullet$};
\node[color=blue] (i4) at (4,0) {$\bullet$};
\node[color=blue] (i5) at (2.5,-0.866) {$\bullet$};
\node[color=red] (i6) at (3.5,-0.866) {$\bullet$};
\draw [line width = 0.3mm] (3,0) circle (1);
\draw [line width = 0.3mm] (i2.center) to [out=-90,in=150] (i6.center);
\draw [line width = 0.3mm] (i2.center) to [out=-30,in=90] (i6.center);
\end{tikzpicture}
  \caption{A diagram constructed from the non square matrix $F$. Each blue vertex is an index $\mu$, each red vertex an index $i$. Each edge is a factor $F_{\mu i}$, and we sum on each vertex the resulting quantity.
   Each connected component of the diagram carries a global factor $\frac{1}{N}$.
   Note that there can only be edges between red and blue vertices.}\label{fig:nonsquare_diagram}
\end{figure}
The analogous to Theorem~2 of \cite{guionnet2005fourier} for this setting can be stated. Let $\Sigma \in \bbR^{M \times N}$ be a matrix
such that the empirical spectral distribution of $D \equiv \Sigma^\intercal \Sigma$ converges (almost surely) as $N \to \infty$ to a probability measure $\rho_D$.
Denote $\mathcal{G}_{\alpha,\rho_D}$ the following function: 
\begin{align}
 \mathcal{G}_{\alpha,\rho_D}(x) &\equiv \frac{1}{2}\inf_{\gamma_1,\gamma_2}\left[\alpha \gamma_1 + \gamma_2 - (\alpha-1) \log \gamma_1 - \int \rho_D(\mathrm{d}\lambda) \log(\gamma_1 \gamma_2 - x^2 \lambda)\right] - \frac{1+\alpha}{2}.
\end{align}
Note that this is obviously an even function of $x$, and that $\mathcal{G}_{\alpha,\rho_D}(0) = 0$.
The function $\mathcal{G}_{\alpha,\rho_D}$ stands as an analog to the integrated $\mathcal{R}$-transform $G_{\rho_D}$
for this problem. Since one could expand the function $G_{\rho_D}$ using the free cumulants of $\rho_D$, we analogously expand formally $\mathcal{G}_{\alpha,\rho_D}(x)$
around $x=0$, and define the coefficients $\Gamma_p(\alpha,\rho_D)$ by:
\begin{align}\label{eq:def_Gammap}
 \mathcal{G}_{\alpha,\rho_D}(x) &\equiv \sum_{p=1}^\infty \frac{1}{2p} \Gamma_p(\alpha,\rho_D) x^{2p}.
\end{align}
Recall that for any function $f(x)$, and any symmetric matrix $J = O D O^\intercal \in {\cal S}_N$, one can define $f(J) \equiv O f(D) O^\intercal$, with
$f(D) = \mathrm{Diag}\,(\{f(d_i)\}_{1 \leq i \leq N})$. If one can expand $f(x) = \sum_{k\geq 0} c_k x^k$, this definition is coherent with $f(J) = \sum_{k \geq 0} c_k J^k$.
Our generalization of Theorem~2 of \cite{guionnet2005fourier} is the following: consider a 
rectangular matrix $\Lambda \in \bbR^{M \times N}$ of finite rank $p$. In other terms, one can write its SVD decomposition as:
\begin{align*}
 \Lambda = U_0  \Delta V_0^\intercal,
\qquad
\Delta=
\begin{pmatrix}
\Lambda_p & 0 \\
0 & 0
 \end{pmatrix} \in \bbR^{M \times N},
 \end{align*}
with $\Lambda_p \in \bbR^{p \times p}$ a square diagonal matrix, and $U_0,V_0$ orthogonal matrices. We can now state:
\begin{align}\label{eq:guionnet_bipartite}
  \lim_{N \to \infty} \frac{1}{N} \log \int_{\mathcal{O}(M)} \mathcal{D}U \int_{\mathcal{O}(N)} \mathcal{D}V e^{\sqrt{\alpha}N \mathrm{Tr} \, \left[\Lambda^\intercal U \Sigma V^\intercal\right]} &= \mathrm{Tr} \, \left[\mathcal{G}_{\alpha,\rho_D}(\Lambda_p)\right],
\end{align}
Note first that the right hand side of this equation can also be written as:
\begin{align}
  \mathrm{Tr} \, \left[\mathcal{G}_{\alpha,\rho_D}(\Lambda_p)\right] &=  \mathrm{Tr} \, \left[\mathcal{G}_{\alpha,\rho_D}(\sqrt{\Lambda^\intercal \Lambda})\right] \overset{(a)}{=}  \mathrm{Tr} \, \left[\mathcal{G}_{\alpha,\rho_D}(\sqrt{\Lambda \Lambda^\intercal})\right],
\end{align}
since $\Lambda$ is of finite rank $p$. Equality $(a)$ is true since $\mathcal{G}_{\alpha,\rho_D}(0) =  0$, and the spectrum of $\Lambda^\intercal \Lambda$ and $\Lambda \Lambda^\intercal$ 
only differ by eigenvalues which are all equal to $0$.
Note also that we already derived and used this relation, for $p=1$, when computing the free entropy of the model of Sec.~\ref{subsec:bipartite_spherical_model}, as
stated in eq.~\eqref{eq:bipartite_all_temp}.
Equipped with the definitions of $\mathcal{G}_{\alpha,\rho_D}$, $\Gamma_p(\alpha,\rho_D)$, and eq.~\eqref{eq:guionnet_bipartite}, we can state the counterpart of all our previous results in this rectangular
setting:
\begin{enumerate}[label=\textbf{R.\arabic*},ref=R.\arabic*]
  \item \label{result:simple_cycle}Consider a simple cycle of size $2p$. Then it converges (in $L^2$) to $\Gamma_{p}(\alpha,\rho_D)$ as $N \to \infty$. More precisely we have:
  \begin{align}\label{eq:convergence_L2_cycles_nonsquare}
  \lim_{N \to \infty} \EE \left|\frac{1}{N} \sum_{\substack{\mu_1,\cdots,\mu_p \\ \text{pairwise distincts}}} \sum_{\substack{i_1,\cdots,i_p \\ \text{pairwise distincts}}} F_{\mu_1 i_1} F_{\mu_1 i_2} F_{\mu_2 i_2} \cdots F_{\mu_p i_p} F_{\mu_p i_1} - \Gamma_p(\alpha,\rho_D) \right|^2 &= 0.  
  \end{align}
  \item\label{result:non_eulerian} Any diagram $G$ that is not \emph{Eulerian} will have a vanishing first and second moment as $N \to \infty$:
  \begin{align}\label{eq:vanishing_G2}
    \lim_{N \to \infty} \EE \, G^2 &= 0.
  \end{align}
  \item\label{result:irreducible} Any diagram $G$ that is \emph{strongly irreducible} (that is it can not be disconnected by removing a single vertex) but not a simple cycle, as in Fig.~\ref{fig:nonsquare_diagram}, 
  will also have a vanishing first and second moment.
  \item\label{result:cactus} If $G$ is a \emph{cactus} (a tree made of simple cycles joining at vertices \cite{parisi1995mean}) made of $r$ simple cycles of size $(2p_1,\cdots,2p_r)$, we have:
  \begin{align}
  \lim_{N \to \infty} \EE\, \left|G - \prod_{l=1}^r \Gamma_{p_l}(\alpha,\rho_D)\right|^2 &= 0.     
  \end{align}
\end{enumerate}
Since every argument to show points \ref{result:simple_cycle} to \ref{result:cactus} is straightforwardly given by slightly modifying what we already did in Sec.~\ref{sec:diagrammatics}, 
the rest of Sec.~\ref{subsubsec:generalization_non_square} will be devoted to show point \ref{result:simple_cycle}, and we leave the remaining points for the reader.

\paragraph{Justifying \ref{result:simple_cycle}}
In order to show eq.~\eqref{eq:convergence_L2_cycles_nonsquare}, we proceed as in Sec.~\ref{subsec:freecum_expectation} and 
begin by showing:
\begin{align}\label{eq:simple_cycle_nonsquare_first_moment}
\lim_{N \to \infty} \EE \left[\frac{1}{N} \sum_{\substack{\mu_1,\cdots,\mu_p \\ \text{pairwise distincts}}} \sum_{\substack{i_1,\cdots,i_p \\ \text{pairwise distincts}}} F_{\mu_1 i_1} F_{\mu_1 i_2} F_{\mu_2 i_2} \cdots F_{\mu_p i_p} F_{\mu_p i_1}\right] &= \Gamma_p(\alpha,\rho_D) + \smallO_N(1).
\end{align}
By rotation invariance of the indices we can replace the left-hand side of eq.~\eqref{eq:simple_cycle_nonsquare_first_moment} by a term without summation on the indices,
and as in Sec.~\ref{subsec:freecum_expectation} we obtain at leading order in $N$:
\begin{align*}
 \alpha^p N^{2p-1} \EE \left[F_{1 1} F_{1 2} F_{2 2} \cdots F_{p p} F_{p 1} \right] &= \frac{1}{N} \frac{\partial^{2p}}{\partial b_1 \cdots \partial b_p \partial c_1 \cdots \partial c_{p}} \left[\int \mathcal{D}U \mathcal{D}V e^{\sqrt{\alpha}N \mathrm{Tr}\,\left[M(\bb,\bc)^\intercal U \Sigma V^\intercal \right]}\right]_{\bb, \bc = 0},
\end{align*}
with $M(\bb,\bc)$ a block matrix of rank $p$ defined as:
\begin{align*}
 M(\bb,\bc) &= \begin{pmatrix}
   M_1(\bb,\bc) & 0 \\
   0 & 0  
 \end{pmatrix}, \\
 M_1(\bb,\bc) &\equiv \begin{pmatrix}
  b_1 & c_{1} & 0 & \cdots & 0 & 0 \\
  0 & b_2 & c_{2} & \cdots & 0 & 0 \\
  0 & 0 & b_3 & \cdots & 0 & 0 \\
  \vdots & \vdots & \vdots & \ddots & \vdots & \vdots \\
  0 & 0 & 0 & \cdots & b_{p-1} & c_{p-1} \\
  c_{p} & 0 & 0 & \cdots & 0 & b_p  
 \end{pmatrix}.
\end{align*}
Using eq.~\eqref{eq:guionnet_bipartite}, we obtain:
\begin{align}
  \label{eq:expansion_nonsquare}
\alpha^p N^{2p-1} &\EE \left[F_{1 1} F_{1 2} F_{2 2} \cdots F_{p p} F_{p 1} \right] = \\
 & \frac{1}{N} \frac{\partial^{2p}}{\partial b_1 \cdots \partial b_{p}\partial c_{1} \cdots \partial c_{p}} \left[\exp\left\{N \sum_{n=1}^\infty \frac{\Gamma_n(\alpha,\rho_D)}{2n} \mathrm{Tr}\,\left[(M(\bb,\bc)^\intercal M(\bb,\bc))^n\right]\right\} \right]_{\bb,\bc = 0} + \smallO_N(1), \nonumber
\end{align}
We define the elementary matrices $(T_{ab})_{ll'} = \delta_{al} \delta_{bl'}$ and the symmetric elementary matrices $E_{ab} = T_{ab} + T_{ba}$.
One easily derives that $\frac{\partial^2}{\partial b_1 \partial c_{1}} M(\bb,\bc)^\intercal M(\bb,\bc) = E_{12}$. In a very similar way to what was done 
in Sec.~\ref{subsec:freecum_expectation}, the dominant terms in eq.~\eqref{eq:expansion_nonsquare} will be given 
by the maximum number of differentiations of the exponential term. However, one can see that the exponential can only be differentiated once: since $M(0,0) = 0$, one would need
to create cycles with the matrices $E_{ab}$, and such a cycle can only appear if one derives a single time the exponential term. As in Sec.~\ref{subsec:freecum_expectation}, there are two cycles
that are created  by the successive derivatives: $E_{12} E_{23} \cdots E_{p1}$ and $E_{21} E_{1p} \cdots E_{32}$. These two cycles yield the dominant contribution:
\begin{align*}
 \alpha^p N^{2p-1} &\EE \left[F_{1 1} F_{1 2} F_{2 2} \cdots F_{p p} F_{p 1} \right] \\
 &= \sum_{n=p}^\infty \frac{\Gamma_n(\alpha,\rho_D)}{2} \mathrm{Tr}\, \left[\left(E_{12} E_{23} \cdots E_{p1} + E_{21} E_{1p} \cdots E_{32}\right) (M(0,0)^\intercal M(0,0))^{n-p}\right] + \smallO_N(1) , \\
 &= \frac{1}{2} \Gamma_p(\alpha,\rho_D)  \mathrm{Tr}\, \left(E_{12} E_{23} \cdots E_{p1} + E_{21} E_{1p} \cdots E_{32}\right) + \smallO_N(1) ,\\
 &= \Gamma_p(\alpha,\rho_D) + \smallO_N(1).
\end{align*}
This shows eq.~\eqref{eq:simple_cycle_nonsquare_first_moment}. The exact same arguments as the ones used in Sec.~\ref{subsec:concentration_J} 
show that we have $L^2$ concentration, which means:
\begin{align*}
\lim_{N \to \infty} \EE \left|\frac{1}{N} \sum_{\substack{\mu_1,\cdots,\mu_p \\ \text{pairwise distincts}}} \sum_{\substack{i_1,\cdots,i_p \\ \text{pairwise distincts}}} F_{\mu_1 i_1} F_{\mu_1 i_2} F_{\mu_2 i_2} \cdots F_{\mu_p i_p} F_{\mu_p i_1} - \Gamma_p(\alpha,\rho_D) \right|^2 &= 0,
\end{align*}
which is the point \ref{result:simple_cycle} we wanted to show.

\subsubsection{The higher order moments and their influence on the diagrammatics}\label{subsubsec:higher_moments}

In Sec.~\ref{subsec:plefka_bipartite_models}, we deal with diagrams which have additional factors coming 
from the higher order moments of the fields $x_i$ and $h_\mu$ at $\beta=0$, while all the results \ref{result:simple_cycle} to \ref{result:cactus} that we derived for the diagrammatics of the Plefka expansion in this context
were made solely out of the matrix elements $\{F_{\mu i}\}$, without any additional factors. We adopt the notation of Sec.~\ref{subsec:higher_order_symmetric}
for the higher order cumulants.
Exactly as in Sec.~\ref{subsec:higher_order_symmetric}, when considering the diagram of Fig.~\ref{fig:nonsquare_diagram},
two possible contributions to the free entropy at order $8$ would be:
\begin{align}
  \label{eq:ex_diag_1_bip}
 &\frac{1}{N} \sum_{\substack{\mu_1, \mu_2, \mu_3 \\ \text{pairwise distincts}}} \sum_{\substack{i_1, i_2,i_3 \\ \text{pairwise distincts}}}  F_{\mu_1 i_1} F_{\mu_1 i_2} F_{\mu_2 i_2} F_{\mu_2 i_3} F_{\mu_3 i_3} F_{\mu_3 i_1} F_{\mu_1 i_3}^2 \, (v^h_{\mu_1})^2 v^h_{\mu_2} v^h_{\mu_3} v^x_{i_1} v^x_{i_2} (v^x_{i_3})^2, \\
  \label{eq:ex_diag_2_bip}
 &\frac{1}{N} \sum_{\substack{\mu_1, \mu_2, \mu_3 \\ \text{pairwise distincts}}} \sum_{\substack{i_1, i_2,i_3 \\ \text{pairwise distincts}}}  F_{\mu_1 i_1} F_{\mu_1 i_2} F_{\mu_2 i_2} F_{\mu_2 i_3} F_{\mu_3 i_3} F_{\mu_3 i_1} F_{\mu_1 i_3}^2\,  \kappa^{(4,h)}_{\mu_1} v^h_{\mu_2} v^h_{\mu_3} v^x_{i_1} v^x_{i_2} \kappa^{(4,x)}_{i_3}.
\end{align}
The assumptions we need to make in order to deal with these diagrams are very similar to
\ref{assumption:non_eulerian} and \ref{assumption:irreducible}, and we state them here for completeness:
\begin{enumerate}[label=\textbf{B.\arabic*},ref=B.\arabic*]
  \item\label{assumption:non_eulerian_bip} From the construction of the diagrams, odd moments of order greater or equal to $3$, like $\kappa^{(3,x)}$, 
  only appear in \emph{non-Eulerian} graphs. By \ref{result:non_eulerian} we know that these diagrams (without the moments as factors) 
  are negligible. 
  We assume that the possible correlations of the higher order moments with the matrix elements $\{F_{\mu i}\}$
  are not strong enough to yield thermodynamically relevant corrections to the free entropy.
  \item\label{assumption:irreducible_bip} Eulerian strongly irreducible diagrams that are not simple cycles are negligible by \ref{result:irreducible}. We assume 
  that the higher order (even) moments that appear as additional factors do not change their scaling, so that they remain negligible in the thermodynamic limit.
\end{enumerate}

\subsection{A note on i.i.d.\ matrices}\label{subsec:iid_matrices}

We make here a side comment on i.i.d.\ rectangular matrices. We consider a random matrix $F \in \bbR^{M \times N}$ whose elements $\{F_{\mu i}\}$
are taken i.i.d., such that $\sqrt{N} F_{\mu i}$ is drawn from a given probability measure $\rho$. We assume that $\rho$ has zero mean and finite moments of all orders.
These matrices appear in our study of the GAMP algorithm in Sec.~\ref{subsec:gamp}.
Except if $\rho$ is a Gaussian probability measure, the matrix $F$ is not rotationally invariant, in the sense that it does not satisfy Model~\ref{model:nsym_rot_inv}. 
However, one can still derive strong results on the diagrammatics of $F$.
We still assume \ref{assumption:non_eulerian_bip} and \ref{assumption:irreducible_bip}, that is we assume that the additional factors in the 
diagrams do not change the scaling of a negligible diagram enough to make it thermodynamically relevant.
It is then easy to see that because the $\{F_{\mu i}\}$ are uncorrelated, \emph{all diagrams with order $p \geq 3$ are negligible in the $N \to \infty$ limit}.
The only diagram that remains in the $N \to \infty$ limit is:
\begin{align}
\begin{tikzpicture}[baseline={([yshift=-.5ex]current bounding box.center)}, scale=1.]
\node (i0) at (0,0) {$\bullet$};
\node (i1) at (1,0) {$\bullet$};
\draw [line width = 0.5mm] (i0.center) [in=135,out=45] to (i1.center);
\draw [line width = 0.5mm] (i0.center) [in=-135,out=-45] to (i1.center);
\end{tikzpicture}
&= \frac{1}{N} \sum_{\mu,i} F_{\mu i}^2 v^h_\mu v^x_i.
\end{align}
In particular, we can only retain this diagram and apply the results of our Plefka expansions in this case 
as well, despite the fact that $F$ is not rotationally invariant.

%% file: appendix_operator_U.tex
\section{The Georges-Yedidia formalism} \label{sec:appendix_operator_U}

In this section we recall the formalism of \cite{georges1991expand} which allows to systematically
expand the free entropy around $\beta=0$, at fixed values of the first and second moments of the variables.
We consider here a generic Hamiltonian $H_J(\bx)$, with variables $(x_1,\cdots,x_N)$. 
We fix the first and second moments $\braket{x_i}_\beta = m_i$ and $\braket{(x_i - m_i)^2}_\beta = v_i$, using Lagrange parameters 
respectivelly denoted $\lambda_i(\beta)$ and $\gamma_i(\beta)$. Recall that $\braket{\cdot}_\beta$ stands for the expectation over the Gibbs 
measure at inverse temperature $\beta$, constrained by the Lagrange multipliers $\lambda_i$ and $\gamma_i$.
From now on, we will drop the $\beta$ supbscript.

We introduce the operator $U$ from Appendix~A of \cite{georges1991expand}:
\begin{equation}
U(\beta,J) = H_{J} - \braket{H_{J}} + \sum_{i=1}^N \partial_\beta\lambda_i(\beta) (x_i - m_i) + \frac{1}{2} \sum_{i=1}^N\partial_\beta \gamma_i(\beta) \left[x_i^2 - v_i - m_i^2 \right].
\end{equation}
Then the derivative of the thermal average of any observable $O$ is given by
\begin{equation}
\frac{\partial \langle O\rangle}{\partial \beta} = \left\langle \frac{\partial O}{\partial \beta} \right\rangle - \langle O U \rangle.
\end{equation}
As the Lagrange multipliers $\lambda_i$ and $\gamma_i$ have been introduced to
fix the average of $x_i$ and its variance one has the following easy identity, valid at any $\beta$:
\begin{equation}
 \langle U \rangle = 0.
\end{equation}
Moreover, given that the magnetizations $\{m_i\}$ and the variances $\{v_i\}$ do not depend on $\beta$ one has:
\begin{equation}
0 = \frac{\partial \langle x_i\rangle}{\partial \beta} =  - \langle x_i U \rangle = - \langle (x_i-m_i) U \rangle.
\end{equation}
\begin{equation}
0 = \frac{\partial \langle x_i^2\rangle}{\partial \beta} =  - \langle x_i^2 U \rangle = -  \langle (x_i^2-v_i-m_i^2) U \rangle.
\end{equation}
Considering the previous results one can compute the derivative of $U$:
\begin{align}
\frac{\partial U}{\partial \beta} &=  \langle H_J U \rangle + \sum_{i=1}^N \partial_\beta^2\lambda_i(\beta) (x_i - m_i) + \frac{1}{2} \sum_{i=1}^N\partial_\beta^2 \gamma_i(\beta) \left[x_i^2 - v_i - m_i^2 \right], \nonumber  \\ 
&=  \langle U^2 \rangle + \sum_{i=1}^N \partial_\beta^2\lambda_i(\beta) (x_i - m_i) + \frac{1}{2} \sum_{i=1}^N\partial_\beta^2 \gamma_i(\beta) \left[x_i^2 - v_i - m_i^2 \right].
\end{align}
Equipped with these relations one can compute the derivatives of the free entropy up to fourth order.
Recall that $\Phi_J$ is the \emph{intensive} free entropy of the system. We obtain its derivatives:
\begin{align}
\frac{\partial N \Phi_J}{\partial \beta} &=  - \braket{H_{J}} + \sum_{i=1}^N \partial_\beta\lambda_i(\beta) \langle(x_i - m_i)\rangle + \frac{1}{2} \sum_{i=1}^N\partial_\beta \gamma_i(\beta) \langle\left[x_i^2 - v_i - m_i^2 \right] \rangle = - \langle H_J \rangle, \\
\label{Eq_U2}
\frac{\partial^2 N \Phi_J}{\partial \beta^2} &=  \braket{H_{J} U} =  \braket{U^2}, \\
\frac{\partial^3 N \Phi_J}{\partial \beta^3} &= -  \braket{U^3} + 2 \braket{U \frac{\partial U}{\partial\beta}} = -  \braket{U^3}, \\
\label{Eq_U4}
\frac{\partial^4 N \Phi_J}{\partial \beta^4} &=  \braket{U^4} - 3  \braket{U^2 \frac{\partial U}{\partial\beta}} =  \braket{U^4} - 3  \braket{U^2}^2
 - 3 \sum_{i=1}^N \partial_\beta^2\lambda_i(\beta) \braket{U^2 (x_i - m_i)}  \nonumber \\
 &- \frac{3}{2} \sum_{i=1}^N\partial_\beta^2 \gamma_i(\beta) 
 \braket{U^2 \left[x_i^2 - v_i - m_i^2 \right] }.
\end{align}
These relations are valid at any inverse temperature $\beta$ ! In the main sections we derive the explicit
expression of the operator $U$ for our particular choice of Hamiltonian, and we will use these relations (and show how to conjecture their higher order counterparts)
 to compute the expansion of the free entropy around $\beta=0$.

%% file: appendix_order4_spherical.tex
\section{Order \texorpdfstring{$4$}{4} of the Plefka expansion for Sec.~\ref{subsec:sym_spherical_model}.} \label{sec:appendix_order4_spherical}

We start from eq.~\eqref{Eq_U4} in Appendix~\ref{sec:appendix_operator_U}, that we consider at $\beta=0$ :
\begin{align}\label{eq:Phi4_spherical}
 N \frac{\partial^4 \Phi_J(\beta)}{\partial \beta^4} &= \braket{U^4}_0 - 3 \braket{U^2}_0^2 - 3 \sum_{i=1}^N \partial^2_\beta \lambda_i \braket{U^2(x_i-m_i)}_0 - \frac{3}{2} \sum_{i=1}^N \partial^2_\beta \gamma_i \braket{U^2(x_i^2-m_i^2-v_i)}_0.
\end{align}
For simplicity we will denote $\tilde{x}_i \equiv (x_i-m_i)$, so that at $\beta=0$ the $\{\tilde{x_i}\}$ variables are Gaussian variables with mean $\braket{\tilde{x}_i} = 0$ and covariance $\braket{\tilde{x_i} \tilde{x_j}} = \delta_{ij} v_i$. 
In particular eq.~\eqref{eq:U_spherical} becomes:
\begin{align}
   U = - \frac{1}{2} \sum_{i \neq j} J_{ij} \tilde{x}_i \tilde{x}_j.
\end{align}
From the calculation at order $2$ we obtain the following relation that we can represent diagrammatically:
\begin{align}\label{eq:order4_spherical_1}
   -3\braket{U^2}^2_0 &= - \frac{3}{4} \left[\sum_{i \neq j} J_{ij}^2 v_i v_j\right]^2 = -\frac{3N}{4} 
\left[\begin{tikzpicture}[baseline={([yshift=-.5ex]current bounding box.center)}, scale=1.]
\node (i0) at (0,0) {$\bullet$};
\node (i1) at (1,0) {$\bullet$};
\draw [line width = 0.5mm] (i0.center) [in=135,out=45] to (i1.center);
\draw [line width = 0.5mm] (i0.center) [in=-135,out=-45] to (i1.center);
\end{tikzpicture}
\right]^2.
\end{align}
We now turn to the next term:
\begin{align}
   &- \frac{3}{2} \sum_{i=1}^N \partial^2_\beta \gamma_i \braket{U^2(x_i^2-m_i^2-v_i)}_0 - 3 \sum_{i=1}^N \partial^2_\beta \lambda_i \braket{U^2(x_i-m_i)}_0 \nonumber \\
   &=  - \frac{3}{2} \sum_{i=1}^N \partial^2_\beta \gamma_i \braket{U^2(\tilde{x}_i^2-v_i)}_0 - 3 \sum_{i=1}^N \braket{U^2 \tilde{x}_i \left(\partial^2_\beta \lambda_i + m_i \partial^2_\beta \gamma_i \right)}_0 , \nonumber \\
   &\overset{(a)}{=} - \frac{3}{2} \sum_{i=1}^N \partial^2_\beta \gamma_i \braket{U^2(\tilde{x}_i^2-v_i)}_0 - 3 N\sum_{i=1}^N \braket{U^2 \tilde{x}_i  \frac{\partial}{\partial m_i} \left(\partial^2_\beta \Phi_J \right) }_0 , \nonumber \\ 
   &\overset{(b)}{=} - \frac{3}{2} \sum_{i=1}^N \partial^2_\beta \gamma_i \braket{U^2(\tilde{x}_i^2-v_i)}_0 + \smallO_N(1).
\end{align}
In $(a)$ we used the Maxwell equation eq.~\eqref{eq:maxwell_spherical_2}, while in $(b)$ we made use of the fact that the order $2$ of the free entropy does 
not depend on the $m_i$ variables.
We obtain:
\begin{align*}
  -\frac{3}{2} \sum_{i=1}^N \partial^2_\beta \gamma_i \braket{U^2(\tilde{x}_i^2-v_i)}_0 &=  \frac{3}{2} \left[\sum_{i \neq j} J_{ij}^2 v_i v_j\right]^2 - 3 \sum_{i \neq j} J_{ij}^2 v_j \braket{U^2 \tilde{x}_i^2}_0, 
\end{align*}
in which we used the Maxwell relation eq.~\eqref{eq:maxwell_spherical_1} to compute $\partial^2_\beta \gamma_i$. 
To compute $\braket{U^2 \tilde{x}_i^2}_0$, we expand:
\begin{align*}
   \braket{U^2 \tilde{x}_i^2}_0 &= \frac{1}{4} \sum_{i_1 \neq j_1} \sum_{i_2 \neq j_2} J_{i_1 j_1} J_{i_2 j_2} \braket{\tilde{x}_i^2\tilde{x}_{i_1}\tilde{x}_{j_1}\tilde{x}_{i_2}\tilde{x}_{j_2}}_0.
\end{align*}
We can then use Wick's theorem to simplify the average. There are two types of contractions (or pairings):
\begin{itemize}
   \item Contractions that do not mix indices $i_1,j_1,i_2,j_2$ with $i$. There are $2$ such possible pairings and in $\frac{\partial^4 \Phi_J}{\partial \beta^4}$ they give rise to the diagram $
\left[N\begin{tikzpicture}[baseline={([yshift=-.5ex]current bounding box.center)}, scale=1.]
\node (i0) at (0,0) {$\bullet$};
\node (i1) at (1,0) {$\bullet$};
\draw [line width = 0.5mm] (i0.center) [in=135,out=45] to (i1.center);
\draw [line width = 0.5mm] (i0.center) [in=-135,out=-45] to (i1.center);
\end{tikzpicture}
\right]^2$.
   \item Contractions that mix these indices with $i$. There are all equivalent and there are $8$ of them, which gives rise to the diagram:
$N\begin{tikzpicture}[baseline={([yshift=-.5ex]current bounding box.center)}, scale=1.]
\node (i0) at (0,0) {$\bullet$};
\node (i1) at (1,0) {$\bullet$};
\node (i2) at (-1,0) {$\bullet$};
\draw [line width = 0.5mm] (i0.center) [in=135,out=45] to (i1.center);
\draw [line width = 0.5mm] (i0.center) [in=-135,out=-45] to (i1.center);
\draw [line width = 0.5mm] (i2.center) [in=135,out=45] to (i0.center);
\draw [line width = 0.5mm] (i2.center) [in=-135,out=-45] to (i0.center);
\end{tikzpicture}$.
\end{itemize}
In the end, we reach:
\begin{align*}
   \braket{U^2 \tilde{x}_i^2}_0 &= \frac{v_i}{2} \sum_{k\neq l} J_{k l}^2 v_k v_l + 2 \sum_{k (\neq i)} J_{ik}^2 v_i^2 v_k .
\end{align*}
We can finally compute the term we were seeking:
\begin{align}\label{eq:order4_spherical_2}
 - \frac{3}{2} \sum_{i=1}^N \partial^2_\beta \gamma_i \braket{U^2(\tilde{x}_i^2-v_i)}_0 &= - 6 N\hspace{-0.5cm} \sum_{\substack{i,j,k \\ \text{pairwise distincts}}} J_{ij}^2 J_{ik}^2 v_i^2 v_j v_k = -6 
\begin{tikzpicture}[baseline={([yshift=-.5ex]current bounding box.center)}, scale=1.]
\node (i0) at (0,0) {$\bullet$};
\node (i1) at (1,0) {$\bullet$};
\node (i2) at (-1,0) {$\bullet$};
\draw [line width = 0.5mm] (i0.center) [in=135,out=45] to (i1.center);
\draw [line width = 0.5mm] (i0.center) [in=-135,out=-45] to (i1.center);
\draw [line width = 0.5mm] (i2.center) [in=135,out=45] to (i0.center);
\draw [line width = 0.5mm] (i2.center) [in=-135,out=-45] to (i0.center);
\end{tikzpicture}.
\end{align}
Note that in this last equation we could add the hypothesis that $j \neq k$. Indeed the term $j=k$ would give rise to the diagram $N
\begin{tikzpicture}[baseline={([yshift=-.5ex]current bounding box.center)}, scale=1.]
\node (i0) at (0,0) {$\bullet$};
\node (i1) at (1,0) {$\bullet$};
\draw [line width = 0.5mm] (i0.center) [in=125,out=55] to (i1.center);
\draw [line width = 0.5mm] (i0.center) [in=165,out=15] to (i1.center);
\draw [line width = 0.5mm] (i0.center) [in=-125,out=-55] to (i1.center);
\draw [line width = 0.5mm] (i0.center) [in=-165,out=-15] to (i1.center);
\end{tikzpicture}
$, which is negligible since for every $i \neq j$
one has $J_{ij} = \mathcal{O}(\frac{1}{\sqrt{N}})$ as a consequence of rotational invariance (Model~\ref{model:sym_rot_inv}).
 We finally turn to the computation of $\braket{U^4}_0$:
\begin{align*}
   \braket{U^4}_0 &= \frac{1}{16} \prod_{\alpha=0}^3 \left[\sum_{i_\alpha \neq j_{\alpha}} J_{i_\alpha j_\alpha}\right] \, \Big\langle\prod_{\alpha=0}^3 \tilde{x}_{i_\alpha} \tilde{x}_{j_\alpha} \Big\rangle_0.
\end{align*}
The possible contractions arising from Wick's theorem yield several contributions, that we can represent by diagrams. 
Note that these diagrams are very different from the diagrams that we described for instance in Fig.~\ref{fig:diagrams}, and are merely 
a way to visualize the contractions in Wick's theorem.
The first column contains the $i_\alpha$ indices and the second contains the $j_\alpha$.
Note that we always have $i_\alpha \neq j_\alpha$. The two different types of contractions are represented as Fig.~\ref{fig:first_diag} and Fig.~\ref{fig:second_diag}. They are $12$ possible contractions of the type of Fig.~\ref{fig:first_diag}
and $48$ of Fig.~\ref{fig:second_diag}. We also take into account that in the pairings of Fig.~\ref{fig:second_diag} indices are not all necessarily pairwise distinct. 
Discarding terms that are $\smallO_N(N)$, we finally reach:
   \begin{figure}[t]
      \centering
      \begin{subfigure}[b]{0.22\textwidth}
         \centering
\begin{tikzpicture}
\draw (0,0) node {$\bullet$} ;
\draw (0.5,0) node {$\bullet$} ;
\draw (0,-0.5) node {$\bullet$} ;
\draw (0.5,-0.5) node {$\bullet$} ;
\draw (0,-1) node {$\bullet$} ;
\draw (0.5,-1) node {$\bullet$} ;
\draw (0,-1.5) node {$\bullet$} ;
\draw (0.5,-1.5) node {$\bullet$} ;
\draw [line width = 0.5mm](0,0) -- (0,-0.5);
\draw [line width = 0.5mm](0.5,0) -- (0.5,-0.5);
\draw [line width = 0.5mm](0,-1) -- (0,-1.5);
\draw [line width = 0.5mm](0.5,-1) -- (0.5,-1.5);
\end{tikzpicture}
         \caption{}\label{fig:first_diag}
      \end{subfigure}
      \begin{subfigure}[b]{0.22\textwidth}
         \centering
\begin{tikzpicture}
\draw (0,0) node {$\bullet$} ;
\draw (0.5,0) node {$\bullet$} ;
\draw (0,-0.5) node {$\bullet$} ;
\draw (0.5,-0.5) node {$\bullet$} ;
\draw (0,-1) node {$\bullet$} ;
\draw (0.5,-1) node {$\bullet$} ;
\draw (0,-1.5) node {$\bullet$} ;
\draw (0.5,-1.5) node {$\bullet$} ;
\draw [line width = 0.5mm](0,0) -- (0,-0.5);
\draw [line width = 0.5mm](0.5,0) -- (0,-1);
\draw [line width = 0.5mm](0.5,-0.5) -- (0,-1.5);
\draw [line width = 0.5mm](0.5,-1) -- (0.5,-1.5);
\end{tikzpicture}
         \caption{}\label{fig:second_diag}
      \end{subfigure}
\caption{Different types of diagrams of indices appearing in $\braket{U^4}_0$}
   \end{figure}
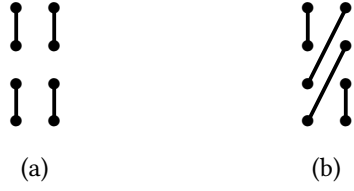

   \begin{align}
     \braket{U^4}_0 &= \frac{3}{4} \left[\sum_{i \neq j}J_{ij}^2 v_i v_j\right]^2 + 6 \hspace{-0.5cm}\sum_{\substack{i,j,k \\ \text{pairwise distincts }}} J_{ij}^2 J_{ik}^2 v_i^2 v_j v_k + 3 \hspace{-0.5cm}\sum_{\substack{i_0,i_1,i_2,i_3 \\ \text{pairwise distincts }}} J_{i_0 i_1} J_{i_1 i_2} J_{i_2 i_3} J_{i_3 i_0} v_{i_0} v_{i_1}v_{i_2}v_{i_3}, \nonumber\\
                    \label{eq:order4_spherical_3}
                    &= \frac{3}{4} 
\left[N\begin{tikzpicture}[baseline={([yshift=-.5ex]current bounding box.center)}, scale=1.]
\node (i0) at (0,0) {$\bullet$};
\node (i1) at (1,0) {$\bullet$};
\draw [line width = 0.5mm] (i0.center) [in=135,out=45] to (i1.center);
\draw [line width = 0.5mm] (i0.center) [in=-135,out=-45] to (i1.center);
\end{tikzpicture}
\right]^2 + 6 N
\begin{tikzpicture}[baseline={([yshift=-.5ex]current bounding box.center)}, scale=1.]
\node (i0) at (0,0) {$\bullet$};
\node (i1) at (1,0) {$\bullet$};
\node (i2) at (-1,0) {$\bullet$};
\draw [line width = 0.5mm] (i0.center) [in=135,out=45] to (i1.center);
\draw [line width = 0.5mm] (i0.center) [in=-135,out=-45] to (i1.center);
\draw [line width = 0.5mm] (i2.center) [in=135,out=45] to (i0.center);
\draw [line width = 0.5mm] (i2.center) [in=-135,out=-45] to (i0.center);
\end{tikzpicture}.
  + 3 N
\begin{tikzpicture}[baseline={([yshift=-.5ex]current bounding box.center)}, scale=0.5]
\node (i0) at (1,0) {$\bullet$};
\node (i1) at (-1,0) {$\bullet$};
\node (i2) at (0,-1) {$\bullet$};
\node (i3) at (0,1) {$\bullet$};
\draw [line width = 0.5mm] (0,0) circle (1cm);
\end{tikzpicture}.
   \end{align}
   Finally, combining eq.~\eqref{eq:order4_spherical_1}, eq.~\eqref{eq:order4_spherical_2}, and eq.~\eqref{eq:order4_spherical_3} to plug them into eq.~\eqref{eq:Phi4_spherical}, we reach:
   \begin{align}
     \frac{1}{4!} \frac{\partial^4 \Phi_J}{\partial \beta^4} &= \frac{1}{8 N} \hspace{-0.5cm}\sum_{\substack{i_0,i_1,i_2,i_3 \\ \text{pairwise distincts }}} J_{i_0 i_1} J_{i_1 i_2} J_{i_2 i_3} J_{i_3 i_0} v_{i_0} v_{i_1}v_{i_2}v_{i_3} + \smallO_N(1), \\
     &= 
\frac{1}{8}\begin{tikzpicture}[baseline={([yshift=-.5ex]current bounding box.center)}, scale=0.5]
\node (i0) at (1,0) {$\bullet$};
\node (i1) at (-1,0) {$\bullet$};
\node (i2) at (0,-1) {$\bullet$};
\node (i3) at (0,1) {$\bullet$};
\draw [line width = 0.5mm] (0,0) circle (1cm);
\end{tikzpicture},
   \end{align}
   which is what we wanted to show !

%% file: appendix.tex
\section{Some definitions and reminders of random matrix theory}\label{sec:appendix_rmt}

For a complete mathematical introduction to random matrix theory, the reader can refer to \cite{mehta2004random, anderson2010introduction}, while a more practical approach is carried out in \cite{tulino2004random}.
Let us consider a compactly supported probability measure $\mu$ on $\bbR$. We denote $\lambda_{\rm max} \equiv \max \mathrm{supp}(\mu)$ and $\lambda_{\rm min} \equiv \min \mathrm{supp}(\mu)$.
One can introduce the \emph{Stieltjes} transform of $\mu$ as:
\begin{align}\label{eq:def_stieltjes}
\mathcal{S}_\mu (z) \equiv \mathbb{E} \left[\frac{1}{X-z}\right] = \int_{\mathbb{R}} \mu(\mathrm{d}\lambda) \frac{1}{\lambda - z}.
\end{align}
On $(\lambda_{\rm max},+\infty)$, $\mathcal{S}_\mu$ induces a strictly increasing $\mathcal{C}^\infty$ diffeomorphism ${\cal S}_\mu : (\lambda_{\rm max},\infty) \hookrightarrow (-\infty,0)$, and 
we denote its inverse $\mathcal{S}^{-1}_\mu$. 
One can then introduce the \emph{${\cal R}$-transform} of $\mu$ as: 
\begin{align}
\forall z > 0, \quad \mathcal{R}_\mu(z) &\equiv \mathcal{S}_\mu^{-1}(-z) - \frac{1}{z}.
\end{align}
$\mathcal{R}_\mu(z)$ is \emph{a priori} defined for $-z \in \mathcal{S}_\mu\left[(\lambda_{\rm min},\lambda_{\rm max})^c\right]$
and admits an analytical expansion around $z = 0$.
We can write this expansion as:
\begin{align}\label{eq:expansion_R_freecum}
\mathcal{R}_\mu(z) = \sum_{k=0}^\infty c_{k+1}(\mu) \, z^k.
\end{align} 
The elements of the sequence $\{c_k(\mu)\}_{k \in \bbN^\star}$ are called the \emph{free cumulants} of $\mu$.
In particular, one can show that $c_1(\mu) = \EE_\mu(X)$ and $c_2(\mu) = \EE_\mu(X^2) - (\EE_\mu X)^2$.
The free cumulants can be recursively computed from the moments of the measure using the so-called \emph{free cumulant equation}: 
\begin{align}\label{eq:free_cum_formula}
\forall k \in \mathbb{N}^*, \quad \mathbb{E}_\mu X^k = \sum_{m = 1}^k c_m(\mu) \sum_{\substack{\{k_i\}_{i \in [|1,m|]} \\ \text{s.t } \sum_i k_i = k}} \prod_{i=1}^m \mathbb{E}_\mu X^{k_i -1}.
\end{align}
For practical purposes, for all $x \in (-\mathcal{S}_\mu(\lambda_{\rm min}),-\mathcal{S}_\mu(\lambda_{\rm max}))$ we can define:
\begin{align}\label{eq:def_G}
   G_\mu(x) \equiv \frac{1}{2} \int_0^x {\rm d}u \, \mathcal{R}_\mu(u). 
\end{align}

\section{Technical derivations and generalizations of the diagrammatics}\label{sec:generalizations_expansions}

We detail here some extensions of the results of Sec.~\ref{sec:diagrammatics}.
In Sec.~\ref{subsec:complex}, we explain how to transpose these results to Hermitian matrix models, and in
Sec.~\ref{subsec:diverging} we show how to extend some of them to diagrams of diverging size (as $N \to \infty$).

\subsection{Hermitian matrix model}\label{subsec:complex} 
\usetikzlibrary{decorations.markings}

One can generalize the results of Sec.~\ref{sec:diagrammatics} to the following
Hermitian matrix model, similar to Model~\ref{model:sym_rot_inv}:
\begin{model}\label{model:complex}
   Let $N \geq 1$ and $\mathcal{U}(N)$ be the unitary group.
Let $J \in \bbC^{N \times N}$ be a random matrix generated as $J = U D U^\dagger$ with $U \in \mathcal{U}(N)$ drawn uniformly and independently from $D$. 
$D$ is a real diagonal matrix such that its empirical spectral distribution 
   $\rho^{(N)}_D \equiv \frac{1}{N} \sum_{i=1}^N \delta_{d_i}$ converges (almost surely) as $N \to \infty$ a.s. to a probability distribution $\rho_D$ with compact support.
   The smallest and largest eigenvalue of $D$ are assumed to converge almost surely to the infimum and supremum of the support of $\rho_D$.
\end{model}
Note that the diagrams are now directed, as $J_{ij} = \overline{J_{ji}}$. We describe such diagrams in Fig.~\ref{fig:diagrams_hermitian}.
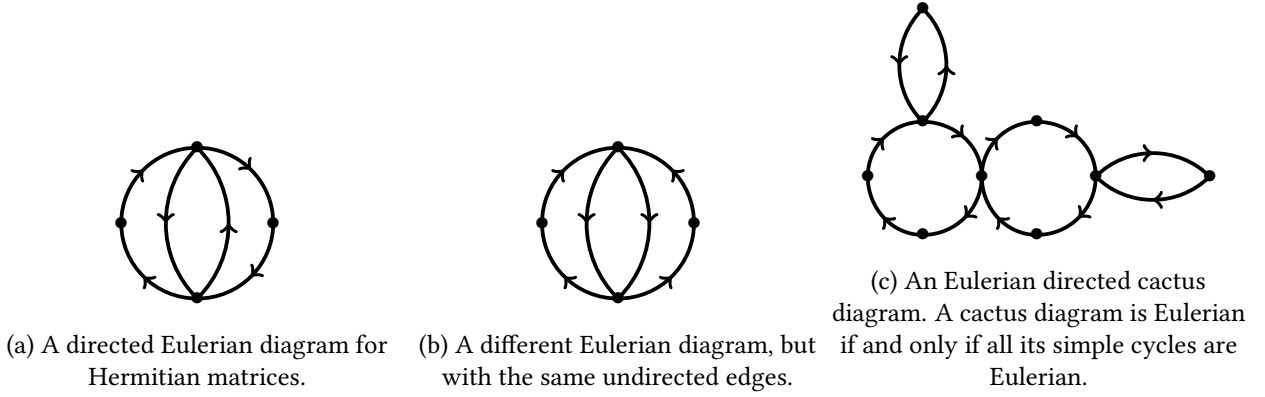
\begin{figure}[t]
 \centering
\captionsetup{justification=centering}
\begin{subfigure}[b]{0.32\textwidth}
   \centering
\begin{tikzpicture}[decoration={markings, mark=at position 0.5 with {\arrow{>}}},scale=1.]
\node (i1) at (0,2) {$\bullet$};
\node (i2) at (1,3) {$\bullet$};
\node (i3) at (2,2) {$\bullet$};
\node (i4) at (1,1) {$\bullet$};
\draw [postaction={decorate},line width = 0.5mm] (i1.center) to [out=90,in=-180] (i2.center);
\draw [postaction={decorate},line width = 0.5mm] (i2.center) to [out=-0,in=90] (i3.center);
\draw [postaction={decorate},line width = 0.5mm] (i3.center) to [out=-90,in=0] (i4.center);
\draw [postaction={decorate},line width = 0.5mm] (i4.center) to [out=180,in=-90] (i1.center);
\draw [postaction={decorate},line width = 0.5mm] (i4.center) to [out=45,in=-45] (i2.center);
\draw [postaction={decorate},line width = 0.5mm] (i2.center) to [in=135,out=-135] (i4.center);
\end{tikzpicture}
\caption{A directed Eulerian diagram for Hermitian matrices.}\label{fig:hermitian_1}
\end{subfigure}
\begin{subfigure}[b]{0.32\textwidth}
   \centering
\begin{tikzpicture}[decoration={markings, mark=at position 0.5 with {\arrow{>}}},scale=1.]
\node (i1) at (0,2) {$\bullet$};
\node (i2) at (1,3) {$\bullet$};
\node (i3) at (2,2) {$\bullet$};
\node (i4) at (1,1) {$\bullet$};
\draw [postaction={decorate},line width = 0.5mm] (i1.center) to [out=90,in=-180] (i2.center);
\draw [postaction={decorate},line width = 0.5mm] (i3.center) to [in=-0,out=90] (i2.center);
\draw [postaction={decorate},line width = 0.5mm] (i4.center) to [in=-90,out=0] (i3.center);
\draw [postaction={decorate},line width = 0.5mm] (i4.center) to [out=180,in=-90] (i1.center);
\draw [postaction={decorate},line width = 0.5mm] (i2.center) to [in=45,out=-45] (i4.center);
\draw [postaction={decorate},line width = 0.5mm] (i2.center) to [in=135,out=-135] (i4.center);
\end{tikzpicture}
\caption{A different Eulerian diagram, but with the same undirected edges.}\label{fig:hermitian_2}
\end{subfigure}
\begin{subfigure}[b]{0.32\textwidth}
   \centering
\begin{tikzpicture}[decoration={markings, mark=at position 0.5 with {\arrow{>}}},scale=1.]
\node (i0) at (0,3) {$\bullet$};
\node (i1) at (1.5,3) {$\bullet$};
\node (i2) at (3,3) {$\bullet$};
\node (i3) at (0.725,3.725) {$\bullet$};
\node (i4) at (0.725,2.225) {$\bullet$};
\node (i5) at (2.225,3.725) {$\bullet$};
\node (i6) at (2.225,2.225) {$\bullet$};
\node (i7) at (4.5,3) {$\bullet$};
\node (i8) at (0.725,5.225) {$\bullet$};
\draw [postaction={decorate},line width = 0.5mm] (i0.center) to [out=90,in=-180] (i3.center);
\draw [postaction={decorate},line width = 0.5mm] (i3.center) to [out=-0,in=90] (i1.center);
\draw [postaction={decorate},line width = 0.5mm] (i1.center) to [out=-90,in=0] (i4.center);
\draw [postaction={decorate},line width = 0.5mm] (i4.center) to [out=180,in=-90] (i0.center);
\draw [postaction={decorate},line width = 0.5mm] (i1.center) to [out=90,in=-180] (i5.center);
\draw [postaction={decorate},line width = 0.5mm] (i5.center) to [out=-0,in=90] (i2.center);
\draw [postaction={decorate},line width = 0.5mm] (i2.center) to [out=-90,in=0] (i6.center);
\draw [postaction={decorate},line width = 0.5mm] (i6.center) to [out=180,in=-90] (i1.center);
\draw [postaction={decorate},line width = 0.5mm] (i2.center) to [out=45,in=135] (i7.center);
\draw [postaction={decorate},line width = 0.5mm] (i7.center) to [in=-45,out=-135] (i2.center);
\draw [postaction={decorate},line width = 0.5mm] (i3.center) to [out=45,in=-45] (i8.center);
\draw [postaction={decorate},line width = 0.5mm] (i8.center) to [in=135,out=-135] (i3.center);
\end{tikzpicture}
\caption{An Eulerian directed cactus diagram. A cactus diagram is Eulerian if and only if all its 
simple cycles are Eulerian.}\label{fig:hermitian_3}
\end{subfigure}
\caption{Diagrams similar to the ones of Fig.~\ref{fig:diagrams_cactus_ncactus}, but for Hermitian matrices. Note that the diagrams
of Fig.~\ref{fig:hermitian_1} and Fig.~\ref{fig:hermitian_2} are different because of the different directions of the edges, but that both are Eulerian.}\label{fig:diagrams_hermitian}
\end{figure}
 For instance, the diagram of Fig.~\ref{fig:hermitian_1} is equal to:
 \begin{align}
   \frac{1}{N} \sum_{\substack{i_1,\cdots,i_4 \\ \text{pairwise distincts}}} J_{i_1 i_2}J_{i_2 i_3}J_{i_3 i_4}J_{i_4 i_1} |J_{i_2 i_4}|^2,
 \end{align}
 while the diagram of Fig.~\ref{fig:hermitian_2} represents the quantity:
 \begin{align}
   \frac{1}{N} \sum_{\substack{i_1,\cdots,i_4 \\ \text{pairwise distincts}}} J_{i_1 i_2} \overline{J_{i_2 i_3}} \overline{J_{i_3 i_4}}J_{i_4 i_1} J_{i_2 i_4}^2.
 \end{align}
 In the complex case, an \emph{Eulerian} graph is similarly defined as a graph in which one can construct 
 a cyclic path (following the directions of the edges) that visits each edge exactly once.
 Note that a \emph{simple cycle} is defined such that the arrows on its edges themselves form a cycle, like the 
 constituent cycles of Fig.~\ref{fig:hermitian_3}.
 We describe the main results we get, using the same kind of techniques as used in Sec.~\ref{subsec:freecum_expectation}:
 \begin{enumerate}
   \item[$(i)$] Only Eulerian diagrams contribute in the $N \to \infty$ limit.
   \item[$(ii)$] Consider a simple cycle ${\cal C}_p$ with $p$ vertices. Then this diagram converges in the $N \to \infty$
   limit to the free cumulant $c_p(\rho_D)$ in $L^2$ norm, as in the real case. More precisely:
   \begin{align}
       \lim_{N\to \infty} \EE\, \left| \frac{1}{N} \sum_{\substack{i_1,\cdots,i_p \\ \text{pairwise distincts}}} \left(U D U^\dagger\right)_{i_1 i_2} \left(U D U^\dagger\right)_{i_2 i_3} \cdots \left(U D U^\dagger\right)_{i_p i_1} -  c_p(\rho_D)\right|^2 &= 0.
   \end{align}
   \item[$(iii)$] Any Eulerian strongly irreducible diagram that is not a simple cycle will be negligible 
   in the $N \to \infty$ limit (in $L^2$ norm).
   \item[$(iv)$] Any Eulerian cactus diagram (like in Fig.~\ref{fig:hermitian_3}) will converge in $L^2$ to the products of the free 
   cumulants of $\rho_D$ corresponding to each one of its constituent simple cycles.
 \end{enumerate}
 These results are straightforward generalizations of the ones obtained for real matrices in Sec.~\ref{sec:diagrammatics}.
  For completeness, we describe how to show a weaker version of $(ii)$, and leave other statements as easy generalizations of Sec.~\ref{sec:diagrammatics}.
 Let us now show that the limit of the expectation of the term in $(ii)$ is the free cumulant, as in Sec.~\ref{subsec:freecum_expectation}.
 As before, by unitary invariance we can assume that $(i_1,\cdots,i_p) = (1,,\cdots,p)$, and we can 
 apply the results of \cite{guionnet2005fourier} to obtain a similar equation to eq.~\eqref{eq:Lp_guionnet}:

 \begin{align}\label{eq:guionnet_complex}
 L_p &\equiv \lim_{N \to \infty} N^{p-1} \EE \, \left[(UDU^\dagger)_{12} \cdots (UDU^\dagger)_{p1} \right], \nonumber \\
 &= \lim_{N \to \infty}  \frac{1}{N} \prod_{l=1}^p \left[\frac{\partial}{\partial b_i} + i \frac{\partial}{\partial c_i}\right]\left[\exp\left\{N\sum_{n=1}^\infty \frac{c_n(\rho_D)}{n} \mathrm{Tr} \, [M(\bb,\bc)^n] \right\}\right]_{\bb,\bc=0},
 \end{align}
 with now the matrix $M(\bb,\bc)$ defined as:
\begin{align}
 M(\bb,\bc) \equiv \frac{1}{2} \begin{pmatrix}
 0 & b_1+i c_1 & 0 & \cdots & 0 & b_p -i c_p\\
 b_1 -i c_1& 0 & b_2 +i c_2& \cdots & 0& 0 \\
 0 & b_2 -i c_2& 0 & \cdots & 0 &0 \\
 \vdots & \vdots & \vdots & \ddots & \vdots & \vdots \\
 0 & 0 & 0 & \cdots & 0 & b_{p-1} +i c_{p-1}\\
 b_p +i c_p& 0 & 0 & \cdots & b_{p-1} -i c_{p-1} & 0
 \end{pmatrix}.
 \end{align}
 Now, we have $\left[\frac{\partial}{\partial b_i} + i \frac{\partial}{\partial c_i}\right] M(\bb,\bc) = F_{i+1,i}$, 
 in which $(F_{a,b})_{ll'} \equiv \delta_{al} \delta_{bl'}$ are elementary non-symmetric matrices.
 In the exact same way as in Sec.~\ref{subsec:freecum_expectation}, the dominant contribution in eq.~\eqref{eq:guionnet_complex}
 will be given by differentiating a single time the exponential term, and creating a cycle with the matrices $F_{i+1,i}$.
 Note that contrary to the symmetric case of Sec.~\ref{subsec:freecum_expectation}, here \emph{only the directed cycle will contribute}, 
 whereas both possible directions of the cycle contributed in eq.~\eqref{eq:cycles_sym}.
 Indeed, the cycles in terms of the matrices $\{F_{a,b}\}$ have to be directed in order to yield a non-zero contribution:
\begin{align*}
  \mathrm{Tr} \, \left[F_{1,2} F_{2,1} F_{1,3}F_{3,2}F_{2,1}\right] &\neq 0 ,\\
  \mathrm{Tr} \, \left[F_{1,2} F_{2,1} F_{2,3}F_{3,2}F_{2,1}\right] &= 0 .
\end{align*}
Thus we have:
 \begin{align*}
   L_p &= \sum_{n=p}^\infty c_n(\rho_D) \mathrm{Tr} \, \left[\left(F_{1,p} F_{p,p-1} \cdots F_{2,1}\right) \, M(0,0)^{n-p}\right], \\
   &= c_p(\rho_D).
 \end{align*}
 In order to get $L^2$ concentration of the simple cycle on the free cumulant, one can exactly repeat the arguments of Sec.~\ref{subsec:concentration_J}.

\subsection{A note on the expectation of diagrams of diverging size}\label{subsec:diverging}

Although it is not directly useful in our Plefka expansions, another side question one can ask on the behavior of these diagrams is: how do diagrams that have a number of edges that
diverge with $N$ behave in the $N \to \infty$ limit ? In all of Sec.~\ref{sec:diagrammatics}
we only considered diagrams of finite size. The behavior of the HCIZ-type integrals with a matrix with diverging rank (as opposed to the finite-rank case)
has been rigorously treated in \cite{guionnet2005fourier} and then generalized in \cite{collins2007new} as soon as the rank of the matrix diverges sub linearly in $N$.
We recall the main result of \cite{collins2007new}:
\begin{theorem}[Collins-\'Sniadyc]\label{thm:collins} Let $A_N,B_N$ be diagonal real matrices of size $N$. Assume that the rank $M(N)$
  of $A_N$ is such that $M(N) = \smallO(N)$, and denote $a_{1,N} \geq \cdots \geq a_{M,N}$ the eigenvalues of $A_N$. Assume that the spectral measure of $B_N$
  converges a.s.\ and in the weak sense to a probability measure $\rho_B$, and that all elements of $A_N$ are bounded by a constant independent of $N$. Then one has:
  \begin{align}
    \frac{1}{N M(N)} \log \int_{\mathcal{U}(N)} \mathcal{D}U \, e^{N \mathrm{Tr}\left[A_N U B_N U^\dagger\right]} &= \frac{2}{M(N)} \mathrm{Tr}\left[G_{\rho_B}\left(A_N\right)\right] + \smallO_N(1). 
  \end{align}
A similar result holds for real orthogonal matrices:
\begin{align}
    \frac{1}{N M(N)} \log \int_{\mathcal{O}(N)} \mathcal{D}O \, e^{\frac{N}{2} \mathrm{Tr}\left[A_N O B_N O^\intercal\right]} &= \frac{1}{M(N)} \mathrm{Tr}\left[G_{\rho_B}\left(A_N\right)\right] + \smallO_N(1). 
\end{align}
\end{theorem}
The techniques of Sec.~\ref{sec:diagrammatics} thus generalize to this case.
We consider real symmetric matrices under Model~\ref{model:sym_rot_inv} (in the Hermitian case, the results also generalize following the line of Appendix~\ref{subsec:complex}). 
We say that a sequence $\{p(N)\}$ satisfies the \emph{bounded free cumulant property} if it satisfies the following:
\begin{property}\label{prop:freecum_bounded}
  There exists $C > 0$ such that for all $N \in \bbN$, $|c_{p(N)}(\rho_D)| < C$. 
\end{property}
We state two of the results of Sec.~\ref{sec:diagrammatics} that can be easily generalized to the diverging size case without changing any of the arguments:
\begin{enumerate}
  \item[$(a)$] Consider a sequence $p(N) = \smallO_N(N)$ that satisfies the bounded free cumulant property. 
   Then one obtains the generalization of eq.~\eqref{eq:conjecture}:
  \begin{align}
  N^{p(N)-1}  \int_{\mathcal{O}(N)} \mathcal{D}O \left[\left(O D O^\intercal\right)_{1 2} \left(O D O^\intercal\right)_{2 3} \cdots \left(O D O^\intercal\right)_{p(N) 1} \right]&= c_{p(N)}(\rho_D) + \smallO_N\left(1\right).
  \end{align}
  \item[$(c)$] Consider a cactus diagram $G$ composed of $P(N)$ simple cycles of size $(r_1(N),\cdots,r_P(N))$, joining at vertices. 
  Assume that $\sum_{i=1}^{P(N)} r_i(N) = \smallO_N(N)$ and that all the sequences $r_i(N)$ satisfy the bounded free cumulant property.
  Then one has:
  \begin{align}
    \EE\, G &= \left[\prod_{i=1}^{P(N)} c_{r_i(N)}(\rho_D)\right] \left(1 + \smallO_N(1)\right).
  \end{align}
\end{enumerate}
Other results obtained in Sec.~\ref{sec:diagrammatics} for finite-size diagrams might also be applicable to the diverging size case, and we leave them for future work.